\documentclass[twoside,leqno]{article}

\usepackage[letterpaper]{geometry}

\pdfmapfile{=libertine.map}

\usepackage[T1]{fontenc}
\usepackage{latexsym,graphicx}
\usepackage{amsmath,amssymb,enumerate, amsthm}
\usepackage{hyperref}
\usepackage{cleveref}
\usepackage{thm-restate}
\usepackage{gensymb}
\usepackage{xspace}
\usepackage{bm}
\usepackage{ifpdf}
\usepackage[dvipsnames]{xcolor}
\usepackage[compact]{titlesec}
\usepackage{titlesec}
\usepackage{algorithm}
\usepackage[noend]{algpseudocode}
\usepackage{subfigure}
\usepackage{verbatim}
\usepackage{mathrsfs}
\usepackage{paralist}
\usepackage{nicefrac}
\usepackage{wrapfig}
\usepackage[shortlabels]{enumitem}
\usepackage{tikz}
\usetikzlibrary{shapes}
\usepackage{setspace}
\usepackage[many]{tcolorbox}
\usepackage{ifthen}
\tcbuselibrary{skins}
\usetikzlibrary{decorations.pathreplacing}
\usepackage{tabularx}
\usepackage{colortbl}
\tikzset{
	mystyle/.style={line width = 1.5pt, color = red!70!black}
}
\usepackage[font={small,it}]{caption}
\usepackage{listings}
\makeatletter
\tcbset{common/.style={
  top = -2mm,
  grow to left by = 4 mm,
  grow to right by = 0 mm,
  colback=gray!5!white,
  colframe=black!50!gray,
 }
}

\Crefname{claim}{Claim}{Claims}

\newtcolorbox{mybox}[1][]{common,#1}

\allowdisplaybreaks

\newcommand{\hoono}[1]{\boxed{#1}}

\newcommand{\OLD}[1]{}


\newtheorem{theorem}{Theorem}[section]
\newtheorem{definition}[theorem]{Definition}
\newtheorem{lemma}[theorem]{Lemma}

\newtheorem{claim}[theorem]{Claim}

\newtheorem{corollary}[theorem]{Corollary}
\numberwithin{algorithm}{section}

\newcommand{\junk}[1]{}
\newcommand{\ignore}[1]{}

\newcommand{\N}[0]{{\ensuremath{\mathbb{N}}}}

\newcommand{\calX}{{\ensuremath{\mathcal{X}}}}
\newcommand{\calC}{{\ensuremath{\mathcal{C}}}}
\newcommand{\calF}{{\ensuremath{\mathcal{F}}}}

\newcommand{\haff}{{\ensuremath{\nicefrac12}}}

\def\ceil#1{\lceil #1 \rceil}

\usepackage{mathtools}

\DeclarePairedDelimiterX{\infdivx}[2]{(}{)}{%
  #1\;\delimsize\|\;#2%
}

\newcommand{\poly}{\operatorname{poly}}
\newcommand{\argmin}{\operatorname{argmin}}

\newcommand{\sse}{\subseteq}

\newcommand{\E}{{\mathbb{E}}}

\newcommand{\e}{\varepsilon}
\newcommand{\eps}{\varepsilon}
\newcommand{\ts}{\textstyle}

\renewcommand{\theequation}{\thesection.\arabic{equation}}

\newcommand{\alert}[1]{{\color{red} #1}}

\newcounter{note}[section]

\newcommand{\initOneLiners}{%
    \setlength{\itemsep}{0pt}
    \setlength{\parsep }{0pt}
    \setlength{\topsep }{0pt}
      \usecounter{myLISTctr}
}
\newenvironment{OneLiners}[1][\ensuremath{\bullet}]
    {\begin{list}
        {#1}
        {\initOneLiners}}
    {\end{list}}

\newcommand{\squishlist}{
 \begin{list}{$\bullet$}
  { \setlength{\itemsep}{0pt}
     \setlength{\parsep}{3pt}
     \setlength{\topsep}{3pt}
     \setlength{\partopsep}{0pt}
     \setlength{\leftmargin}{1.5em}
     \setlength{\labelwidth}{1em}
     \setlength{\labelsep}{0.5em} } }

\newcommand{\squishend}{
  \end{list}  }

\newcounter{asidecounter}

\newcommand{\Opt}{\ensuremath{\mathsf{opt}\xspace}}

\newcommand{\cost}{\ensuremath{\mathsf{cost}\xspace}}

\newcommand{\cA}{\mathcal{A}} 

\newcommand{\cC}{\mathcal{C}} 
\newcommand{\cD}{\mathcal{D}} 
\newcommand{\cF}{\mathcal{F}} 

\newcommand{\cP}{\mathcal{P}} 
\newcommand{\cS}{\mathcal{S}} 
\newcommand{\cT}{\mathcal{T}} 
\newcommand{\cX}{\mathcal{X}}

\newcommand{\kmed}{\mathsf{kmed}}

\newcommand{\cSimple}{\cS}
\newcommand{\xA}{\ensuremath{\mathsf{A}}\xspace}
\newcommand{\xB}{\ensuremath{\mathsf{B}}\xspace}
\newcommand{\xC}{\ensuremath{\mathsf{C}}\xspace}
\newcommand{\xD}{\ensuremath{\mathsf{D}}\xspace}
\newcommand{\xE}{\ensuremath{\mathsf{E}}\xspace}

\newcommand{\xAF}{\ensuremath{\mathsf{AF}}\xspace}
\newcommand{\dSwapHa}{\delta_{\cSimple_1\cap\cA}(c)}
\newcommand{\dSwapHb}{\delta_{\cSimple_2\cap\cA}(c)}

\newcommand{\dSwapTa}{\delta_{\cT_1\cap\cA}(c)}
\newcommand{\dSwapTb}{\delta_{\cT_2\cap\cA}(c)}

\newcommand{\bE}{\ensuremath{\mathbb E}} 
\newcommand{\cE}{\ensuremath{\mathcal E}} 
\usepackage{bbm}
\usepackage{dsfont}
\newcommand{\ind}{\ensuremath{\mathds{1}}} 
\newcommand{\dPQ}[1]{\delta_{(P,Q)}(#1)}
\newcommand{\thd}{\ensuremath{t_{\mathsf{d}}}} 
\newcommand{\thh}{\ensuremath{t_{\mathsf{h}}}} 

\renewcommand{\a}{\alpha}
\renewcommand{\b}{\beta}
\newcommand{\lAngle}{\langle\!\langle}
\newcommand{\rAngle}{\rangle\!\rangle}
\newcommand{\move}[1]{{\lAngle #1 \rAngle}}
\newcommand{\epsd}{\ensuremath{O(\varepsilon)(d^* + d_1)}} 
\newcommand{\WClb}{\ensuremath{-10d_1}} 
\newcommand{\extra}{r(\eps)}
\newcommand{\extras}{r}

\newcommand{\clientT}{\mathcal{ST}} 

\newcommand{\inlineeqnum}{\refstepcounter{equation}~~\mbox{(\theequation)}}

\newcommand{\Gsum}[1]{\Delta_\cA({#1})}

\title{An Improved Local Search Algorithm for $k$-Median}
\author{Vincent Cohen-Addad\thanks{Google Research, Zurich\ and Sorbonne Universit\'e,
     Paris.} \and 
  Anupam Gupta\thanks{Carnegie Mellon University, Pittsburgh PA
    15217.} \and 
  Lunjia Hu\thanks{Stanford University.} \and
  Hoon Oh$^\dagger$ \and
  David Saulpic\thanks{Sorbonne Universit\'e, Paris.}
}
\date{}

\begin{document}

\maketitle

\begin{abstract}
  We present a new local-search algorithm for the $k$-median
  clustering problem. We show that local optima for this algorithm
  give a $(2.836+\epsilon)$-approximation; our result improves upon
  the $(3+\epsilon)$-approximate local-search algorithm of Arya et
  al.~\cite{Arya2001LocalSearch}. Moreover, a computer-aided analysis
  of a natural extension suggests that this approach may lead to an
  improvement over the best-known approximation guarantee for the
  problem. 

  The new ingredient in our algorithm is the use of a potential
  function based on both the closest and second-closest facilities to
  each client. Specifically, the potential is the sum over all
  clients, of the distance of the client to its closest facility, plus
  (a small constant times) the truncated distance to its
  second-closest facility.  We move from one solution to another only
  if the latter can be obtained by swapping a constant number of
  facilities, and has a smaller potential than the former. This
  refined potential allows us to avoid the bad local optima given by
  Arya et al.\ for the local-search algorithm based only on the cost
  of the solution. 
\end{abstract}

\vfill

\thispagestyle{empty}

\pagebreak 

\setcounter{page}{1}

\section{Introduction}
\label{sec:introduction}
The \emph{$k$-median} problem is a classic optimization problem
for metric spaces, and has been widely studied by the algorithm-design
community 
with a two-fold motivation:
on the one hand getting good algorithms for the $k$-median problem immediately
yields important practical implications in operations research, bioinformatics,
or data analysis. On the other hand, 
the study of the approximability of $k$-median has given us
a deeper understanding of key algorithmic ideas like primal-dual techniques and
Lagrangian-multiplier preserving algorithms, sophisticated dependent
LP roundings, local search, iterative rounding, and algorithmic
notions of stability.

Concretely, given a finite
metric space $(\calX, d)$, where the point set $\calX$ is partitioned
into \emph{client} locations $\cC$ and possible \emph{facility}
locations $\cF$,
with $\cX := \cC \cup \cF$, and a parameter $k$, the \emph{$k$-median}
problem asks to pick
$k$ ``medians'' $F \sse \cF$ 
to
minimize
\begin{gather}
  \kmed(F) := \sum_{c \in \calC} d(c,F). \label{eq:1}
\end{gather}
Throughout the paper, given a set $S \sse \calX$, and point
$x \in \calX$ we let $d(x,S)$ denote $\min_{s \in S} d(x,s)$.

An interesting perspective on the $k$-median problem is to 
view it as a ``metric set cover'' problem, where one needs to find $k$
medians (seen as ``sets'') to cover the clients (seen as the universe)
-- with the relaxation
that each client pays a cost that is a function of how well it is covered
and this cost function is a metric.
This perspective has long been known (see e.g.
~\cite{GK98,Jain2002LowerBound}), 
but although the
complexity of the classic set cover problem is well-understood since the
90s, the approximability of this metric variant is still quite open.

The current-best result is the
$2.675$-approximation of Byrka et
al.~\cite{Byrka2015ImprovedApproximation}, improving on a breakthrough
$2.732$-factor of Li and Svensson~\cite{LS16}.
These papers use the clever idea of finding pseudo-approximations
(i.e., solutions with good cost but opening a few extra facilities) by
first giving bi-point solutions (i.e., a feasible fractional solution
that is the convex combination of two integer solutions) using
the primal-dual framework, and then rounding these bi-point solutions
carefully into integer solutions. Nevertheless, the gap between these
results and the
current best
hardness bound of $1+2/e$ remains large.
While various techniques can give good approximations for
$k$-median in specific metrics, the current arsenal for getting a
better approximation bound for the general
case is not very rich. E.g., a significant improvement 
using the bi-point rounding approach seems challenging, since it requires 
either improving the quality of the bi-point solution
computed (on which no progress has been made over the last 20 years), 
or improving on the rounding scheme.
Other techniques to obtain $O(1)$-approximations
are primal-dual, or greedy-plus-pruning, but the best
 bounds using these techniques do not even give a 3-approximation.
Finally, the best result before~\cite{LS16}
was an analysis of the $p$-swap local-search algorithm that tries to
improve the current solution by closing some $p$ facilities and
opening $p$ others.  Arya et al.~\cite{Arya2001LocalSearch} showed
that any local optimum was a $(3+2/p)$-approximation. However, they
also showed instances with a matching ``locality gap'' for this
algorithm (see \S\ref{sec:3-2b-example} for a simple example showing a gap
arbitrarily close to $3$). 
In summary, the only known way to do better than a factor of 3 remains
bi-point
rounding.

In this paper, we draw on parallels with set cover and submodular
optimization problems and propose an extension of the simple
local-search paradigm that  has
the potential to improve the current best-known approximation factor.
While our current
analysis does not improve the best approximation it provides the first
alternative to bi-point solutions to go below a 3-approximation---namely,
to 2.836---and offers the possibility of better results.
The new idea is to perform the local search with respect to some other
``surrogate'' potential $\Phi(F)$ instead of the $k$-median objective
function. This allows us to avoid the bad local minima present in the
standard local search. Of course, this $\Phi$ needs to be easily
computable, and also to be close enough to the original objective function
so that finding a local-optimum with respect to $\Phi$ implies a good
approximation for $k$-median objective as well. Such local-search
procedures are called \emph{non-oblivious} in the literature, and have
been successful in several
settings~\cite{Alimonti94,Khanna1998OnSyntaticversusComputationalViews,Filmus2012MaxCoverageLocalSearch,FilmusW14,GGKMSSV18}.

\subsection{Our Approach and Results}

Let $d_i(c,F)$ be the distance between the client $c$ and the facility
in $F$ that is $i^{th}$-closest to it, so that $d_1(c,F) = d(c,F)$ as
defined above. Define the potential function
\begin{gather}
  \Phi(F) := \sum_{c \in \calC} \Big( \underbrace{d_1(c,F)}_{\text{closest}} + \b \min \big\{
  \underbrace{d_2(c,F) , \alpha\,
  d_1(c,F)}_{\text{truncated second-closest}} \big\} \Big). \label{eq:2}
\end{gather}
For almost all of the paper, we choose $\alpha = 3$ and
$\beta = 1/5$. While
  we motivate the potential in detail in
  \S\ref{sec:our-techniques}, consider two
  clients whose closest facilities are both at distance $D$: one with
  its second-closest facility at the same distance $D$ pays
  $(1 + \beta)D \approx 1.2\,D$, whereas another whose second-closest
  facility is much farther away pays $(1+\a\b)D \approx 1.6\,D$. Hence a
  lower potential prefers solutions with good ``backup'' facilities,
  so that local moves can then explore a richer space. Our main
result is the following:
\begin{theorem}[Pseudo-approximation]
  \label{thm:ub}
  Let $\a = 3, \b = 1/5$, and let $p(\e),\extra$ be sufficiently large constants that depend only on $\e$. If $F$ is a local minimum of our non-oblivious local-search
  procedure with $|F| = k$ facilities and swap size $p(\e)$, then 
  \[ \kmed(F) \leq (2.836+\eps)\cdot \kmed(F^*) \]
  for any solution $F^*$
  with $k - \extra$ facilities. 
\end{theorem}
We can convert this pseudo-approximation into a regular approximation
using ideas from~\cite{LS16,AwasthiBlum2010Stability}. Indeed,
if the original instance is ``stable'' (i.e., if reducing the number
of facilities by $\extra)$ causes the optimal cost to increase by more
than $(1+\e)$), we can get a PTAS~\cite{AwasthiBlum2010Stability} in time $\poly(|\cX|^{\extra})$. Hence, this
reduction of the number of facilities does not change the optimal cost
much, and then the pseudo-approximation of Theorem~\ref{thm:ub} is also a
true approximation. 

We are yet to understand the limitations of this specific potential
function, and of this general approach. The best lower bound for this
potential function we currently know is the following:
\begin{theorem}[Lower Bound for $\Phi$]
  \label{thm:lb}
  There exists $\eps > 0$ and an infinite family of instances on which
  the local-minimum $F$ of our non-oblivious local-search function
  with constant-sized swaps satisfies
  \[ \kmed(F) \geq \min\{\max\{(3-2\beta-\eps, 1+4\beta-\eps)\}, \max\{2,\a-\eps\}\} \cdot\kmed(F^*). \]
  Balancing the two terms gives us a locality gap lower bound of 
  $2\cdot\kmed(F^*)$ for all values of $\a,\b$.
\end{theorem}
This lower bound holds even if $F$ is allowed to have more facilities
than $F^*$. 
The gap between the two
results above suggests that local-search with respect to $\Phi$ still
has the possibility of beating the current-best approximation
bounds.

\paragraph*{Extending our Potential Function.} We consider
extending this non-oblivious approach using more expressive
potentials. E.g., we can look at the $q=3$ closest facilities, as
follows: (we use $d_i$ as shorthand for $d_i(c,F)$, and
$(a \land b) := \min(a,b)$)
\[
\Phi_3(F) = \sum_{c\in \cC} \Big( d_1 + \b_2 \underbrace{(\a_2 d_1 \land d_2)}_{\text{truncated second-closest}} + \b_3
\underbrace{(\a_3 d_1 \land d_3)}_{\text{truncated third-closest}}\Big).
\]
Again $\alpha_i, \beta_i$ are constants, discussed in \S\ref{sec:LP}.
A preliminary implementation of this LP discussed in that section
suggests that we can get an approximation ratio of $2.69$. However,
these are based on experiments, and since we do not have a formal
proof, computer-assisted or otherwise, these should just be considered
circumstantial evidence and promising
first steps. We hope that we (or others) will be able to extend these
to a formal proof.

\subsection{Our Techniques}
\label{sec:our-techniques}
Since the algorithm is just the $p$-swap local search algorithm, 
all the work is in the
analysis of the local optima.

\medskip\textbf{The choice of the objective function.}
Our potential function 
is inspired by the work of Filmus and
Ward~\cite{Filmus2012MaxCoverageLocalSearch,FilmusW14}, who improved
the local-search algorithm for submodular maximization from a
$\nicefrac12$-factor to the optimal $(1-1/e)$-factor. We describe
their idea in the context of max-$k$-coverage: the potential gets a bonus if it covers elements multiple times. I.e.,
for each element, we get a value of $1$ if we cover it once, a small
bonus $\beta_2$ if we cover it at least twice, a smaller additional
bonus $\beta_3$ if we cover it at least thrice, etc.  The total overall bonus is small compared to the gain in covering
it once (so that the potential remains close to the true objective), but
enough to evade the bad local minima. Indeed, if an element is covered twice, the algorithm has more flexibility in
choosing local-search steps, since any single-set swap will leave this
element still covered.

The $k$-median problem is a minimization problem, so the natural
objective is $\sum_c d_1(c) + \sum_{i\geq 2} \beta_i d_i(c)$, where
$d_i(c)$ is the distance from $c$ to its $i^{th}$-closest facility:
this penalty term can incentivize each facility to have ``backup''
facilities close to it. Indeed, just using $d_1 + \beta_2 d_2$ (for
small constant $\b_2 > 0$) side-steps the standard bad examples with
respect to the objective function $d_1$. However, this potential
penalizes us too heavily for not having backups. So if the
instance has $k$ widely-separated clusters, the penalty term
overwhelms the original cost. This suggests the potential~(\ref{eq:2})
we eventually use:
\[ \sum_c d_1(c) \, \bigg[ 1 + \text{(small constant)} \times \min\bigg(1,
  \frac{d_2(c)}{\text{(large constant)} \times \,
    d_1(c)}\bigg) \bigg]. \] However, the introduction of the minimum in the
objective function makes the analysis more involved, since it forces a
case distinction between clients which pay the truncated and
untruncated values.

\newcommand{\Floc}{F}
\newcommand{\Fnew}{F_{\text{new}}}

\medskip\textbf{Important Swaps.} The standard approach to
analyze the quality of local optima for clustering problems is to
define a subset of swaps we call \emph{important}. Since all
swaps are non-improving, these important ones are too. This
non-improvement gives linear inequalities that relate the cost of the
solution $\Fnew$ after the swap to the cost of the local
optimum $\Floc$. To relate $\Fnew$ to the optimal solution
$F^*$, we define important swaps to be ones that replace a constant
number of local facilities $P \sse \Floc$ with the same number of
optimal facilities $Q \sse F^*$. Hence, the cost of $\Fnew$ is the sum
of the costs for (1) ``happy'' clients that are now served optimally
(or even better) in $\Fnew$ because their optimal facility is in $Q$,
(2) the ``sad'' clients which were previously assigned to the
facilities in $P$ that were swapped out, but which are not happy and
hence require \emph{reassignment}, and (3) the remaining
``indifferent'' clients. The art in these proofs is to define the
important swaps to control the reassignment cost for the sad clients.

For example, we can pair each optimal facility with its closest local
facility (assume for now this is a bijection), and form the important
swaps by swapping some constant-sized subset of these pairs. This ensures:
\[\sum_{c \text{ happy}} d(c,F^*) + \sum_{c \text{ sad}} \left(
  d(c,\Floc)+2d(c,F^*)\right) + \sum_{c \text{ indifferent}} d(c,\Floc)  
\ge \cost(\Fnew) \ge \cost(\Floc) = \sum_{c} d(c,\Floc).\]
(see~\cite{Gupta2008SimplerLS} for details). Simplifying gives
\[\sum_{c \text{ happy}} d(c,F^*) + \sum_{c \text{ sad}} 2d(c,F^*) \ge
  \sum_{c \text{ happy}} d(c,\Floc).\] Summing over important swaps
(one per local facility) means each client appears on the left at most
twice (once when happy, and once when sad) and on the right exactly
once, which means $ALG \leq 3 OPT$. 
Handling the non-bijective case loses another $\e$ factor, so the local optimum is at most
$(3+\e)$ times the global optimum.  The important lessons are that
(a)~important swaps need to be ``rich'' enough to infer the
small locality gap, and (b)~``simple'' enough to be able to reason
about.

However, the important swaps used in past
works~\cite{Arya2001LocalSearch,Gupta2008SimplerLS} do not work with
the new potential: \Cref{fig:badcasesingleGupta,fig:badcasesingleArya}
in \Cref{sec:fig-examples} show instances and local solutions that
cost three times the optimum but are not locally optimal with respect
to the new objective function. Yet previously-used important swaps are
not rich/expressive enough to deduce non-local-optimality, and only
prove a $3$-approximation.

\medskip\textbf{New Swaps.}  
Given a local solution $\Floc$, 
we
distinguish the \emph{far} clients $c$ with
$d_2(c, \Floc) \geq \alpha d(c, \Floc)$ from the \emph{close} ones with
$d_2(c,\Floc) < \alpha d(c, \Floc)$.
The type of a client determines which value attains the minimum in the potential function \eqref{eq:2}:
a far client $c$ pays $(1 + \a \b)d(c, F)$ while a close one pays $d(c, F) + \b d_2(c,F)$.
The two types of clients require different analysis.

\emph{Far Clients.} 
Consider a facility $\ell_2$ of $\Floc$ closest to the optimal
facility $f^*$ for far client $c$. If $\ell_2$ is also the local
facility that is closest to $c$, and if we pair it with $f^*$, client
$c$ is a happy client (as described above) and we get a good bound on
the cost of client $c$ (so we should always associate $f^*$ with
$\ell_2$). Else if $\ell_2$ is not a facility that is the closest to
$c$, then a simple argument using the triangle-inequality shows there
exists a second facility in the local solution at distance
$2d(c,f^*) + d(c,\Floc)$ to $c$.  But $c$ is a far client, so this
facility cannot be too close:
$2d(c,f^*) + d(c,\Floc) \geq \alpha d(c,\Floc)$, and so
$d(c,\Floc) \leq \frac{2}{\a -1} d(c,f^*)$, which is an excellent bound.

\emph{Close Clients.} On the other hand, the close clients,
may now be sad both when
their closest facility closes, and also \emph{when their
second-closest closes}. E.g., consider a client whose closest optimal
facility is far from the rest of the instance, but which has two local
facilities at the same distance to it (with $d_1\approx d_2$). 
(See \Cref{fig:singlevsmulti}.)  In
this case, moving from two facilities to one in the local solution
without opening the optimal facility incurs a large reassignment
cost. Hence, such  clients want the swap which opens the optimal
facility to also close both local facilities close
to them. If not, closing any one of these close local facilities would mean
 reassigning them to the other, and suffering a cost of
$(1+\a\b) d_1$. These woud be \emph{very sad} clients. 
So we would like to close both the
facilities for the close clients at the same time.  
Else the potential that was helping the far clients now hurts these close
ones when they become very sad. 

Our approach mitigates the risks: we define \emph{two different swap
  structures} and take a linear combination of the inequalities
obtained from these.  Since the local-search algorithm tries
all possible swaps, the  resulting inequalities remain valid.
The two swaps structures can be viewed as follows. One of them,
referred to as \emph{simple swaps}, is similar to the one described
by~\cite{Gupta2008SimplerLS}, where each facility of $F^*$ is mapped
to its closest facility in $\Floc$. The other one, which resolves the
``bad example'' described in \Cref{fig:badcasesingleGupta} for single
swaps, is to also consider the reverse map: i.e., to map each facility of
$\Floc$ to its closest one in $F^*$. These two maps induce a directed
graph $G$ where the vertices are $F^* \cup \Floc$, with an arc from
$f_1$ to $f_2$ if $f_1$ is mapped to $f_2$ in the appropriate
map. This graph $G$ has outdegree-1 and hence has a nice structure. We
show how to break it into pieces of bounded size; these define 
\emph{tree swaps}. We then work with all the inequalities coming from
these two families of swaps.

A final ingredient is randomization: instead of always mapping each
facility $f$ in one of the solutions to its closest facility $f'$ in
the other solution, we randomize these maps---we  map $f$ to its
second-closest facility in the other solution with some probability
that depends on their relative distances.
This allows us to again
mitigate bad and good scenarios for different types of clients that
are in tension.

In summary, here's what we do: we flip a coin to either consider
simple swaps or tree swaps. In either case, we randomly map some
facilities to the closest or second-closest facilities in the other
solution, and use this to build a set of important swaps. Since all
these are non-improving, this gives us linear inequalities that relate
the local cost to the optimum. Finally, we deduce the approximation
ratio from these linear inequalities.

\begin{figure}
\centering
\includegraphics[scale=0.5]{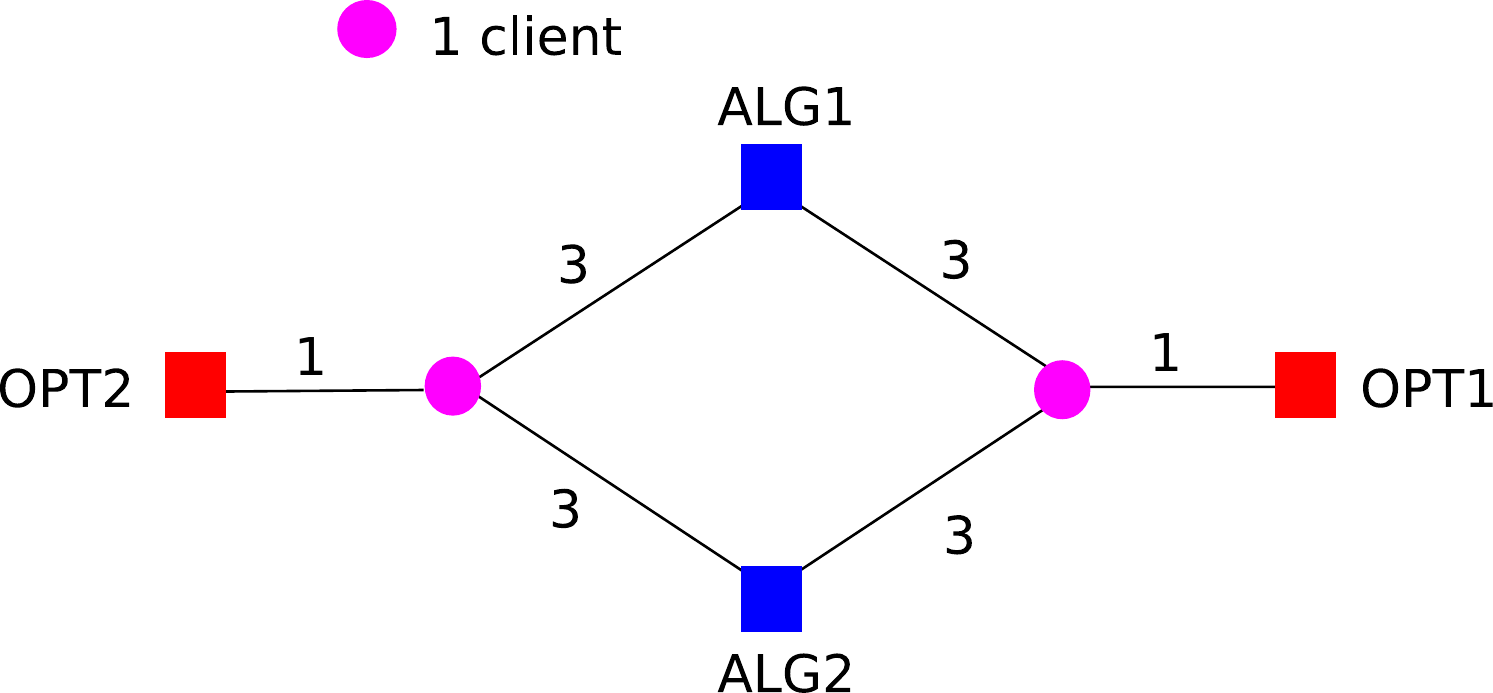}
\caption{Illustration of the tension between clients for defining the swap structure.
	In order to get a good bound for the right client, we need to
        open OPT1 and close both local facilities ALG1 and
        ALG2. However, closing both facilities and opening OPT1
        increases potential value of the left client to $(1+\a\b)7$
        from $(1+\b)3$.
         }          
\label{fig:singlevsmulti}
\end{figure}


\subsection{Related Work}
\label{sec:related-work}

The first $O(1)$-approximation for the $k$-median problem
was given by Charikar et al.~\cite{Charikar1999Constant-Factor}. After
many developments using, e.g., the primal-dual
schema~\cite{JainV01,CharikarG99}, greedy algorithms (and dual fitting)~\cite{JainMMSV03}, improved LP
rounding~\cite{CharikarL12}, local-search~\cite{Arya2001LocalSearch}, and pseudo-approximations~\cite{LS16}, the current best approximation
guarantee is 2.675~\cite{Byrka2015ImprovedApproximation}. The best
hardness result is
$(1+2/e)$~\cite{GK98,Jain2002LowerBound}.
Local-search algorithms have been widely used for clustering
problems. Despite their simplicity, they often give good theoretical
guarantee: the $(3+\eps)$-approximation result
of~\cite{Arya2001LocalSearch} was the best factor for some time; a
simplified proof is given in~\cite{Gupta2008SimplerLS}.
The best results for the closely related $k$-means problem are by
Ahmadian et al.~\cite{AhmadianNSW17}, who give a
$6.35+\eps$-approximation for Euclidean metrics and $9+\eps$ for
general metrics, both using the primal-dual method: these improve on
results of Kanungo et al.~\cite{Kanungo2002LSk-Means} who show that
the simple local-search with respect to the objective function gives a
$(9+\eps)$-approximation for Euclidean $k$-means.

Ahmadian et al.~\cite{AhmadianFS13} give a local-search algorithm for
mobile $k$-median, where they also construct a $1$-tree using the
optimal and algorithm's centers (and the original centers, which play
a role for that problem), and consider swaps based on its
subtrees. However, the details of the analysis seem to be different
from ours, since the concerns in the two problems are quite different.


The use of an alternate potential function instead of the objective
function in local-search was termed non-oblivious
by~\cite{Alimonti94,Khanna1998OnSyntaticversusComputationalViews}.
Filmus and Ward~\cite{Filmus2012MaxCoverageLocalSearch,FilmusW14} used
non-oblivious local-search for the \emph{maximum coverage} and
\emph{submodular maximization} problems, getting
$1-1/e$-approximations in both cases. (A further simplification of the
submodular algorithm/analysis appears in~\cite{FFSW17}.)


\subsection{Paper Outline}

We formally define the algorithm in \S\ref{sec:algorithm}, and the set
of important swaps in \S\ref{sec:analysis}. We classify the clients
into types in \S\ref{subsec:ClientType}, and bound the expected change
in potential for each client type in \S\ref{sec:bounds}; combining them
proves Theorem~\ref{thm:ub}. In \S\ref{sec:LP}, we present how to
construct a linear program that mimics our analysis. In
\Cref{sec:lb}, we prove the lower bound from
Theorem~\ref{thm:lb}.
 Details of calculations, as well as  deferred proofs, appear in the appendix.





\section{The Local Search Algorithm}
\label{sec:algorithm}


The algorithm performs swaps of constant size $p = p(\eps) > \nicefrac{1}{\eps}
  $: given any
solution $F$ (initially arbitrarily chosen) of $k$ facilities from
$\calF$, it tries to find an \emph{improving valid swap}. Here, a swap
$(P,Q) \in \binom{\calF}{\leq p} \times \binom{\calF}{\leq p}$ is
\emph{valid} if $P\subseteq F$, $Q\subseteq \cF\setminus F$, and
$|P| = |Q|$, so that we close as many facilities as we
open. A valid swap is \emph{improving} if
\[ \Phi((F\setminus P)\cup Q) < \Phi(F), \] where $\Phi$ is as defined
in~(\ref{eq:2}). If the algorithm finds an improving valid swap
$(P,Q)$, it sets $F \gets (F\setminus P)\cup Q$, and continues; if
there are no such swaps it returns the local optimum $F$.

This algorithm can be made to run in polynomial time by only
considering swaps that improve the potential by
$(1+\delta n^{-O(p)})$-factor; standard techniques (presented e.g. in Arya et al.~\cite{Arya2001LocalSearch}) show that this
changes the approximation factor by at most $(1+\delta)$, since there are $n^{O(p)}$ many different swaps. Observe that
checking whether we are at a (near)-local optimum, or finding an
improving valid swap can be done in $n^{O(p)}$ time. In the rest of
the paper we show the pseudo-approximation claimed in
Theorem~\ref{thm:ub}, i.e., the cost of a local optimum is comparable
to the cost of any solution $F^*$ with $k - \extra$
facilities, where $\extra$ is the number of extra local facilities.

Throughout the paper, we choose the swap size $p(\varepsilon)$ to be $M(\lceil \nicefrac 1\varepsilon\rceil + 1)^{4\lceil \nicefrac 1\varepsilon\rceil^{\lceil \nicefrac 1\varepsilon\rceil}}$, and choose the number of extra local facilities to be $\extra = M(\lceil \nicefrac 1\varepsilon\rceil + 1)^{1 + 16\lceil \nicefrac 1\varepsilon\rceil^{\lceil \nicefrac 1\varepsilon\rceil}}$ for a sufficiently large absolute constant $M$.

\subsection{Proof Strategy}
\label{sec:proof-strategy}

Let us fix some notation: fix a local optimum $F$ of size $k$ and a global optimum
$F^*$ of size $k - \extra$; we call the former the \emph{local} and the latter the 
\emph{optimal} facilities. For a client $c$, let
\begin{OneLiners}
\item $d^*(c) := d(c,F^*)$ be its cost and $f^*$ its closest facility
  in the optimal solution $F^*$,
\item $d_1(c)$ and $d_2(c)$ be its distances to the closest
  and second-closest facilities, and $f_1$ and $f_2$ be these facilities
  in $F$, and
\item $\Phi^c := d_1(c) + \b \min(d_2(c), \a d_1(c))$ be client $c$'s
  contribution to the potential. From now on, we fix  $\a = 3$ and
  $\b = 1/5$.
\end{OneLiners}

Our proof of Theorem~\ref{thm:ub} is based on the fact that 
at the local optimum $F$, the potential change induced by a valid swap $(P,Q)$ is non-negative, i.e., $\Phi ((F\setminus P)\cup Q) - \Phi (F)\geq 0$. 
Defining the \emph{potential change} of client $c$ on swap $(P,Q)$ to be
\begin{gather}
  \dPQ{c}:= \Phi^c ((F\setminus P)\cup Q) - \Phi^c (F), \label{eq:Delta}
\end{gather}
we have
\begin{equation*}
0 \leq \sum_{c\in\cC}\dPQ{c}.
\end{equation*}
This inequality holds for all valid swaps $(P,Q)$; it remains true even if we extend the definition of valid swaps to allow $Q$ to intersect $F$ and/or to have a size smaller than $P$, because doing so never decreases the potential change. We can thus take linear combinations of the inequality over all valid swaps $(P,Q)$. In particular, for any random set $\cP$ of valid swaps, 
\begin{equation*}
0 \leq \bE_\cP\Big[\sum_{(P,Q)\in \cP}\sum_{c\in \cC}\dPQ{c}\Big] =  \sum_{c\in \cC}\E_\cP\Big[\sum_{(P,Q)\in \cP}\dPQ{c}\Big].
\end{equation*}
Theorem~\ref{thm:ub} is thus implied by the following lemma (and
observing that $\frac{2.5203}{0.8888} \leq 2.836$):
\begin{lemma}
\label{lm:main}
There is a distribution over sets $\cP$ of valid swaps such that for all clients $c\in\cC$, 
\begin{equation*}
\bE\Big[\sum_{(P,Q)\in\cP}\dPQ{c}\Big]\leq 2.5203\,d^*(c) - 0.8888\,d_1(c) + O(\varepsilon)\,(d^*(c) + d_1(c)).
\end{equation*}
\end{lemma}
In order to prove this lemma, we build a randomized procedure
generating the set $\cP$ of swaps (where we call elements of $\cP$ \emph{important
  swaps}), and divide our analysis into two cases: the \emph{amenable}
case and the \emph{defiant} case. In particular, given a client $c$,
we define a suitable amenable event $\cA$ and its complement
defiant event $\cD$, and show the following two lemmas, which immediately
imply \Cref{lm:main}.
\begin{restatable}[Defiant Case]{lemma}{DefiantSwaps}
  \label{lem:combo}
  There is a distribution over sets $\cP$ of valid swaps such that for all clients $c\in\cC$, 
  \begin{align}
    \bE\Big[\ind_\cD\sum_{(P,Q)\in\cP}\dPQ{c}\Big] &\leq 
                                                     O(\varepsilon)\,
                                                     (d^*(c) +
                                                     d_1(c)). \label{eq:strategy1}
  \end{align}
\end{restatable}

\begin{restatable}[Amenable Case]{lemma}{AmenableSwaps}
  \label{lem:combo2}
  For the distribution over valid swap sets from
  \Cref{lem:combo}, for any $c\in\cC$, 
  \begin{align}
    \bE\Big[\ind_\cA\sum_{(P,Q)\in\cP}\dPQ{c}\Big] &\leq
                                                     2.5203\,d^*(c) - 0.8888\,d_1(c) + O(\varepsilon)\,(d^*(c) + d_1(c)).\label{eq:strategy2}
  \end{align}
\end{restatable}
In \S\ref{sec:analysis}, we
define the distribution over sets $\cP$ of important swaps. In \S\ref{subsec:types} we classify the clients into types. We define the
amenable and defiant events for clients of each type and prove~\Cref{lem:combo} in \S\ref{sec:amenable-defiant-swaps}, and then 
prove \Cref{lem:combo2} in \S\ref{sec:bounds}.

\newcommand{\Psimple}{\cP^{\text{simple}}}
\newcommand{\Ptree}{\cP^{\text{tree}}}
\newcommand{\Eta}{\tau}
\newcommand{\bbT}{\alert{\pmb{\tau}}}
\newcommand{\heavy}{heavy\xspace}

\section{Generating the Important Swaps}
\label{sec:analysis}
In this section, we describe our randomized procedure generating
$\cP$, the set of important swaps, that proves
\Cref{lem:combo,lem:combo2}. 
$\cP$ contains valid swaps $(P,Q)$, where
$P\subseteq F$ has size at most $p(\eps)$, and $Q$ is an arbitrary set of
facilities with size at most $|P|$. Every swap we generate has $Q$
being a subset of $F^*$, the set of optimal facilities. We say swap
$(P,Q)$ \emph{closes} the local facilities in $P$, and \emph{opens}
the optimal facilities in $Q$. (By duplicating points in the metric
space, we assume $F$ and $F^*$ are disjoint, and so are $P, Q$.)
Sometimes we say the swap \emph{contains} the local facilities in $P$
and the optimal facilities in $Q$. 

In order to prove \Cref{lem:combo,lem:combo2}, we want to minimize the
potential change of every client by always opening a ``nearby'' optimal
facility whenever we close a local facility. Roughly, we generate both
simple and tree swap sets by constructing a directed graph $G$ over
the vertex set $F\cup F^*$, where every edge connects ``nearby'' local
and optimal facilities. We perform some surgery on this graph if
needed: we remove vertices in $F$, duplicate vertices in both $F$ and
$F^*$, and remove some edges, so that every connected
component of the resulting graph has a small size. Finally, we combine
these connected components of $G$ into small-sized groups so that the
number of local facilities in each group is no smaller than the
optimal facilities in it. The swap set $\cP$ consists of the swap
defined by each of these groups, closing/opening all the local/optimal
facilities in it. In the following subsections, we describe in detail
our procedures generating the simple and tree swap sets. (Again,
recall this is all in the analysis, since the algorithm is just the
$p$-swap local search that attempts to improve the potential.)

\subsection{Generating the Important Simple Swaps}\label{subsec:simpleswap}
\label{sec:simple}

We start by constructing a random directed graph $G_0$ over vertices
$F\cup F^*$. The graph is defined 
by a random function $\tau:F^*\rightarrow F$ that maps each optimal facility to
a local facility: this gives a bipartite graph with $F^*$ vertices
have out-degree one, and $F$ vertices having no out-degree.
In previous
  analyses, $\tau(f^*)$ was defined as the closest local facility to
  $f^*$, but in our analysis, we choose $\tau(f^*)$ randomly from the
  two closest local facilities to $f^*$ in order to cover a larger neighborhood with good balance.
Indeed, independently for every optimal
facility $f^*$, we choose $\tau(f^*)$ from $\eta_1$ and $\eta_2$, where
$\eta_1 = \eta_1(f^*)$ and $\eta_2 = \eta_2(f^*)\in F$ are the first
and second closest local facilities to $f^*$. The probability of
choosing $\eta_i$ depends on the value of
$\rho = \rho(f^*):=\frac{d(f^*,\eta_1)}{d(f^*,\eta_2)}\in[0,1]$. When
$\rho(f^*)\leq\nicefrac 34$, we choose $\tau(f^*) = \eta_1$ with
probability 1; when $\rho(f^*) > \nicefrac 34$, we choose
$\tau(f^*) = \eta_1$ with probability $(\nicefrac 52 - 2\rho)$ and
$\tau(f^*) = \eta_2$ with the remaining probability
$(2\rho - \nicefrac 32)$. 

Intuitively, $\tau(f^*)$ is the facility used as
a fallback to serve clients of $f^*$'s cluster when their closest local
facility is swapped out. More precisely, we design the swaps such that
either $f^*$ or $\tau(f^*)$ is open. To bound the reassignment cost to $\tau(f^*)$,
we therefore must ensure that $\tau(f^*)$ is as close as possible to $f^*$. 
When  $\rho(f^*)$ is small, there is therefore a huge incentive 
in choosing $\tau(f^*) = \eta_1$. However, 
when  $\rho(f^*)$ is close to $1$, there is no difference between $\eta_1$ or $\eta_2$. 
Our probability distribution is chosen such as to implement that intuition.
It has been tuned 
experimentally: using our LP formulation, we were able to look for a choice of
of $\tau$ that gives a good approximation guarantee while being simple enough to prove that guarantee.

This defines the graph $G_0$. 
We wish to generate swaps according to the connected components of
$G_0$, i.e., every swap closes all the local facilities in a connected
component and opens all the optimal facilities in the same connected
component. However, such swaps may not be valid because 1) the size of
a connected component may be much larger than $p$, and 2) there may be
more optimal facilities in a connected component than local facilities
(since every connected component of $G_0$ contains exactly one local facility). We solve these issues by two procedures: \emph{degree reduction} and \emph{balancing}.
%
%
%
\paragraph{Degree reduction.} 
The size of a connected component of $G_0$ being too large is
 caused by local facilities with high in-degree. We solve the problem by removing all local
facilities that could potentially have high in-degree from the
graph. We call these the \emph{\heavy} local facilities. To keep the number of local facilities in the graph unchanged, we duplicate other local facilities, which we call \emph{local surrogates}. We formally define \heavy local facilities and local surrogates as follows.
We first define $N(f^*)\subseteq \{\eta_1,\eta_2\}$ and call it the set of local
neighbors of $f^*$. 
If $\rho(f^*)\leq \nicefrac 23$, we define $N(f^*) = \{\eta_1\}$;
otherwise, we define $N(f^*) = \{\eta_1, \eta_2\}$. We choose $\thd = \lceil\nicefrac 1\varepsilon\rceil$ as the \emph{degree threshold}. Now the \heavy
local facilities are as follows:
\begin{definition}[\heavy local facility]
A local facility $f \in F$ is \heavy if it is a local neighbor of more than $\thd+1$ optimal facilities.
\end{definition}
Note that
$\tau(f^*)$ must be a local neighbor of $f^*$
because $\nicefrac 34 > \nicefrac 23$.
Therefore,
only \heavy local facilities can have in-degree more than $\thd+1$ in $G_0$. For every \heavy local facility, 
we choose a local surrogate 
uniformly at random from the \emph{local candidates} defined as follows:
\begin{definition}[local candidate]
A local facility $f\in F$ is a \emph{local candidate} if it is not \heavy and every optimal facility in $\tau^{-1}(f)$ has a \heavy local neighbor.
\end{definition}
Note that, unlike our definition of \heavy local facilities, the
definition of local candidates depends on the random function
$\tau$. The following claim (proved in
\Cref{sec:proof-local-candidate}) shows that there are enough local candidates from which the \heavy local facilities can choose:
\begin{restatable}{claim}{lightclients}
  \label{clm:light-simple}
  The number of local candidates is at least $\nicefrac \thd2$ times the number of \heavy local facilities.
\end{restatable}
We are ready to describe our degree reduction procedure:
\begin{OneLiners}
\item[1.] Remove all the edges incident to \heavy local facilities;
\item[2.] Replace each \heavy local facility $f$ by its local
  surrogate $s$, chosen uniformly at random \emph{without replacement}
  from the local candidates. Hence, in the graph the vertex labeled
  $f$ (and now having no in-edges due to step 1) is replaced by one
  labeled $s$. So a local surrogate appears twice now: the original
  copy of $s$, and a single isolated vertex as a surrogate for
  $f$. 
\end{OneLiners}
Let $G_1$ denote the graph after degree reduction. Clearly, every
local facility has degree at most $\thd+1$ in $G_1$, and thus every
connected component has size at most $\thd+2$. The next claim follows directly from \Cref{clm:light-simple}:

\begin{claim}
\label{clm:simple-degree-reduction}
The constructed graph $G_1$ satisfies following properties:
\begin{OneLiners}
\item[i.] Heavy local facilities do not appear in $G_1$. 
\item[ii.] Local facilities chosen as local surrogates appear twice: once as the original copy and once as an isolated vertex. 
\item[iii.] Other local facilities and all optimal facilities appear once. 
\item[vi.] Every optimal facility $f^*$ points to the original copy of $\tau(f^*)$ unless $\tau(f^*)$ is heavy.
\item[v.] Any local facility is chosen as a local surrogate with
  probability at most $\nicefrac 2\thd$, and only when it is a local candidate.
\end{OneLiners}
\end{claim}
%
%
%
%
%
%

\paragraph{Balancing.} Since a connected component of $G_1$ may contain
more optimal facilities than local ones we combine connected
components together to form groups with at least as many local
facilities as optimal ones, using the following claim
(proved in \Cref{sec:proof-balancedGroup}):

\begin{restatable}[Balancing Procedure]{claim}{balancedGroup}
  \label{claim:balancedGroup}
  Consider a universe $U = R \cup G$ of red points $R$ and green
  points $G$, with $|G| = |R|+\extras$. Let the collection of sets
  $S_1,\ldots,S_N$ partition $U$, and let $|S_i| \leq x$ for all
  $i$. Moreover, let $H$ be a graph on the vertices $[N]$ with maximum
  degree at most $\theta \leq \extras$.
  Lastly, $\extras \geq \Omega\big(\frac{x^5\theta^3}{\eps}\big)$ for some $0\leq \eps\leq 1$.  Then we can merge
  these sets together into new sets $T_0,\ldots, T_M$ such that
  \begin{OneLiners}
  \item[(i)] each $T_j$ has size $|T_j| \leq O(x^2)$,
  \item[(ii)]
    $|T_j \cap R| \le |T_j \cap G|$,
  \item[(iii)] if there is an edge $\{i,j\}$ for $i, j \in [N]$, then
    $S_i$ is not merged with $S_j$, and 
  \item[(iv)] for all $i \neq j$, $S_i$ is merged with $S_j$ with
    probability at most $\eps$.
  \end{OneLiners}  
\end{restatable}

Recall that our degree reduction step did not change the total number of local and optimal facilities, so there are still $\extra$ more local facilities than optimal facilities. 
We identify $F^*,F$ with $R,G$ in \Cref{claim:balancedGroup}
respectively,  and define every $S_i$ as the set of facilities in
every connected component of $G_1$. Note that $|S_i| \leq \thd +
2$. $S_i$ and $S_j$ are connected by an edge in $H$ 
if and only if they contain
two copies of the same local facility: one contains the original copy
of a local facility and the other contains a new copy created as a
local surrogate. The maximum degree of $H$ is at most $1$ due to
\Cref{clm:simple-degree-reduction} and the fact that there is at most
one local facility in each connected component. Since $\extra \geq
\Omega((\thd + 2)^5/\varepsilon)$, we use \Cref{claim:balancedGroup} to combine components of $G_1$ into balanced groups, where every group contains at most $O((\thd + 2)^2) \leq p(\eps)$ facilities.
Every group thus defines a valid swap, and we define $\cP$ as the set of these swaps.
\Cref{fig:simpleGroup} shows an example of the simple swap set $\cP$ we generate.

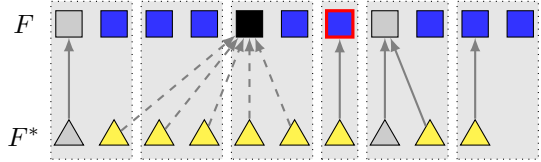
\begin{figure}[H]
  \center
\begin{tikzpicture}[scale=0.6,
		opt/.style={shape=regular polygon,regular polygon sides = 3,draw=black,minimum size=0.45cm,inner sep = 0pt},
		local/.style={shape=rectangle, draw=black,minimum size= 0.35cm},
		swap/.style={fill=gray!20!white,dotted}
]
\pgfset{
  foreach/parallel foreach/.style args={#1in#2via#3}{evaluate=#3 as #1 using {{#2}[#3-1]}},
}
	\def \n {4}
	
	\draw[swap] (0.6,-0.5) rectangle (2.4,3);
	\draw[swap] (2.6,-0.5) rectangle (4.4,3);	
	\draw[swap] (4.6,-0.5) rectangle (6.4,3);	
	\draw[swap] (6.6,-0.5) rectangle (7.4,3);
	\draw[swap] (7.6,-0.5) rectangle (9.4,3);	
	\draw[swap] (9.6,-0.5) rectangle (11.4,3);		

	\def \cc {black!20}
	\def \candc {blue!80!white} 
	\def \srgtc {black}   
	\def \srgtoc {yellow!80} 
	
	\def \colorlistf{"\cc","\candc","\candc","\candc","\srgtc","\candc","\candc","\cc","\candc","\candc","\candc"}	
	\def \colorlisto{"\cc","\srgtoc","\srgtoc","\srgtoc","\srgtoc","\srgtoc","\srgtoc","\cc","\srgtoc","\srgtoc"}

	\foreach \i [parallel foreach=\c in \colorlisto via \i] in {1,...,10}{
		\node [opt,fill=\c] (o\i) at (\i,0) {};
	}
	\foreach \i [parallel foreach=\c in \colorlistf via \i] in {1,...,6,8,9,10,11}{
		\node [local,fill=\c] (f\i) at (\i,2.5) {};
	}
	\node [local,fill=\candc,draw = red,very thick] (f7) at (7,2.5) {};

	
	
	\draw (o1) edge[->,>=latex,gray,thick] (f1);
	\draw (o2) edge[->,>=latex,gray,thick,dashed] (f5);
	\draw (o3) edge[->,>=latex,gray,thick,dashed] (f5);
	\draw (o4) edge[->,>=latex,gray,thick,dashed] (f5);
	\draw (o5) edge[->,>=latex,gray,thick,dashed] (f5);
	\draw (o6) edge[->,>=latex,gray,thick,dashed] (f5);
	\draw (o7) edge[->,>=latex,gray,thick] (f7);
	\draw (o8) edge[->,>=latex,gray,thick] (f8);
	\draw (o9) edge[->,>=latex,gray,thick] (f8);
	\draw (o10) edge[->,>=latex,gray,thick] (f10);


	
	


	\node at (0,0) {$F^*$};
	\node at (0,2.5) {$F$};	

\end{tikzpicture}

\caption{An example of a simple swap set $\cP$. Edges correspond to $\tau(f^*)$'s. Dashed edges are removed. The black facility is a local surrogate replacing a heavy local facility. The original copy of the black surrogate is the facility with red boundary, chosen randomly from the local candidates (blue), assuming every yellow optimal facility has a heavy local neighbor. Gray boxes correspond to the swaps in $\cP$.}
\label{fig:simpleGroup}
\end{figure}

\subsection{Generating the Important Tree Swaps}\label{subsec:treeswap}
\label{sec:tree}

Again, we start by constructing a directed graph $G_0$. Unlike simple
swaps where only optimal facilities have out-edges, tree swaps require
every local facility to also have an out-edge to an optimal facility in $G_0$. In particular, every local facility $f$ has an out-edge to $\pi(f)$, the optimal facility closest to it. Every optimal facility still has an out-edge to $\tau(f^*)\in\{\eta_1,\eta_2\}$, but we pick $\tau(f^*)$ from a different distribution: if
$\rho(f^*) \leq \nicefrac 23$, then $\tau(f^*) = \eta_1$ with probability 1;
else $\tau(f^*) = \eta_1$ with probability $\nicefrac 12$ and
$\tau(f^*) = \eta_2$ otherwise. 

\begin{figure}[H]
  \center
  \begin{tikzpicture}[scale=0.8,
		opt/.style={shape=regular polygon,regular polygon sides = 3,draw=black,minimum size=0.5cm,inner sep = 0pt},
		sopt/.style={shape=regular polygon,regular polygon sides = 3,draw=black,minimum size=0.4cm,inner sep = 0pt},
		local/.style={shape=rectangle, draw=black,minimum size= 0.7cm},
		main/.style={shape=rectangle, draw=black,minimum size= 0.3cm},
]
\pgfset{
  foreach/parallel foreach/.style args={#1in#2via#3}{evaluate=#3 as #1 using {{#2}[#3-1]}},
}
	\def \n {8}
	\pgfmathsetmacro{\nm}{\n-1}
	\def \margin {10}
	\def \radius {2cm}

	\def \RED {red!30}
	\def \BLUE {blue!30}
	\def \GREEN { green!30}
	
	\def \colorlist{"\BLUE","\BLUE","\RED","\RED","\RED","\GREEN","\GREEN","\GREEN"}
	\def \namelist{"f^*_8","f_7","f^*_6","f_5","f^*_4","f_3","f^*_2","f_1"}
	\def \edgelist{"solid","solid","solid","solid","solid","solid","solid","solid"}
	\foreach \s [parallel foreach=\isDash in \edgelist via \s]  in {1,...,\n}
        {    
            \draw[->, >=latex,\isDash] ({360/\n * (\s - 1)+\margin}:\radius) 
                arc ({360/\n * (\s - 1)+\margin}:{360/\n * (\s)-\margin}:\radius);
        }
        
	\foreach \s [count=\x,
			parallel foreach=\c in \colorlist via \x,
			parallel foreach=\name in \namelist via \x]
			in {sopt,main,sopt,main,sopt,main,sopt,main}{
		\node[\s] (c\x) at ({360/\n*(\x-1)}:\radius) {};	
	}

	\def \c {\RED}
	\node[main] (rf1) at (-3,0.75) {};
	\node[main] (rf2) at (-3,0.25) {};
	\node[main] (rf3) at (-3,-0.25) {};
	\def \c {red!10}	
	\node[main] (rf4) at (-3,-0.75) {};
	
	\foreach \x in {1,...,4}
		\draw (rf\x) edge[->,>=latex] (c5);

	\def \c {\RED}
	\node[sopt] (luf1) at (-2.6,1.8) {};
	\node[sopt] (luf2) at (-2.2,2.2) {};	
	\node[sopt] (luf3) at (-1.8,2.6) {};	
	
	\foreach \x in {1,...,3}
		\draw (luf\x) edge[->,>=latex] (c4);

	\def \c {\BLUE}
	\node[main] (lf1) at (3,0.5) {};
	\node[main] (lf2) at (3,0) {};	
	\node[main] (lf3) at (3,-0.5) {};
	
	\foreach \x in {1,...,3}
		\draw (lf\x) edge[->,>=latex] (c1);

	\def \c {violet!30}
	
	\node[sopt](lf20) at (4,-0.5) {};
	\draw (lf20) edge[->,>=latex] (lf3);
	
	\foreach \s [count=\z] in {main,sopt}{
		\foreach \x in {1,...,3}{
			\pgfmathsetmacro{\y}{-1.5+0.5*\x}
			\pgfmathsetmacro{\h}{4+\z}
			\pgfmathtruncatemacro{\zm}{\z-1}				
			\node[\s] (lf\x\z) at (\h,\y) {};
			\draw (lf\x\z) edge[->,>=latex] (lf2\zm);
		}	
	}

	\node[sopt] (lfx) at (4,0) {};
	\node[main] (lfxx) at (5,1) {};
	\draw (lfxx) edge [->,>=latex] (lfx) {};
	\node[sopt] (lfy) at (4,-1) {};
	\node[main] (lfyy) at (5,-2) {};
	\draw (lfyy) edge [->,>=latex] (lfy) {};	
	\draw (lfx) edge[->,>=latex] (lf3) {};
	\draw (lfy) edge[->,>=latex] (lf3) {};


	\def \c {\GREEN}
	\node[main] (cloc) at (0,-3) {};
	\draw (cloc) edge[->,>=latex] (c7);

	\def \cLBL {lime!30}
	\def \cLBR {olive!70}
	
	\def \colorlist{"\cLBL","lime!90","lime!90!black","olive!30","\cLBR"}
	\foreach \x [parallel foreach=\c in \colorlist via \x] in {1,...,5}{
		\pgfmathsetmacro{\y}{-1.5+0.5*\x}
		\node[sopt] (f\x) at (\y,-3.75) {};
		\draw (f\x) edge[->,>=latex] (cloc);
	}

	\foreach \x in {1,...,3}{
		\pgfmathsetmacro{\y}{-2+0.5*\x}
		\node[main] (f1\x) at (\y,-4.5) {};
		\draw (f1\x) edge[->,>=latex] (f1);
	}
	
	\foreach \x in {1,...,3}{
		\pgfmathsetmacro{\y}{0+0.5*\x}
		\node[main] (f2\x) at (\y,-4.5) {};
		\draw (f2\x) edge[->,>=latex] (f5);
	}
	
	
\end{tikzpicture}

\caption{Every connected component of $G_0$ is a 1-tree. Local facilities are represented by squares, while optimal facilities are represented by triangles.}\label{fig:1tree}
\end{figure}
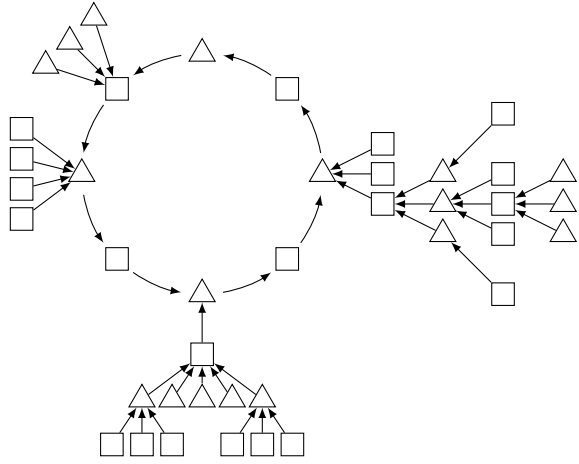

Since every vertex of $G_0$ has out-degree one, $G_0$ is a
\emph{1-forest}, with every connected component being a \emph{1-tree},
i.e., a directed tree with a directed cycle as its root (see
\Cref{fig:1tree}), hence the name \emph{tree swaps}. Having
constructed $G_0$, we generate the tree swap set $\cP$ by three
procedures: degree reduction, edge deletion, and balancing. The
balancing step remains essentially the same as in simple swaps, but
the degree reduction step requires a new ingredient to deal with
optimal facilities with high in-degree, which did not exist in the
simple swaps case. The edge deletion step is also unique to tree
swaps. Next, we describe these three steps in detail.

\paragraph{Degree reduction.} We first modify $G_0$ so that every vertex has in-degree bounded by $\thd+1$. In the same way as simple swaps, we can remove local facilities with high in-degree by removing \heavy local facilities, but we need an additional procedure to deal with \emph{heavy optimal facilities} with high in-degree. Specifically, we say $f^*$ is a heavy optimal facility if it has in-degree more than $\thd$ after \heavy local facilities are removed, in other words, $|\pi^{-1}(f^*)\backslash\{\textup{\heavy local facilities}\}|>\thd$. For such a heavy optimal facility $f^*$ with in-degree $s$, we partition its children into $\lceil s/\thd \rceil$ groups. Every group, except sometimes the last one, contains exactly $\thd$ children. We make sure that the first group contains the $\thd$ closest children to $f^*$. 
For each group other than the first one, we create a new copy of $f^*$ and change the out-edges from the children in the group to point to the new copy of $f^*$. The new copy of $f^*$ has an out-edge pointing to a new copy of a local facility $f$ chosen uniformly at random from the previous group. We call the new copy of $f$ an \emph{optimal surrogate}. They are needed to keep the difference between the number of local and optimal facilities unchanged. We also add an out-edge from $f$ pointing back to the new copy of $f^*$, as illustrated in \Cref{fig:heavyOptDecomp}.

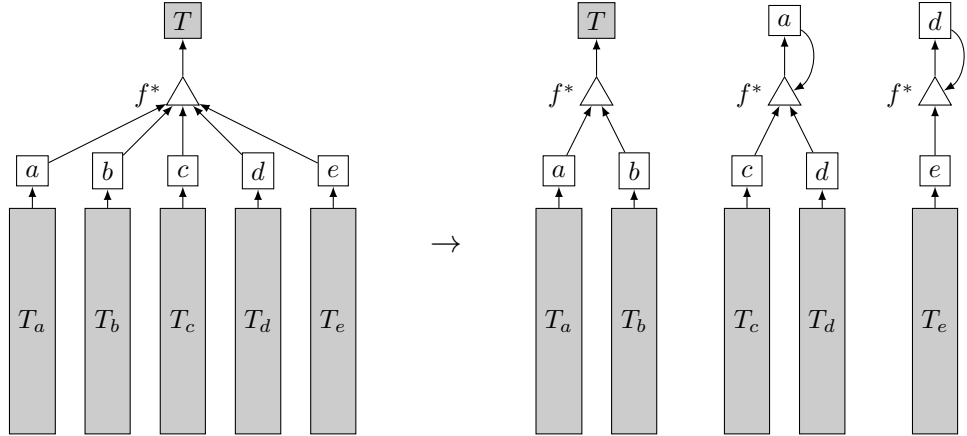
\begin{figure}[tbp]
\center
\begin{tikzpicture}[scale=1,
tree/.style = {shape=rectangle, draw = black, minimum height = 3cm, minimum width = 0.6cm,fill=black!20!white},
		opt/.style={shape=regular polygon,regular polygon sides = 3,draw=black,minimum size=0.5cm,inner sep = 0pt},
		local/.style={shape=rectangle, draw=black,minimum size= 0.4cm},
		client/.style={shape=circle, draw=black,minimum size=0.25cm,inner sep = 0pt,fill=black}]
]
    	
    \node[local,fill=black!20!white]
         (f) at (0,2) {$T$};
         
    \node[opt,label=west:$f^*$] 
    	 (f*) at (0,1) {};
         
    \node[local] 
         (a) at (-2,0) {$a$};
    \node[local] 
         (b) at (-1,0) {$b$};
    \node[local] 
         (c) at (0,0) {$c$};
    \node[local] 
         (d) at (1,0) {$d$};
    \node[local] 
         (e) at (2,0) {$e$};

    \node[tree]  (Ta) at (-2,-2) {$T_a$};
    \node[tree]  (Tb) at (-1,-2) {$T_b$};
    \node[tree]  (Tc) at (0,-2) {$T_c$};
    \node[tree]  (Td) at (1,-2) {$T_d$};
    \node[tree]  (Te) at (2,-2) {$T_e$};

    \path [<-,>=latex](f) edge node {} (f*);
    \path [<-,>=latex](f*) edge node {} (a);
    \path [<-,>=latex](f*) edge node {} (b);
    \path [<-,>=latex](f*) edge node {} (c);
    \path [<-,>=latex](f*) edge node {} (d);
    \path [<-,>=latex](f*) edge node {} (e);  
    \path [->,>=latex](Ta) edge node {} (a);
    \path [->,>=latex](Tb) edge node {} (b);
    \path [->,>=latex](Tc) edge node {} (c);
    \path [->,>=latex](Td) edge node {} (d);
    \path [->,>=latex](Te) edge node {} (e);

	\node (right) at (3.5,-1) {\large$\rightarrow$};

    \node[local,fill=black!20!white]
         (f) at (5.5,2) {$T$};
    \node[local]
         (aa) at (8,2) {$a$};
         
    \node[local]
         (dd) at (10,2) {$d$};

    \node[opt,label=west:$f^*$] 
    	 (f*1) at (5.5,1) {};
    \node[opt,label=west:$f^*$] 
    	 (f*2) at (8,1) {};
    \node[opt,label=west:$f^*$] 
    	 (f*3) at (10,1) {};
         
    \node[local] 
         (a) at (5,0) {$a$};
    \node[local] 
         (b) at (6,0) {$b$};
    \node[local] 
         (c) at (7.5,0) {$c$};
    \node[local] 
         (d) at (8.5,0) {$d$};
    \node[local] 
         (e) at (10,0) {$e$};

    \node[tree]  (Ta) at (5,-2) {$T_a$};
    \node[tree]  (Tb) at (6,-2) {$T_b$};
    \node[tree]  (Tc) at (7.5,-2) {$T_c$};
    \node[tree]  (Td) at (8.5,-2) {$T_d$};
    \node[tree]  (Te) at (10,-2) {$T_e$};

    \path [<-,>=latex](f) edge node {} (f*1);
    \path [<-,>=latex](f*1) edge node {} (a);
    \path [<-,>=latex](f*1) edge node {} (b);
    \path [<-,>=latex](f*2) edge node {} (c);
    \path [<-,>=latex](f*2) edge node {} (d);
    \path [<-,>=latex](f*3) edge node {} (e);  
    \path [->,>=latex](Ta) edge node {} (a);
    \path [->,>=latex](Tb) edge node {} (b);
    \path [->,>=latex](Tc) edge node {} (c);
    \path [->,>=latex](Td) edge node {} (d);
    \path [->,>=latex](Te) edge node {} (e);  
    \path [<-,>=latex](aa) edge node {} (f*2);   
    \path [<-,>=latex](dd) edge node {} (f*3);    
    \path[->,>=latex](aa) edge[in=30,out=330] node {} (f*2);
    \path [->,>=latex](dd) edge[in=30,out=330] node {}  (f*3);   
\end{tikzpicture}

\caption{The figure shows the decomposition of high in-degree optimal facility $f^*$ for 
$\thd=2$. Shaded rectangular boxes correspond to part of the original tree that does not change. Since the degree of $f^*$ is 5 $\geq \thd$, 
we create $\ceil{5/\thd}$ trees. The first tree stays in the original tree. 
Each remaining tree gets a $f^*$'s child chosen uniformly at random from the previous tree. $f^*$ gets open 2 extra times, 
but we also close $a$ and $d$ to balance the number of opening and closure.
In this example $a$ and $d$ are chosen as optimal surrogates. And $a$ and $b$ are two closest children to $f^*$ among $\{a,b,c,d,e\}$.}
\label{fig:heavyOptDecomp}
\end{figure}

In summary, the degree reduction procedure for tree swaps consists of the following steps:
\begin{OneLiners}
\item[1.] Remove edges incident to all \heavy local facilities;
\item[2.] Replace every \heavy local facility by its local surrogate, chosen uniformly at random without replacement from the local candidates;
\item[3.] Deal with heavy optimal facilities as above;
\item[4.] Add self-loops to vertices with no out-edge (due to step 1) to retain the 1-forest structure (this facilitates a cleaner presentation of our next procedure: edge deletion).
\end{OneLiners}
Let $G_1$ denote the graph after degree reduction. $G_1$ is still a 1-forest, and every vertex in $G_1$ now has in-degree at most $\thd+1$. Moreover, the following claim is apparent (by observing that \Cref{clm:light-simple} still holds in the tree swaps case because its proof is completely independent of the distribution of $\tau(f^*)$):
\begin{claim}
\label{clm:tree-degree-reduction}
Constructed graph $G_1$ follows following properties:
\begin{OneLiners}
\item[i.] Every optimal facility appears in $G_1$ at least once. 
\item[ii.] Every local facility appears in $G_1$ at most three times: once as the original copy, once as a local surrogate, and once as an optimal surrogate. 
\item[iii.] Heavy local facilities do not appear in $G_1$. 
\item[iv.] No two copies of the same facility appear in the same connected component.
\item[v.]  The original copy of any optimal facility $f^*$ points to the original copy of $\tau(f^*)$, unless $\tau(f^*)$ is \heavy.
\item[vi.] The original copy of any local facility $f$ points to $\pi(f)$, although it might be a new copy of $\pi(f)$. 
\item[vii.] Any local facility is chosen as a local surrogate with probability at most $\nicefrac 2\thd$, and as an optimal surrogate with probability at most $\nicefrac 1\thd$. 
\item[viii.] Every local surrogate is a local candidate.
\end{OneLiners}
\end{claim}
The degree-reduction step ensures that vertices in $G_1$ have bounded
in-degree, but a connected component of $G_1$ could still have large
size (it could have large height or contain a long cycle). We deal with this problem in our next procedure: \emph{edge deletion}.

\paragraph{Edge deletion.}
Next, we remove edges from $G_1$ to ensure that every connected component in the resulting graph is a tree of height at most $\thh - 1$, where we choose the height threshold $\thh$ uniformly at random from $2\lceil \nicefrac 1\varepsilon\rceil, 2\lceil \nicefrac 1\varepsilon\rceil^2,\cdots, 2\lceil \nicefrac 1\varepsilon\rceil^{\lceil \nicefrac 1\varepsilon\rceil}$.
%
%
%
%
%
%
Specifically, for each connected component $T$ of $G_1$, if the
root cycle has length less than $\thh$, we insert dummy vertices into the cycle to make
the length exactly $\thh$. Then we pick a vertex $r$ in the root cycle
uniformly at random, and delete the out-edge from $r$. This makes $T$ a directed tree rooted at $r$. We then delete edges on the $a\cdot \thh$-th levels for all $a\in \N$. See \Cref{fig:treeswapbounded} for an example.

 \begin{figure}[htb]
   \center
\begin{tikzpicture}[scale = 0.8,
		opt/.style={shape=regular polygon,regular polygon sides = 3,draw=black,minimum size=0.5cm,inner sep = 0pt},
		sopt/.style={shape=regular polygon,regular polygon sides = 3,draw=black,minimum size=0.4cm,inner sep = 0pt},
		local/.style={shape=rectangle, draw=black,minimum size= 0.4cm},
		main/.style={shape=rectangle, draw=black,minimum size= 0.3cm},
]
\pgfset{
  foreach/parallel foreach/.style args={#1in#2via#3}{evaluate=#3 as #1 using {{#2}[#3-1]}},
}
	\def \n {8}
	\pgfmathsetmacro{\nm}{\n-1}
	\def \margin {12}
	\def \radius {2cm}

	\def \RED {red!50}
	\def \BLUE {blue!30}
	\def \GREEN { green!30}
	
	\def \colorlist{"\BLUE","\BLUE","\BLUE","\RED","\RED","\RED","\RED","\BLUE"}
	\def \namelist{"","","r","","","","",""}	
	\def \edgelist{"solid","solid","dashed","solid","solid","solid","dashed","solid"}
	\foreach \s [parallel foreach=\isDash in \edgelist via \s]  in {1,...,\n}
        {    
            \draw[->, >=latex,\isDash] ({360/\n * (\s - 1)+\margin}:\radius) 
                arc ({360/\n * (\s - 1)+\margin}:{360/\n * (\s)-\margin}:\radius);
        }
        
	\foreach \s [count=\x,
			parallel foreach=\c in \colorlist via \x,
			parallel foreach=\name in \namelist via \x]
			in {opt,local,opt,local,opt,local,opt,local}{
		\node[\s,fill=\c] (c\x) at ({360/\n*(\x-1)}:\radius) {\small $\name$};	
	}

	\def \c {\RED}
	\node[main,fill=\c] (rf1) at (-3,0.75) {};
	\node[main,fill=\c] (rf2) at (-3,0.25) {};
	\node[main,fill=\c] (rf3) at (-3,-0.25) {};
	\node[main,fill=\c] (rf4) at (-3,-0.75) {};
	
	\foreach \x in {1,...,4}
		\draw (rf\x) edge[->,>=latex] (c5);

	\def \c {\RED}
	\node[sopt,fill=brown!70] (luf1) at (-2.6,1.8) {};
	\node[sopt,fill=gray!50] (luf2) at (-2.2,2.2) {};	
	\node[sopt,fill=violet!10] (luf3) at (-1.8,2.6) {};	
	
	\foreach \x in {1,...,3}
		\draw (luf\x) edge[->,>=latex,dashed] (c4);

	\def \c {\BLUE}
	\node[main,fill=\c] (lf1) at (3,0.5) {};
	\node[main,fill=\c] (lf2) at (3,0) {};	
	\node[main,fill=\c] (lf3) at (3,-0.5) {};
	
	\foreach \x in {1,...,3}
		\draw (lf\x) edge[->,>=latex] (c1);


	\def \c {\GREEN}
	
	\node[sopt,fill=\c](lf20) at (4,-0.5) {};
	\draw (lf20) edge[->,>=latex,dashed] (lf3);
	
	\foreach \s [count=\z] in {main,sopt}{
		\foreach \x in {1,...,3}{
			\pgfmathsetmacro{\y}{-1.5+0.5*\x}
			\pgfmathsetmacro{\h}{4+\z}
			\pgfmathtruncatemacro{\zm}{\z-1}				
			\node[\s,fill=\c] (lf\x\z) at (\h,\y) {};
			\draw (lf\x\z) edge[->,>=latex] (lf2\zm);
		}	
	}

	\node[sopt,fill=lime!30] (lfx) at (4,0) {};
	\node[main,fill=lime!30] (lfxx) at (5,1) {};
	\draw (lfxx) edge [->,>=latex] (lfx) {};
	\node[sopt,fill=olive!70] (lfy) at (4,-1) {};
	\node[main,fill=olive!70] (lfyy) at (5,-2) {};
	\draw (lfyy) edge [->,>=latex] (lfy) {};	
	\draw (lfx) edge[->,>=latex,dashed] (lf3) {};
	\draw (lfy) edge[->,>=latex,dashed] (lf3) {};


	\def \c {\RED}
	\node[main,fill=\c] (cloc) at (0,-3) {};
	\draw (cloc) edge[->,>=latex] (c7);

	\def \cLBL {lime!30}
	\def \cLBR {olive!70}

	\foreach \x in {1,...,5}{
		\pgfmathsetmacro{\y}{-1.5+0.5*\x}
		\node[sopt,fill=\c] (f\x) at (\y,-4) {};
		\draw (f\x) edge[->,>=latex] (cloc);
	}
		
	
	\foreach \x in {1,...,3}{
		\pgfmathsetmacro{\y}{-2+0.5*\x}
		\node[main,fill=\c] (f1\x) at (\y,-4.75) {};
		\draw (f1\x) edge[->,>=latex] (f1);
	}
	
	\foreach \x in {1,...,3}{
		\pgfmathsetmacro{\y}{0+0.5*\x}
		\node[main,fill=\c] (f2\x) at (\y,-4.75) {};
		\draw (f2\x) edge[->,>=latex] (f5);
	}
	
	
\end{tikzpicture}

 \caption{Example for $\thh=4$. Nodes with the same color
 correspond to nodes in the same connected component after edge deletion. We start from $r$ (randomly chosen), 
 and repeatedly cut edges on $a\cdot\thh$  steps away from $r$ (dashed edges).
 }\label{fig:treeswapbounded}
 \end{figure}
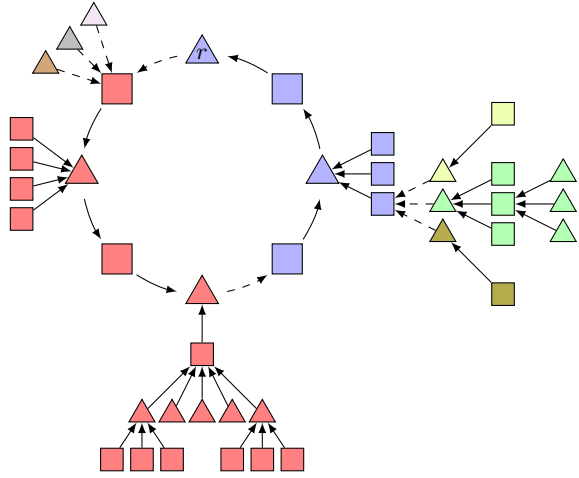

Let $G_2$ be the graph after the edge deletion step. It is clear that every connected component of $G_2$ is a directed tree with height at most $\thh - 1$, possibly containing some dummy vertices. Moreover, every vertex $v$ has in-degree at most $\thd+1$ due to the degree reduction procedure. Therefore, the number of vertices in every connected component of $G_2$ is at most $(\thd+1)^{\thh}$. Moreover, we have the following claim for every connected component $T$ of $G_1$, which is apparent from our edge deletion procedure:

\begin{claim}
\label{claim:edge-deletion}
After dummy vertices are added into $T$, the edge out of vertex $v\in T$ is deleted if and only if the (unique) simple path from $v$ to $r$ has length divisible by $\thh$.
\end{claim}

If the cycle length of $T$ is at most $\thh$, vertices on the cycle are still connected after edge deletion. Indeed, we delete only one edge in the cycle in this case. Therefore, after edge deletion, we ignore all the dummy vertices and still consider all the edges on the original cycle as not deleted by convention. This doesn't change the (non-dummy) vertices in every connected component of $G_2$, and thus doesn't change $\cP$ we eventually generate. With this convention, we have the following corollary of \Cref{claim:edge-deletion}:

\begin{corollary}
\label{cor:edge-deletion}
Any edge in $G_1$ is deleted with probability at most $\nicefrac 2\thh$. Moreover, if the cycle length is at most $\thh$, edges on the cycle are never deleted.
\end{corollary}
\begin{proof}
The second part is assumed by our convention. We thus assume henceforth that the edge is not on the cycle, or the cycle length is more than $\thh$. Suppose the edge is the out-edge of vertex $v$. By \Cref{claim:edge-deletion}, the edge is deleted if and only if the simple path $p^*$ from $v$ to $r$ has length divisible by $\thh$. Suppose the cycle length after dummy vertices are added to it is $\ell \geq \thh$, and let $\ell = u\thh + w$ for $u,w\in\mathbb Z$ with $0\leq w < \thh$. There are at most $u + 1$ choices of $r$ such that $p^*$ has length divisible by $\thh$. Therefore, the edge is deleted with probability at most $(u + 1)/\ell = u/\ell + 1/\ell \leq 1/\thh + 1/\thh = 2/\thh$.
\end{proof}

After edge deletion, each connected component of $G_2$ contains at most $(\thd + 1)^{\thh} \leq p(\varepsilon)$ vertices. However, the number of local
and optimal facilities in the component may not match (e.g., the blue
tree containing $r$ in Figure~\ref{fig:treeswapbounded} has three extra local facilities,
whereas the rightmost tree has one extra optimal facility). We fix this in the same way as in the simple swaps case using the balancing procedure.
\paragraph{Balancing.} The balancing procedure is essentially the same
as in the simple swaps case, based on \Cref{claim:balancedGroup}
again. The only difference is that the size of every connected
component is now much larger ($(\thd + 1)^{\thh}$), and the maximum degree of $H$ is also
much larger. Since optimal facilities may now have new copies, we may combine two connected components each containing a copy of the same optimal facility in the balancing step; this is fine because it only decreases the number of optimal facilities in a swap. However, we still need to make sure that no two copies of the same local facility are combined together, again by adding edges into $H$ between connected components containing copies of the same local facility.
Since a local facility can have at most 3 copies by \Cref{clm:tree-degree-reduction}, the maximum degree of $H$ is at most $2(\thd+1)^{\thh}$. Since we kept the number of extra local facilities unchanged, it's still $\extra\geq \Omega(((\thd+1)^{\thh})^5(2(\thd+1)^{\thh})^3/\varepsilon)$, so \Cref{claim:balancedGroup} gives balanced groups each containing at most $O((\thd+1)^{2\thh}) \leq p(\eps)$ facilities.
Every group thus defines a valid swap, and we define $\cP$ as the
set of these tree swaps.

\section{Client Types}
\label{subsec:types}
\label{subsec:ClientType}

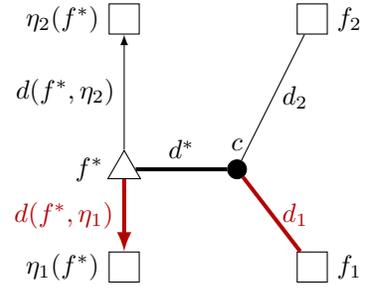
\begin{wrapfigure}{r}{0.35\textwidth}
\center
\begin{tikzpicture}[scale=1,
		opt/.style={shape=regular polygon,regular polygon sides = 3,draw=black,minimum size=0.5cm,inner sep = 0pt},
		local/.style={shape=rectangle, draw=black,minimum size= 0.4cm},
		client/.style={shape=circle, draw=black,minimum size=0.25cm,inner sep = 0pt,fill=black}]
		
    \node[opt,label=west:$f^*$]  (f*) at (0,0) {};
    	
    \node[local,label=west:$\eta_2(f^*)$]   (eta2) at (0,2) {};
    \node[local,label=west:{$\eta_1(f^*)$}] (eta) at (0,-1.3) {};
         
    \node[client,label=north:$c$] (c) at (1.5,0) {};

    \node[local,label=east:{$f_1$}]     	 (f1) at (2.5,-1.3) {};
    \node[local,label=east:{$f_2$}]       (f2) at (2.5,2) {} ;
    \path [-,line width=1.5pt](f*) edge[above] node {$d^*$} (c);

    \path [-](c) edge[right] node {$d_2$} (f2);
    \path [-,mystyle](c) edge[right] node {$d_1$} (f1);
    \path [->,>=latex,mystyle](f*) edge[left] node {$d(f^*,\eta_1)$}(eta);
    \path [->,>=latex](f*) edge[left] node {$d(f^*,\eta_2)$}(eta2);

\end{tikzpicture}
\caption{The squares are local facilities, triangles are optimal
  facilities, circle are clients. The thick red edge out of $c$ goes
  to its closest local facility $f_1$; the thick red edge out of $f^*$
  goes to its closest local facility $\eta_1(f^*)$.}
\label{fig:general}
\end{wrapfigure}

We now classify the clients into a small number of types
(based on how the client connects to facilities in the local and
global solutions).
The classification allows us to give a \emph{client-by-client}
analysis instead of a \emph{swap-by-swap} analysis used in prior
works. We make this change in perspective because the potential $\Phi$
depends on the two closest facilities, and so we need a better handle
on the local neighborhood of a client to bound the reassignment costs
when closing one of the close facilities.

For a client $c$, recall that $f_1(c)$ and $f_2(c)$ are the closest
and second-closest local facilities; we say $f_1$ and $f_2$ when there
is no ambiguity.  Figure~\ref{fig:general} shows a picture of a
generic client $c$ and its related facilities.



We partition the set of clients into types based on the relationships
between their local and optimal facilities, as follows.  The
\emph{far} clients are those for which $d_2 \geq \a d_1$, and hence
the potential just depends on the closest facility ($f_1$); the other kinds of
clients are called \emph{close}, for which both $f_1$ and $f_2$ are relevant.

\begin{itemize}
\item \textbf{Far case} (where $d_2 \geq \a d_1$). Note that $f_2$
  does not play any role in the far case, so the clients are
  classified according to how $f_1$ and $f^*$ are related.
\begin{itemize}
\item Type \xA: $\eta_1(f^*) = f_1$.
\item Type \xB: $\eta_2(f^*) = f_1$. 
\item Type \xE: $f_1 \not\in \{\eta_1(f^*), \eta_2(f^*)\}$.  
\end{itemize}

\begin{figure}[H]
\center
\subfigure[Type \xA]{
\begin{tikzpicture}[scale=1,
		opt/.style={shape=regular polygon,regular polygon sides = 3,draw=black,minimum size=0.5cm,inner sep = 0pt},
		local/.style={shape=rectangle, draw=black,minimum size= 0.4cm},
		client/.style={shape=circle, draw=black,minimum size=0.25cm,inner sep = 0pt,fill=black}]
		
    \node[opt,label=west:$f^*$] 
    	 (f*) at (0,0) {};
    	
    \node[local,label=west:{$\eta_2$}]
         (eta2) at (0,2) {};
    \node[local,label=west:{$\eta_1=f_1$}] 
         (eta) at (0,-1.3) {};
         
    \node[client,label=north:$c$] 
         (c) at (1.5,0) {};


    \path [-](f*) edge[above] node {} (c);

    \path [-,mystyle](c) edge[above] node {} (eta);
    \path [->,>=latex,mystyle](f*) edge[right] node {}(eta);
    \path [->,>=latex](f*) edge[right] node {}(eta2);

\end{tikzpicture}
}
\subfigure[Type \xB]{
\begin{tikzpicture}[scale=1,
		opt/.style={shape=regular polygon,regular polygon sides = 3,draw=black,minimum size=0.5cm,inner sep = 0pt},
		local/.style={shape=rectangle, draw=black,minimum size= 0.4cm},
		client/.style={shape=circle, draw=black,minimum size=0.25cm,inner sep = 0pt,fill=black}]
    \node[opt,label=west:$f^*$] 
    	 (f*) at (0,0) {};
    	
    \node[local,label=west:{$\eta_2=f_1$}]
         (eta2) at (0,2) {};
    \node[local,label=west:{$\eta_1$}] 
         (eta) at (0,-1.3) {};
         
    \node[client,label=north:$c$] 
         (c) at (1.5,0.5) {};


    \path [-](f*) edge[above] node {} (c);

    \path [-,mystyle](c) edge[above] node {} (eta2);
    \path [->,>=latex,mystyle](f*) edge[right] node {}(eta);
    \path [->,>=latex](f*) edge[right] node {}(eta2);

\end{tikzpicture}
}
\subfigure[Type \xE]{
\begin{tikzpicture}[scale=1,
		opt/.style={shape=regular polygon,regular polygon sides = 3,draw=black,minimum size=0.5cm,inner sep = 0pt},
		local/.style={shape=rectangle, draw=black,minimum size= 0.4cm},
		client/.style={shape=circle, draw=black,minimum size=0.25cm,inner sep = 0pt,fill=black}]
    \node[opt,label=west:$f^*$] 
    	 (f*) at (0,0) {};
    	
    \node[local,label=west:{$\eta_2$}]
         (eta2) at (0,2) {};
    \node[local,label=west:{$\eta_1$}] 
         (eta) at (0,-1.3) {};
         
    \node[client,label=north:$c$] 
         (c) at (1.5,0) {};

    \node[local,label=east:{$f_1$}] 
    	 (f1) at (2.5,0) {};

    \path [-](f*) edge[above] node {} (c);

    \path [-,mystyle](c) edge[above] node {} (f1);
    \path [->,>=latex,mystyle](f*) edge[right] node {}(eta);
    \path [->,>=latex](f*) edge[right] node {}(eta2);

\end{tikzpicture}
}

\caption{Far Case}
\end{figure}

\item \emph{Close case} (where $d_2 \leq \a d_1$); now clients are
  classified according to how $f_1, f_2$ and $f^*$ are related.
\begin{itemize}
\item Type \xA: $\eta_1(f^*) = f_1$ and  $\eta_2(f^*) \neq f_2$.
\item Type \xB: $\eta_1(f^*) \neq f_2$ and  $ \eta_2(f^*) = f_1. $ 
\item Type \xC: $\eta_1(f^*) = f_1$ and $ \eta_2(f^*) = f_2$. 
\item Type \xD: $\eta_1(f^*) = f_2$ and  $ \eta_2(f^*) = f_1$.  
\item Type \xE: $f_1 \not\in \{\eta_1(f^*), \eta_2(f^*)\}$. 
\end{itemize}
\end{itemize}

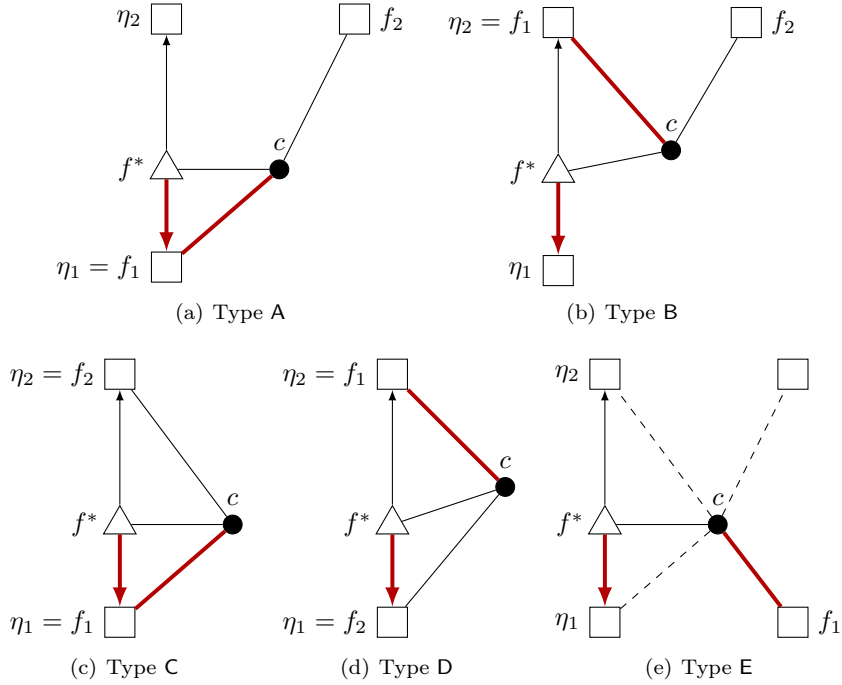
\begin{figure}[H]
\center
\subfigure[Type \xA]{
\begin{tikzpicture}[scale=1,
		opt/.style={shape=regular polygon,regular polygon sides = 3,draw=black,minimum size=0.5cm,inner sep = 0pt},
		local/.style={shape=rectangle, draw=black,minimum size= 0.4cm},
		client/.style={shape=circle, draw=black,minimum size=0.25cm,inner sep = 0pt,fill=black}]
    \node[opt,label=west:$f^*$] 
    	 (f*) at (0,0) {};
    	
    \node[local,label=west:$\eta_2$]
         (eta2) at (0,2) {};
    \node[local,label=west:{$\eta_1=f_1$}] 
         (eta) at (0,-1.3) {};
         
    \node[client,label=north:$c$] 
         (c) at (1.5,0) {};

    \node[local,label=east:{$f_2$}] 
    	 (f2) at (2.5,2) {} ;

    \path [-](f*) edge[above] node {} (c);
    \path [-](c) edge[above] node {} (f2);
    \path [-,mystyle](c) edge[above] node {} (eta);
    \path [->,>=latex,mystyle](f*) edge[right] node {}(eta);
    \path [->,>=latex](f*) edge[right] node {}(eta2);

\end{tikzpicture}

}
\subfigure[Type \xB]{
\begin{tikzpicture}[scale=1,
		opt/.style={shape=regular polygon,regular polygon sides = 3,draw=black,minimum size=0.5cm,inner sep = 0pt},
		local/.style={shape=rectangle, draw=black,minimum size= 0.4cm},
		client/.style={shape=circle, draw=black,minimum size=0.25cm,inner sep = 0pt,fill=black}]
    \node[opt,label=west:$f^*$] 
    	 (f*) at (0,0) {};
    	
    \node[local,label=west:{$\eta_2=f_1$}]
         (eta2) at (0,2) {};
    \node[local,label=west:{$\eta_1$}] 
         (eta) at (0,-1.3) {};
         
    \node[client,label=north:$c$] 
         (c) at (1.5,0.3) {};

    \node[local,label=east:{$f_2$}] 
    	 (f2) at (2.5,2) {} ;

    \path [-](f*) edge[above] node {} (c);

    \path [-](c) edge[above] node {} (f2);
    \path [-,mystyle](c) edge[above] node {} (eta2);
    \path [->,>=latex,mystyle](f*) edge[right] node {}(eta);
    \path [->,>=latex](f*) edge[right] node {}(eta2);

\end{tikzpicture}
}

\subfigure[Type \xC]{
\begin{tikzpicture}[scale=1,
		opt/.style={shape=regular polygon,regular polygon sides = 3,draw=black,minimum size=0.5cm,inner sep = 0pt},
		local/.style={shape=rectangle, draw=black,minimum size= 0.4cm},
		client/.style={shape=circle, draw=black,minimum size=0.25cm,inner sep = 0pt,fill=black}]
    \node[opt,label=west:$f^*$] 
    	 (f*) at (0,0) {};
    	
    \node[local,label=west:{$\eta_2=f_2$}]
         (eta2) at (0,2) {};
    \node[local,label=west:{$\eta_1=f_1$}] 
         (eta) at (0,-1.3) {};
         
    \node[client,label=north:$c$] 
         (c) at (1.5,0) {};


    \path [-](f*) edge[above] node {} (c);

    \path [-,mystyle](c) edge[above] node {} (eta);
    \path [-](c) edge[above] node {} (eta2);
    \path [->,>=latex,mystyle](f*) edge[right] node {}(eta);
    \path [->,>=latex](f*) edge[right] node {}(eta2);

\end{tikzpicture}
}
\subfigure[Type \xD]{
\begin{tikzpicture}[scale=1,
		opt/.style={shape=regular polygon,regular polygon sides = 3,draw=black,minimum size=0.5cm,inner sep = 0pt},
		local/.style={shape=rectangle, draw=black,minimum size= 0.4cm},
		client/.style={shape=circle, draw=black,minimum size=0.25cm,inner sep = 0pt,fill=black}]
    \node[opt,label=west:$f^*$] 
    	 (f*) at (0,0) {};
    	
    \node[local,label=west:{$\eta_2=f_1$}]
         (eta2) at (0,2) {};
    \node[local,label=west:{$\eta_1=f_2$}] 
         (eta) at (0,-1.3) {};
         
    \node[client,label=north:$c$] 
         (c) at (1.5,0.5) {};


    \path [-](f*) edge[above] node {} (c);

    \path [-,mystyle](c) edge[above] node {} (eta2);
    \path [-](c) edge[above] node {} (eta);
    \path [->,>=latex,mystyle](f*) edge[right] node {}(eta);
    \path [->,>=latex](f*) edge[right] node {}(eta2);

\end{tikzpicture}
}
\subfigure[Type \xE]{
\begin{tikzpicture}[scale=1,
		opt/.style={shape=regular polygon,regular polygon sides = 3,draw=black,minimum size=0.5cm,inner sep = 0pt},
		local/.style={shape=rectangle, draw=black,minimum size= 0.4cm},
		client/.style={shape=circle, draw=black,minimum size=0.25cm,inner sep = 0pt,fill=black}]
    \node[opt,label=west:$f^*$] 
    	 (f*) at (0,0) {};
    	
    \node[local,label=west:{$\eta_2$}]
         (eta2) at (0,2) {};
    \node[local,label=west:{$\eta_1$}] 
         (eta) at (0,-1.3) {};
         
    \node[client,label=north:$c$] 
         (c) at (1.5,0) {};

    \node[local,label=east:{$f_1$}] 
    	 (f1) at (2.5,-1.3) {};
    \node[local] 
    	 (f2) at (2.5,2) {} ;

    \path [-](f*) edge[above] node {} (c);

    \path [-,mystyle](c) edge[above] node {} (f1);
    \path [->,>=latex,mystyle](f*) edge[right] node {}(eta);
    \path [->,>=latex](f*) edge[right] node {}(eta2);
    \path [-,dashed](c) edge[above] node {} (f2);
    \path [-,dashed](c) edge[above] node {} (eta);
    \path [-,dashed](c) edge[above] node {} (eta2);

\end{tikzpicture}
}

\caption{Close Case. For Type \xE, the client's $f_2$ can be any one of dashed edges.}
\end{figure}

\section{Amenable and Defiant Events}
\label{sec:amenable-defiant-swaps}
%
%
Not all swaps are easy to argue about. Having fixed a client $c$, we
define the \emph{amenable event} and \emph{defiant event} for this
client---the former captures the case where the swaps in $\cP$ are easy to reason
about, and the latter the case where we throw up our hands and use a
crude bound on the potential change. Thankfully, the latter happens
very rarely, so the loss is small.

Recall that $f_1(c), f_2(c)$ are the two closest local facilities to
$c$. Let $f^* = f^*(c)$ be the optimal facility that $c$ is assigned to; then
$\eta_1(f^*), \eta_2(f^*)$ are the two closest local facilties to
$f^*$. We define the amenable and defiant events as follows:
\begin{restatable}[Amenable/Defiant]{definition}{amenableDefiantSwaps}
  \label{def:amenable-defiant-swaps}
  The \emph{defiant event} $\cD$ for a client $c$ of type \xA, \xB or \xE is the union of the following events:
  \begin{OneLiners}
  \item[(i)] $f_1$, $f_2$, or $\tau(f^*)$ is chosen as a local or optimal surrogate in the degree reduction step;
  \item[(ii)] $\cP$ is a tree swaps set, and the out-edge from the original copy of $f^*$, $f_1$ or $f_2$ is deleted in the edge deletion step.
  \item[(iii)] $\cP$ is a simple swaps set, and two connected components each containing a facility in $\{f^*\}\cup \{f_1,f_2\}\cup\{\eta_1,\eta_2\}$ are grouped together in the balancing step.
  \end{OneLiners}
The \emph{amenable event} $\cA$ is the complement of $\cD$.
\end{restatable}

For type \xC and \xD clients, we enlarge the defiant event slightly to include $g^*:=\pi(f_1)$ and 
\[
g:=\argmin_{h\in F\backslash\{f_1,f_2\}}d(h,g^*)
\]
as follows:
\begin{definition}[Amenable/Defiant for type \xC and \xD]
\label{def:amenable-defiant-swaps-CD}
  The \emph{defiant event} $\cD$ for a client $c$ of type \xC or \xD is the union of the events (i), (ii), (iii) in \Cref{def:amenable-defiant-swaps} and the following events:
   \begin{OneLiners}
  \item[(i')] $\tau(g^*)$ is chosen as a local or optimal surrogate in the degree reduction step;
  \item[(ii')] $\cP$ is a tree swaps set, and the out-edge from the original copy of $g^*$ is deleted in the edge deletion step.
  \item[(iii')] $\cP$ is a simple swaps set, and two connected components each containing a facility in $\{f_1,g\}$ are grouped together in the balancing step.
  \end{OneLiners} 
The \emph{amenable event} $\cA$ is the complement of $\cD$.
\end{definition}

The events $\cA$ and $\cD$ depend on the client $c$, but we choose to
omit $c$ in our notation because we will always focus on a fixed
client $c$ in our proof. We now turn to proving \Cref{lem:combo} on
the potential change due to defiant events. The approach is simple: we
first
show a crude upper bound that holds for all swap sets $\cP$ that we
generate, and then show that the probability of the defiant event is
small enough so that we can afford to apply this crude upper bound. 
\begin{restatable}{claim}{Crude}
  \label{clm:crude}
  There is an absolute constant $\gamma>0$ such that for any client $c$, and any swap set $\cP$ that  we generate, we have 
  $\sum_{(P,Q)\in\cP}\dPQ{c}\leq \gamma(d^*(c) + d_1(c))$. 
\end{restatable}
\begin{restatable}{claim}{CrudeTwo}
  \label{clm:crude2}
  $\Pr[\cD]\leq O(\varepsilon)$ for all clients $c$.
\end{restatable}
The proof of \Cref{clm:crude2} follows from
\Cref{clm:simple-degree-reduction,claim:balancedGroup,clm:tree-degree-reduction},
\Cref{cor:edge-deletion},
and a trivial union bound. We defer the proof of \Cref{clm:crude} to
Appendix~\ref{sec:defiant}.
The two claims above imply~\Cref{lem:combo}, and hence control the
effect of the defiant events.
We focus next on the amenable
events and the proof of \Cref{lem:combo2}.

\newcommand{\change}{{\Delta}}
\newcommand{\WCchange}{{\delta}}
\newcommand{\trtype}{\gamma}

\section{The Potential Change due to Amenable Events}
\label{sec:bounds}


Having bounded the potential change due to defiant events, we now turn
to bounding the potential change due to amenable events. Let us recall
the claim we want to 
prove:
\AmenableSwaps*

%
This section gives an
explicit proof that can be verified by hand.  In \S\ref{sec:LP} we
show how to generate a much larger set of valid inequalities. Solving
the resulting linear program gives improved bounds, but these are more
tedious to verify manually.


\subsection{Implications of Amenability}
\label{sec:amenability}
\begin{claim}[Implications of amenability]
\label{clm:implication-amenability}
For any client, swap sets $\cP$ generated on the amenable event $\cA$ have the following properties:
\begin{OneLiners}
\item[(i)] Any local facility $f\in\{f_1,f_2\}$ is closed in at most one swap in $\cP$;
\item[(ii)] Any swap in $\cP$ closing $\tau(f^*)$ must open the original copy of $f^*$;
\item[(Tii)] If $\cP$ is a tree swap set, any swap in $\cP$ closing $f\in\{f_1,f_2\}$ must open $\pi(f)$;
\item[(Siii)] If $\cP$ is a simple swap set, no swap in $\cP$ closes two local facilities in $\{f_1,f_2\}\cup \{\eta_1,\eta_2\}$ simultaneously;
\item[(Siv)] If $\cP$ is a simple swap set, any swap in $\cP$ closing a local facility in $\{f_1,f_2\}\backslash\{\tau(f^*)\}$ does not open $f^*$.
\end{OneLiners}
%
%
For clients of type \xC or \xD, we additionally have the following:
(recall that we defined $g^*$ as $\pi(f_1)$, and $g$ as the local facility closest to $g^*$ other than $f_1$ and $f_2$):
\begin{OneLiners}
\item[(ii')]Any swap in $\cP$ closing $\tau(g^*)$ must open the original copy of $g^*$;
\item[(Siii')]If $\cP$ is a simple swap set, no swap in $\cP$ closes both $f_1$ and $g$.
\end{OneLiners}
\end{claim}

\begin{proof}[Proof of \Cref{clm:implication-amenability}]
Recall that the amenable event $\cA$ is the complement of the defiant event $\cD$, defined in \Cref{def:amenable-defiant-swaps}.

Implication (i) follows from item (i) of \Cref{def:amenable-defiant-swaps} directly. 

Implication (ii) follows from items (i) and (ii) of \Cref{def:amenable-defiant-swaps}. Without loss of generality, we assume $\tau(f^*)$ is not \heavy, since \heavy local facilities are never closed. On the amenable event, $\tau(f^*)$ is closed only as its original copy, by item (i) of \Cref{def:amenable-defiant-swaps}. The edge to $\tau(f^*)$ from the original copy of $f^*$ is never deleted by item (ii) of \Cref{def:amenable-defiant-swaps}, so the original copies of $f^*$ and $\tau(f^*)$ must be in the same swap.

Implication (Tii) also follows from items (i) and (ii) of \Cref{def:amenable-defiant-swaps}, for a similar reason. Again, assume without loss of generality that neither $f_1$ nor $f_2$ is heavy. On the amenable event, $f_1$ and $f_2$ are closed only as their original copies by item (i) of \Cref{def:amenable-defiant-swaps}, and the edges $f_i\rightarrow \pi(f_i)$ are never deleted by item (ii).

Implications (Siii) and (Siv) both follow from item (iii) of \Cref{def:amenable-defiant-swaps}. When we generate the simple swap set, every connected component of the graph $G_1$ contains at most one local facility, and  thus different facilities in $\{f_1,f_2\}\cup\{\eta_1,\eta_2\}$ must be in different connected components, which are not combined in the balancing step due to item (iii) of \Cref{def:amenable-defiant-swaps}. This proves implication (Siii). Moreover, the connected component of $f^*$ doesn't contain any local facility other than $\tau(f^*)\in\{\eta_1,\eta_2\}$. This proves implication (Siv).

(ii') and (Siii') can be proved in the same way as (ii) and (Siii) using \Cref{def:amenable-defiant-swaps-CD}.
\end{proof}
\subsection{Notation and Useful Inequalities}
\label{sec:bounds-notation}
\label{sec:usefulineq}


Let $\change_\cE(c)$ denote the expected potential change on client $c$
restricted to some generic event $\cE$:
\begin{equation*}
\change_\cE(c) := \bE \Big[ \ind_{\cE} \sum_{(P,Q)\in\cP}\dPQ{c} \Big].
\end{equation*}
Our goal in \Cref{lem:combo2} is thus to upper bound $\change_\cA(c)$
for the amenable event $\cA$. In our proof, we consider sub-events
$\cE$ of $\cA$, and prove worst-case upper-bounds for the potential change
restricted to each sub-event $\cE$. Formally, given a suitable
partition $\cA = \cE_1\cup\cdots\cup\cE_t$, we define
$\WCchange_\cE(c) := \sum_{(P,Q)\in\cP}\dPQ{c}$ to be the worst-case
(maximum) value for each event $\cE$, and then use:
%
%
\begin{equation}
  \label{eq:partition}
  \change_\cA(c) = \sum_{i=1}^t\change_{\cE_i}(c)
\leq \sum_{i=1}^t \Pr[\cE_i]\; \WCchange_{\cE_i}(c).
\end{equation}

For technical reasons, it is more convenient to assume $\WCchange_\cE(c)$ is no smaller than, say, $\WClb(c)$. We thus re-define $\WCchange_\cE(c)$ as $\WClb(c)$ when $\WCchange_\cE(c) < \WClb(c)$. This doesn't affect our analysis, as all our upper bounds for $\WCchange_\cE(c)$ are larger than $\WClb(c)$. Also, \Cref{clm:crude} implies that $\WCchange_\cE(c)\leq O(d^* + d_1)$.

%
%

To apply (\ref{eq:partition}), we need to understand $\Pr[\cE]$ and
$\WCchange_{\cE}(c)$ for the following events (and
their intersections): the \emph{amenable event} $\cA$ and its
complement \emph{defiant event} $\cD$, the \emph{simple event} $\cS$
and its complement \emph{tree event} $\cT$. The simple event $\cS$ is
further partitioned into $\cS_1$ and $\cS_2$, and the tree event $\cT$
is partitioned into $\cT_1$ and $\cT_2$, representing whether
$f^*$ points to $\eta_1$ or $\eta_2$. These events are defined for
a fixed client $c$, and we omit $c$ in our notations for brevity.


Recall that $f^*$ is the optimal facility closest to $c$, and
$\rho = \rho(f^*) := \frac{d(f^*, \eta_1(f^*))}{d(f^*, \eta_2(f^*))}$.  To
generate the set $\cP$ of important swaps, we choose $\tau(f^*)$
from
different distributions depending on the value of $\rho(f^*)$, and
thus the probability of the events $\cS_1,\cS_2,\cT_1,\cT_2$ depends
on $\rho(f^*)$ as follows:

\begin{table}[h]
\begin{center}
\begin{tabular}{| c | c | c| c| c|}
\hline
    \rowcolor{black!90} \textcolor{white}{Ratio-types} 
    & \textcolor{white}{$\Pr[\cS_1]$} 
    & \textcolor{white}{$\Pr[\cS_2]$}
    & \textcolor{white}{$\Pr[\cT_1]$}
    & \textcolor{white}{$\Pr[\cT_2]$}\\ 
    \hline
    $0 \leq \rho(f^*) \leq \nicefrac{2}{3}$ & $\nicefrac{1}{2}$ & $\cdot$ & $\nicefrac{1}{2}$ & $\cdot$\\
    \hline
    $\nicefrac{2}{3} < \rho(f^*) \leq \nicefrac{3}{4}$ & $\nicefrac{1}{2}$ & $\cdot$ & $\nicefrac{1}{4}$
    		 & $\nicefrac{1}{4}$\\
    \hline
    $\nicefrac{3}{4}<\rho(f^*) \leq 1 $ & $\nicefrac{5}{4}-\rho$ & $\rho-\nicefrac{3}{4}$
    	 & $\nicefrac{1}{4}$ & $\nicefrac{1}{4}$\\
    \hline
\end{tabular}
\end{center}
\caption{Probability distribution for each ratio-type.}\label{tab:prob}
\end{table}

Since $\Pr[\cD] = O(\e)$ due to \Cref{clm:crude2}, the probability of
any event $\cE \cap \cA$ is at least $\Pr[\cE] - O(\e)$.

\medskip\textbf{Bounding the worst-case change $\WCchange_{\cE}(c)$.}
We fix an arbitrary swap set $\cP$ generated under event $\cE$, and
analyze the effect of each swap in $\cP$. Let $\move{f^*}$ denote the
swap in $\cP$ that opens $f^*$; such a swap always exists. There may
be multiple such swaps in $\cP$ when we perform tree swaps, in which case we let
$\move{f^*}$ be the swap that opens the original copy of $f^*$. For a
local facility $f\in\{f_1,f_2\}$, let $\move{\neg f}$ denote the swap
in $\cP$ that closes $f$. By implication~(i) of amenability, there is at most one such
swap as long as $\cE$ is a sub-event of the amenable event $\cA$. When there
is no swap closing $f$ (which happens when $f$ is a heavy facility), we are often in a better situation because our
bound for $\WCchange_{\move{\neg f}}(c)$ is often non-negative, so we
will mostly focus on the case where $\move{\neg f}$ does exist.

Before we begin giving bounds for the various client types, let us
record in 
Table~\ref{tab:ineqs} some inequalities we will frequently
use. Recall that $\eta_1(f^*)$ and $\eta_2(f^*)$ are the closest and
second-closest local facilities to $f^*$, and $\pi(f)$ is the closest
optimal facility to $f$.  These inequalities are proven in
\Cref{ap:proofUseful}.

\begin{table}[h]
\begin{center}
\begin{tabular}{| c l| c | }
\hline
    \rowcolor{black!90} \textcolor{white}{Bound}& &
                                                    \textcolor{white}{Conditions
                                                    (if any)} \\ 
    \hline
	$d_2 \leq 2d^* + d_1$ & $ \inlineeqnum\label{eqn:ub5}$ & $\eta_1(f^*) \neq f_1$\\
     \hline
    $d(c,\pi(f_1))\leq 2d_1+d^*$ & $\inlineeqnum\label{eqn:ubpartner}$&\\
     \hline
    $\max\{d(c,\eta_1(f^*)), d(c,\eta_2(f^*))\} \leq 2d^* + d_1$ & $\inlineeqnum\label{eqn:ubeta}$ & $\eta_1(f^*) \neq f_1$\\
      \hline
    $\min(d^*,d_1)+\b\max(d^*,d_1) \leq (1-\b)\, d^* + 2\b\, d_1$ & $\inlineeqnum\label{eqn:bestBalance}$  &\\
     \hline
\end{tabular}
\end{center}
\caption{Useful Inequalities}\label{tab:ineqs}
\end{table}


%
%

\subsection{Bounds for Clients of Type \xE}
\label{sec:type-z}

We now give an upper bound for the expected potential change $\change_\cA(c)$
for any client $c$ of type \xE. We give the entire proofs here; for
clients of other types we will defer the proofs to the appendices.

\begin{lemma}\label{lem:typez}
  For any client $c$ of type \xE, we have
   \[ \change_\cA(c) \leq 2.5\, d^*(c)  - 0.9\,
  d_1(c) + O(\varepsilon)(d^* + d_1). \]
\end{lemma}

In our proof, we partition the amenable event
$\cA = (\cS\cap\cA) \cup (\cT\cap\cA)$ depending on whether we have a
simple swap or a tree swap, and then bound $\change_\cA(c)$ by
\begin{align}
\change_\cA(c)\leq & \Pr[\cS\cap\cA]\cdot \WCchange_{\cS\cap\cA}(c) + \Pr[\cT\cap\cA]\cdot \WCchange_{\cT\cap\cA}(c)\nonumber \\
\leq & \Pr[\cS]\cdot\WCchange_{\cS\cap\cA}(c) + \Pr[\cT]\cdot\WCchange_{\cT\cap\cA}(c) + \epsd.\label{eq:E1}
\end{align}
The second inequality is implied by \Cref{clm:crude2} and our assumption that $\WCchange_\cE(c)\geq \WClb$.
To use~(\ref{eq:E1}) we give upper bounds for
$\WCchange_{\cS\cap\cA}(c)$ and $\WCchange_{\cT\cap\cA}(c)$ for
clients of both subtypes (close and far) in the next subsections. In
other words, we pick an arbitrary swap set $\cP$ generated under these
events, and bound the potential change for client $c$ due to the swaps
in $\cP$.


\subsubsection{Far Clients of Type \xE: \texorpdfstring{$d_2(c) \geq
    \a\,d_1(c)$}{[d2(c) >= alpha d1(c)]}}

\textbf{Simple Swaps.}
We fix a ``far'' client $c$ and an arbitrary swap set $\cP$ generated
conditioned on the event $\cS\cap\cA$ for this client, and bound the sum
$\sum_{(P,Q)\in\cP}\dPQ{c}$.

\begin{itemize}
\item Given the swap $\move{f^*} \in \cP$ (which is not $\move{\neg
    f_1}$ by implication (Siv) of amenability), 
   $c$ has an additional option of going to $f^*$,
  giving
  \[ \WCchange_{\move{f^*}}(c)\leq (d_1 + \b\, d^*) - (1+\a\b)\, d_1. \]
\item Next, by implication (i) of amenability, the set $\cP$ contains at most one
  swap $\move{\neg f_1}$. If $\move{\neg f_1}$ does
  exist, 
  both $\eta_1$ and $\eta_2$ are open (by implication (Siii) of amenability), and both at distance
  $\leq 2d^* + d_1$ from $c$. Therefore,
  \[ \WCchange_{\move{\neg f_1}}(c) \leq (1+\b)(2d^*+d_1) - (1+\a\b)\,
    d_1. \]
  This quantity is non-negative: since $c$ has type \xE, $\eta_1 \neq f$ and
also $d(c,\eta_1) \leq 2d^*+d_1$. But $c$ is a far client, then
$d(c,\eta_1) \geq \a d_1$. Putting the two together:
\begin{equation*}
(1+\b)(2d^*+d_1) - (1+\a\b)\, d_1\geq d_1 + \b (2d^* + d_1) - (1+\a\b)\, d_1\geq 0.
\end{equation*}
\item Finally, all other swaps in $\cP$ leave $f_1$ open, and thus
  they cannot increase the potential for $c$.
\end{itemize}
Combining these, when the swap $\move{\lnot f_1}$ exists,
\begin{equation}
\label{eq:E-far-simple1}
\sum_{(P,Q)\in\cP}\WCchange_{(P,Q)}(c) \leq \WCchange_{\move{f^*}}(c) + \WCchange_{\move{\neg f_1}}(c) \leq (2+3\b)\, d^* - (2\a\b-\b)\, d_1.
\end{equation}

In case $\move{\neg{f_1}}$ does not exist, (\ref{eq:E-far-simple1})
still holds since our bound for $\WCchange_{\move{\neg f_1}}$ is
non-negative. 
Since $\cP$ was a generic swap set conditioned on being amenable,
\begin{equation*}
\WCchange_{\cS\cap\cA}(c)\leq \hoono{(2+3\b)\, d^* - (2\a\b-\b)\, d_1}.
\end{equation*}

\medskip
\textbf{Tree Swaps.}
We now turn to \emph{tree swaps}, and fix an arbitrary swap set $\cP$
generated on the event $\cT\cap\cA$. Again, $\cP$ contains at most one swap $\move{\neg f_1}$
that closes $f_1$, by amenability. 
We first consider the case where $\move{\neg f_1}$ exists and is the
same as $\move{f^*}$. In this case, all other swaps in $\cP$ have non-positive potential changes, so
\begin{equation}
\label{eq:E-far-tree1}
\sum_{(P,Q)\in\cP}\WCchange_{(P,Q)}(c)\leq \WCchange_{\move{\neg f_1}}(c)\leq (1+\a\b)d^* - (1+\a\b)d_1.
\end{equation}
Next, consider the case where $\move{\neg f_1} \neq \move{f^*}$. On
swap $\move{f^*}$, client $c$ can go to both $d^*$ and $d_1$. On swap
$\move{\neg f_1}$, $c$ can go to $\pi(f_1)$ at distance $\leq 2d_1 + d^*$, and
also to $\tau(f^*) \in \{\eta_1,\eta_2\}$ at distance $\leq 2d^* +
d_1$. Both these facilities $\pi(f_1)$ and $\tau(f^*)$ must be open
after the swap $\move{\neg f_1}$ due to implications (ii) and (Sii) of amenabilityx. All other swaps in $\cP$ have non-positive potential changes, so 
\begin{align}
\sum_{(P,Q)\in\cP}\WCchange_{(P,Q)}(c) \leq {} & \WCchange_{\move{f^*}}(c) + \WCchange_{\move{\neg f_1}}(c)\nonumber \\
\leq {} & d_1 + \b\, d^* - (1+\a\b)\, d_1 \tag{$\WCchange_{\move{f^*}}$}\\
& +   (2d^*+d_1) + \b(2d_1+d^*) - (1+\a\b)\, d_1 \tag{$\WCchange_{\move{\neg f_1}}$}\\
= {} & (2+2\b)\, d^* - (2\a\b-2\b)\, d_1\label{eq:E-far-tree2}.
\end{align}
In the case where $\move{\neg f_1}$ doesn't exist,
(\ref{eq:E-far-tree2}) still holds, because our bound for
$\WCchange_{\move{\neg f_1}}(c)$ is non-negative. By our choice of $\a
= 3$
and $\b = \nicefrac15$, (\ref{eq:E-far-tree1}) is dominated by
(\ref{eq:E-far-tree2}). Since $\cP$ is a generic swap set,
\begin{equation*}
\WCchange_{\cT\cap\cA}(c)\leq \hoono{(2+2\b)\, d^* - (2\a\b-2\b)\, d_1}.
\end{equation*}
Summarizing the simple swaps case and the tree swaps case, we have

\begin{mybox}
\begin{align*}
\WCchange_{\cS\cap\cA}(c) \leq  & (2+3\b)\, d^* - (2\a\b-\b)\, d_1
	 \leq 2.6\, d^* - d_1,\\
\WCchange_{\cT\cap\cA}(c) \leq  & (2+2\b)\, d^* - (2\a\b-2\b)\, d_1
	 \leq 2.4\, d^* - 0.8\, d_1.
\end{align*}
\end{mybox}

Now substituting into (\ref{eq:E1}), we get a bound for all type \xE far clients $c$:
\begin{align}
\change_\cA(c) &\leq\haff \cdot \WCchange_{\cS\cap\cA}(c) + \haff
                 \cdot \WCchange_{\cT\cap\cA}(c) + \epsd \notag\\
 &\leq \hoono{2.5\, d^*  - 0.9\, d_1} + \epsd. \label{eq:Zfar}
\end{align}
This proves \Cref{lem:typez} for far clients of type \xE. 
The proof for all other types of clients will have a similar
structure: we will identify which swaps affect client $c$, then we sum
up the inequalities with the right probabilities. In some cases we
will need to look at cases depending on $\rho$.

\subsubsection{Close Clients of Type \xE: \texorpdfstring{$d_2(c) \geq
    \a\,d_1(c)$}{[d2(c) >= alpha d1(c)]}}

\textbf{Simple swaps.} Now we consider the case of close clients $c$. We fix an arbitrary swap set $\cP$, and
focus on $\move{f^*}$, $\move{\neg f_1}$, and $\move{\neg f_2}$ (All
other swaps cause a non-positive potential change).  Suppose these
three swaps are different.  When $f^*$ opens, the client $c$ can be
served by both $f^*$ and $f_1$.  When $f_1$ closes, $c$ can be served
by $f_2$ and $\eta_1$, and when $f_2$ closes, $c$ can be served by
$f_1$ and $\eta_1$: in both these cases, we use implication (Siii) of amenability
to ensure that both the corresponding facilities are open. We know that
$d_2 \leq d(c,\eta_1)$ because $c$ has type \xE; by \eqref{eqn:ubeta}
we get $d(c, \eta_1) \leq 2d^* + d_1$. Putting everything together,
the three swaps yield:
\begin{align*}
\WCchange_{\cS\cap\cA}(c) &\leq 
 d^* + \b\, d_1 - d_1 - \b\, d_2 \tag{$\WCchange_{\move{f^*} }$}\\
& + d_2 + \b(2d^*+d_1) - d_1 - \b\, d_2 \tag{$\WCchange_{\move{\neg f_1}}$}\\
& + d_1 + \b(2d^* + d_1) - d_1 - \b\, d_2 \tag{$\WCchange_{\move{\neg f_2}}$}\\
&= \hoono{(1+4\b)\, d^* - (2-3\b)\, d_1 + (1-3\b)\, d_2}.
\end{align*}

We address the assumption that the three swaps are different. As
argued above, condition (Siv) of amenability for type~\xE clients means that for simple swaps,
$\move{\neg f_1} \neq \move{f^*}$. However, $f_2$ could be $\tau(f^*)$, so it may happen that 
$\move{\neg f_2} = \move{f^*}$, and hence that 
$\WCchange_{\cS\cap\cA}(c) \leq \WCchange_{\move{f^*}} +
\WCchange_{\move{\neg f_1}}$. Moreover, $\move{\lnot f_1}$ may not
exist, in which case $\WCchange_{\cS\cap\cA}(c) \leq \WCchange_{\move{f^*}} +
\WCchange_{\move{\neg f_2}}$ or even $\WCchange_{\cS\cap\cA}(c) \leq \WCchange_{\move{f^*}}$. But since our bounds above
for both $\WCchange_{\move{\neg f_1}}$  and $\WCchange_{\move{\neg f_2}}$
are non-negative, we infer that 
the boxed upper bound
remains
valid in all these cases.

\medskip \textbf{Tree swaps.} We now consider tree swaps. Fix an
arbitrary swap set $\cP$ generated on the event $\cT\cap \cA$.  For a
client $c$ in the close case, there are three swaps that are relevant
to $c$---those containing $f^*$, $f_1$, and $f_2$---although some of these
swaps may coincide. (Also, no other swaps can increase the potential.)

\paragraph*{When $f_1$ and $f_2$ belong to the same swap.}
First suppose that $f_1$ and $f_2$ belong to the same swap in $\cP$. We start from the case where $\move{f^*}\neq \move{\neg f_1,\neg f_2}$. For the swap $\move{f^*}$, the client $c$ can be served by both $f^*$
and $f_1$. And when $f_1$ and $f_2$ are both closed, $c$ can be served
by $\tau(f^*)$ (which is either $\eta_1$ or $\eta_2$) and
$\pi(f_1)$. By \eqref{eqn:ubeta} we get that $d(c,\tau(f^*))$ is at most
$2d^* +d_1$, and by \eqref{eqn:ubpartner} we get $d(c,\pi(f_1))\leq
2d_1+d^*$. Hence,
\begin{align*}
\sum_{(P,Q)\in\cP}\WCchange_{(P,Q)}(c) &\leq  
 \WCchange_{\move{f^*}}  + \WCchange_{\move{\neg f_1, \neg f_2}}\\
&\leq  (d^* + \b\, d_1) - (d_1 + \b\, d_2) \tag{$\WCchange_{\move{f^*}}$}\\
& ~~~~ +  (2d^*+d_1) + \b(2d_1 + d^*) - (d_1 + \b\, d_2) \tag{$\WCchange_{\move{\neg f_1,\neg f_2}}$}\\
&= \hoono{(3+\b)\, d^* - (1-3\b)\, d_1 - 2\b\, d_2}. \nonumber
\end{align*}
On the other hand, if $f^*, f_1$, and $f_2$ all belong to the same
swap, we can assign $c$ to $f^*$
\begin{align*}
\sum_{(P,Q)\in\cP}\WCchange_{(P,Q)}(c) \leq {} &  \WCchange_{\move{f^*,
                   \neg f_1, \neg f_2}} \\ 
\leq {} &  (1+\a\b)\, d^* - d_1 - \b\, d_2 \tag{$\WCchange_{\move{f^*,\neg f_1, \neg f_2}}$}\\
& + \b(2d^*+d_1 - d_2)\tag{since $2d^*+d_1 \geq d_2$}\\
= {} & (1+\a\b+2\b)\, d^* - (1-\b)\, d_1 -2\b\, d_2\\
\leq {} & 3.2\, d^* - 0.4\, d_1 - 0.4\, d_2.
\end{align*}
These two bounds are identical for our choices of $\a = 3$ and
$\b = \nicefrac15$.

\paragraph*{When $f_1$ and $f_2$ belong to different swaps.}
Next, consider the case when $f_1$ and $f_2$ belong to
different swaps in $\cP$. Let us first assume $\move{f^*}$ is neither $\move{\neg f_1}$ nor $\move{f_2}$. In the swap $\move{f^*}$ the client can served by
$f^*$ and $f_1$. When one of $f_1$ or $f_2$ is closed, the client $c$
can be served by the other facility, and by $\tau(f^*)$, which is at
distance at most $2d^*+d_1$ from $c$ (by \eqref{eqn:ubpartner}). Hence,
\begin{align*}
\sum_{(P,Q) \in \cP} \dPQ{c} &\leq (1-\beta)\, d^* + 2\b\, d_1 - d_1 - \b\, d_2 \tag{$\WCchange_{\move{f^*}}$ with \eqref{eqn:bestBalance}}\\
& ~~~+ d_2 + \b(2d^*+d_1) - d_1 - \b\, d_2 \tag{$\WCchange_{\move{\neg f_1}}$}\\
& ~~~+ d_1 + \b(2d^* + d_1) - d_1 - \b\, d_2 \tag{$\WCchange_{\move{\neg f_2}}$}\\
  &= \hoono{(1+3\b)\, d^* - (2-4\b)\, d_1 + (1-3\b)\, d_2}.
\end{align*}

Our bound for $\WCchange_\move{f^*}$ does not require $f_2$ to remain open after the swap, and our bound for $\WCchange_\move{\neg f_2}$ is non-negative. Therefore, the above bound also holds when $\move{\neg f_2} = \move{f^*}$. When $\move{\neg f_1} = \move{f^*}$, we still have the above bound:
\begin{align*}
\sum_{(P,Q) \in \cP} \dPQ{c} &\leq d^* + \b\, d_2 - d_1 - \b\, d_2 \tag{$\WCchange_{\move{f^*,\neg f_1}}$}\\
& ~~~+ d_1 + \b(2d^* + d_1) - d_1 - \b\, d_2 \tag{$\WCchange_{\move{\neg f_2}}$}\\
& ~~~+ \b\, d^* + \b\, d_1 + (1-2\b)(d_2 - d_1) \tag{non-negative terms}\\
  &= \hoono{(1+3\b)\, d^* - (2-4\b)\, d_1 + (1-3\b)\, d_2}.
\end{align*}

Summarizing all these
bounds (using that $\a = 3$ and $\b = 0.2$), 
\begin{mybox}
\begin{align*}
\WCchange_{\cS \cap \cA}(c)  &\leq (1+4\b)\, d^* - (2-3\b)\, d_1 + (1-3\b)\, d_2
	= 1.8\, d^* - 1.4\, d_1 + 0.4\, d_2\\
\WCchange_{\cT \cap \cA}(c) &\leq \max\{3.2\, d^* - 0.4\, d_1 - 0.4\,
  d_2, 1.6\, d^* - 1.2\, d_1 + 0.4\, d_2 \}.
\end{align*}
\end{mybox}
Combining and using \eqref{eqn:ub5} to get $d_2 \leq 2d^* +
d_1$ if the $d_2$ terms do not cancel out, we get for close clients $c$:
\begin{align}
\change_{\cA}(c) &\leq  \haff \cdot \WCchange_{\cS \cap \cA} +
                   \haff\cdot\WCchange_{\cT \cap \cA} + \epsd \notag \\ 
&\leq \hoono{2.5\, d^* - 0.9\, d_1} +\epsd. \label{eq:Zclose}
\end{align}


\Cref{lem:typez} follows from the bound in \eqref{eq:Zfar} for the far clients and the
one from \eqref{eq:Zclose} for the close clients.

\subsection{All Other Client Types}
\label{sec:all-other-client}

Similarly, we can bound $\Gsum{c}$ for every other client type
\xA--\xD. We summarize this in the following theorem: the calculations
behind the expressions can be found in \Cref{apx:bounds}.

\begin{restatable}[]{lemma}{boundsAll}
  \label{lem:boundsAll}
  For any far client $c$ of type \xA or \xB, we have
  \begin{align}
    \Gsum{c}\leq 2.47\, d^*(c) - 1.13\, d_1(c) + \epsd\label{eqn:bndFar}
  \end{align}
  
  For any close  client $c_i$ of type $i\in\{\xA,\xB,\xC,\xD\}$, we have
  \begin{align}
    \Gsum{c_{\xA}} {} & \leq 2.375\, d^*(c_\xA) - 0.9\, d_1(c_\xA)  + \epsd \label{eqn:bndA}\\
    \Gsum{c_{\xB}} {} & \leq 2.4\, d^*(c_\xB) - 0.9\, d_1(c_\xB)  + \epsd\label{eqn:bndB}\\    
    \Gsum{c_{\xC}} {} & \leq 2.2\, d^*(c_\xC) - 0.8888\, d_1(c_\xC)  + \epsd\label{eqn:bndC}\\    
    \Gsum{c_{\xD}} {} & \leq 2.5203\, d^*(c_\xD) - 0.8888\, d_1(c_\xD)  + \epsd\label{eqn:bndD}
  \end{align}
\end{restatable}

\Cref{lem:typez,lem:boundsAll} imply that every client $c$ satisfies
\[ \Gsum{c} \leq 2.5203\, d^*(c) - 0.8888\, d_1(c)  +\epsd. \] This proves 
\Cref{lem:combo2}, and hence \Cref{lm:main} and Theorem~\ref{thm:ub}.




 \newcommand{\cR}{\mathcal{R}}

\section{A Computer-Aided Analysis using Linear Programming}\label{sec:LP}

In this section we show how to generate a set of valid inequalities,
then solve the resulting linear program to find an upper bound on our
approximation ratio. We describe the ideas for the potential $\Phi_2$
that only takes the second-closest facility into account, and indicate
how to extend it to $\Phi_q$ for higher values of $q \geq 2$. Of course, the
size of the LP increases exponentially as $q$ increases. 

To recall, our proof strategy in the previous section was to consider
a local optimum, and then:
\begin{enumerate}[noitemsep]
\item define a (randomized) collection of important swaps that are contained
  within our actual set of swaps;
\item for every client type, write constraints that apply to all clients of that type;
\item carefully combine those constraints to have only a few remaining
  constraints; and
\item manually check these remaining contraints.
\end{enumerate}

An automated proof could avoid the last two steps by directly checking
the entire set of constraints. Since every constraint we derive is a
linear inequality on the distances, a linear program can be used for
this automated proof.
Put differently, our goal is to write a linear program that constructs
a ``worst-case example'' for our potential function. Specifically, the
program seeks values of the distances $d_1, d_2$, and $d^*$ for each
client type, so as to maximize the ratio between the costs of the
optimum and local solutions, while respecting the set of constraints.
\footnote{In fact, it does not come up with a concrete example, since we do not
maintain all the triangle inequalities between the clients, but only
the triangle inequalities in some local neighborhood around each
client. It is conceivable that using more triangle inequalities would
lead to an even better result, but that increases the complexity even further.}

\paragraph*{Variables and constraints of the LP.}
Let us focus on simple swaps, the constraints for tree swaps are
similar.  We want to express the fact that simple swaps at a local
optimum do not decrease the potential. We first classify facilities into
types according to their ratio $\rho$; we consider only a fine net of
values for $\rho$, and use continuity of the potential to control the
loss due to this discretization. All facilities with a given ratio are
treated the same way in the proof: our LP considers that all facilities
of the same type are swapped at the same time. Specifically, we have
a variable $s_\rho$ corresponding to the difference in the potential
function after applying simple swaps for all facilities with ratio
$\rho$. The constraint saying that simple swaps do not decrease the
potential is therefore $\sum_\rho s_\rho \geq 0$.

The value of the variable $s_\rho$ is controlled by the clients
connected to facilities having ratio $\rho$: each client type $i$ has a
contribution to it. In an $S_j$-swap (where $j \in \{1, 2\}$), let
$\delta_{S_j}(i; x, y^1, y^2)$ be the potential change due to all clients of
type-$i$ connected to facilities with ratio $\rho$, in function of $x, y^1$ and $y^2$, respectively the total distance from those clients to the optimal solution, their closest and their second closest facility of the local solution. This difference of
potential is described in Section~\ref{sec:bounds}: we illustrate it
with clients of type \xA, in the far case. We denote \xAF those
clients.  As presented in \ref{subsec:farA}, the $S_1$ and
$S_2$ swaps for those clients show
\begin{align*}
\dSwapHa \leq  &(1+\a\b)\, d^* - (1+\a\b)\, d_1\\
\dSwapHb \leq  &((1+\nicefrac{1}{\rho})(1+\a\b)+\b)\, d^* -((1-\nicefrac{1}{\rho})(1+\a\b)+\a\b)\, d_1
\end{align*}

For bounding $\dSwapHb$, we upper bounded the potential value
of the swap by $d^*+\b d_1$ when $f^*$ is opened. However, we could be more precise:  
it could be the case that $d^* \leq d_1$ or $d_1 \leq \a d^*$. Therefore, this lead to 3 other possible upperbounds, namely $(1+\a\b)d^*$, $d^* + \b d_1$, and $(1+\a\b)d_1$.

This translates to three other inequalities, one for each of those cases:
\begin{align*}
\dSwapHb & \leq
 ((2+\nicefrac{1}{\rho})(1+\a\b))\, d^* -((2-\nicefrac{1}{\rho})(1+\a\b))\, d_1 \tag{when we choose $(1+\a\b)d^*$}\\
\dSwapHb &  \leq
((1+\nicefrac{1}{\rho})(1+\a\b)+1)\, d^* -((2-\nicefrac{1}{\rho})(1+\a\b)-\b)\, d_1 \tag{when we choose $d^* + \b d_1$}\\
\dSwapHb &\leq 
((1+\nicefrac{1}{\rho})(1+\a\b))\, d^* -((1-\nicefrac{1}{\rho})(1+\a\b))\, d_1
\tag{when we choose $(1+\a\b)d_1$}.
\end{align*}
More generally, the LP encodes all possible combinations of variables giving valid bound on the potential after a swap. Note that then number of such inequalities grows exponentially with $q$, because each term $\min(\alpha_j d_1(c), d_j(c))$ doubles the number of valid inequalities. 

Going back to type $\xA$, this gives rise to the constraints 
\begin{align}
\delta_{S_1}(\xAF; x_{\xAF, \rho}, y_{\xAF, \rho}^1, y_{\xAF, \rho}^2)
  &\leq (1+\a\b)x_{\xAF, \rho} - (1+\a\b)y_{\xAF, \rho}^1 \label{ex:uppot1}\\
  \delta_{S_2}(\xAF; x_{\xAF, \rho}, y_{\xAF, \rho}^1, y_{\xAF, \rho}^2) &\leq ((1+\nicefrac{1}{\rho})(1+\a\b)+\b)\, x_{\xAF, \rho} -((1-\nicefrac{1}{\rho})(1+\a\b)+\a\b)\, y_{\xAF, \rho}^1\\
  \delta_{S_2}(\xAF; x_{\xAF, \rho}, y_{\xAF, \rho}^1, y_{\xAF, \rho}^2) &\leq ((2+\nicefrac{1}{\rho})(1+\a\b))\, x_{\xAF, \rho} -((2-\nicefrac{1}{\rho})(1+\a\b))\, y_{\xAF, \rho}^1\\
  \delta_{S_2}(\xAF; x_{\xAF, \rho}, y_{\xAF, \rho}^1, y_{\xAF, \rho}^2) &
  \leq ((1+\nicefrac{1}{\rho})(1+\a\b)+1)\, x_{\xAF, \rho} -((2-\nicefrac{1}{\rho})(1+\a\b)-\b)\, y_{\xAF, \rho}^1\\
  \delta_{S_2}(\xAF; x_{\xAF, \rho}, y_{\xAF, \rho}^1, y_{\xAF, \rho}^2) 
  &\leq ((1+\nicefrac{1}{\rho})(1+\a\b))\, x_{\xAF, \rho} -((1-\nicefrac{1}{\rho})(1+\a\b))\, y_{\xAF, \rho}^1, \label{ex:uppot2}
\end{align}

where the variables $x_{i, \rho}$ denote the total cost of clients of
type $i$ connected to facilities with ratio $\rho$ in the optimal
solution, and $y_{i, \rho}^j$ denote the total distance from those
clients to their $j$-th closest facility in the local solution. 
This definition of $\delta_{S_j}(i)$ yields the following constraint on $s_\rho$ : for $j = 1, 2$, 
\[ \sum_{i\in \clientT}\delta_{S_j}(i; x_{i, \rho},y_{i, \rho}^1,
  y_{i, \rho}^2) \geq s_{\rho}, \] where $\clientT$ is the set of
client types.

Moreover, the triangle inequality gives constraints on the variables
$x_{i, \rho},y_{i, \rho}^1, y_{i, \rho}^2$: for instance, for type
\xAF, we would have
\begin{align}
 y_{\xAF, \rho}^2 \leq x_{\xAF, \rho} +( \nicefrac{1}{\rho}) (\leq x_{\xAF, \rho}+y_{\xAF, \rho}^1). \label{ex:trangleEx}
\end{align}

The constraints due to tree swaps are defined analogously, with a
variable $t_\rho$ being the potential change after applying tree swaps
for all facilities of type $\rho$, and $\delta_{T_j}(i)$ being the
potential change due to all clients of type-$i$ connected to facilities
with ratio $\rho$. 
For $q>2$, we need to consider more than one ratio, so we let $\rho$ be
the vector of size $q-1$ that describes ratio of all two consecutive $\eta_j$ and $\eta_{j+1}$ for all $j\in \{1,...,q-1\}$. 
Let $\cR$ be the set of values of $\rho$ after
discretization: we use $\cR := \{\frac{i}{100} \mid i \in \{0, \ldots,
100\} \}^{q-1}$. In that case, all clients with ratio in $[i \cdot 10^{-2}, (i+1)\cdot 10^{-2})$ are considered to have $\rho = i \cdot 10^{-2}$ for each index.
 This means that our bounds for $\delta$ are  slightly relaxed to cover an interval instead of a precise $\rho$.  
 The $j^{th}$ index of a ratio correspond to $\frac{d(f^*(c),\eta_{j}(f^*(c)))}{d(f^*(c),\eta_{j-1}(f^*(c)))}$.
Let $\cC$ be the set of client types. For $q=2$, each $i\in \cC$ contains
client type $(\xA, \xB, \xC, \xD, \xE)$, ratio $\rho$, underlying form of tree-graph (e.g. $f_1$ and $f_2$ belong to same tree in $\eta_1$-swap),
and whether $d_2 \leq \a d_1$ or not.
For $j\in \{2,...,q\}$, let $\cC_j$ be the set of clients with $d_j \leq \a_j d_1$.
For $q=2$, $\cC_2$ is the set of close clients (i.e., $d_2 \leq \a d_1$).
Let $C^{\rho}$ denote set of clients with ratio $\rho$.
The general structure of the LP is the following:

\begin{align}
\max & \sum_{i\in \cC} y^1_i\\
\text{s.t. }& \sum_{i\in \cC} x_i = 1\\
& y^j_i \leq \a_j y^{1}_i & \forall j\in\{2,\ldots,q\}, i\in\cC_{j} \label{cons:close}\\
& y^j_i \geq \a_j y^{1}_i & \forall j\in\{2,\ldots,q\}, i\not\in\cC_{j} \label{cons:far}\\
& \sum_{i\in\cC^\rho}\delta_{S_j}(i;x_i,y_i^1,\ldots,y_i^q) \geq s_{\rho} & \forall \rho\in \cR, j\in \{1,\ldots,q\} \label{cons:simplecons}\\
& \sum_{i\in\cC^\rho}\delta_{T_j}(i;x_i,y_i^1,\ldots,y_i^q) \geq t_{\rho} & \forall \rho\in \cR, j\in \{1,\ldots,q\} \label{cons:treecons}\\
& \sum_{\rho\in \cR} s_\rho \geq 0 \label{cons:simpleSum}\\
& \sum_{\rho\in \cR} t_\rho \geq 0 \label{cons:treeSum}\\
& \text{Triangle-inequalities} \label{cons:triangle}\\
& \delta_{S_j}(i;x_i,y_i^1,\ldots,y_i^q)  \leq \text{enumerated-upperbounds} & \forall i\in \cC, j\in \{2,...,q\}\label{cons:simpleBestof}\\
& \delta_{T_j}(i;x_i,y_i^1,\ldots,y_i^q)  \leq \text{enumerated-upperbounds} & \forall i\in \cC, j\in \{2,...,q\}\label{cons:treeBestof}\\
& y^j_i \geq 0,\quad x_i \geq 0 & \forall i\in \cC, j\in \{1,\ldots,q\}
\end{align}

Note that $\sum_{i\in \cC} y^1_i / \sum_{i\in\cC}x_i = \sum_{i\in\cC} y^1_i$ is the locality gap.
Constraints \eqref{cons:close} and \eqref{cons:far} restrict each distance based on whether they are `far' client or `close' client.
Constraints \eqref{cons:simplecons} can be seen as the following: for each ratio $\rho$ 
we pick $j\in \{1,...,q\}$ that minimizes the sum of potential difference 
after performing $S_j$ swap, then make $\tau(f^*) = \eta_j$ for all $f^*$ with ratio $\rho$.  Similarly, \eqref{cons:treecons} chooses $\tau(\cdot)$ 
for tree swaps. Then \eqref{cons:simpleSum} and \eqref{cons:treeSum} ensure
the potential difference is non-negative after performing simple swap and tree swap respectively.
We also add triangle inequalities (e.g., \eqref{ex:trangleEx}). Lastly, we add upperbounds
for each potential difference in
\eqref{cons:simpleBestof} and \eqref{cons:treeBestof} (e.g., \eqref{ex:uppot1} - \eqref{ex:uppot2}).

Implementing this approach, and then
solving the resulting LP for for potential $\Phi_2$ and $\Phi_3$ gives
us the following numbers:
\begin{center}
\begin{tabular}{ | c | c | }
\hline
    \rowcolor{black!90} \textcolor{white}{Potential} & \textcolor{white}{Bound} \\ 
    \hline
     $\Phi_2$ & $2.7786$\\
     \hline
     $\Phi_3$ & $2.6861$\\
     \hline
\end{tabular}
\end{center}
For $\Phi_2$, the LP finds that taking $\a=3, \b=0.2$ yields the best result, whereas for
$\Phi_3$ we set manually 
$\a=2.5, \b=0.3, \b_2 = \b\cdot 0.34$.
As always, we get an additive $\eps$ term because of the defiant
swaps. However, let us emphasize that these implementations should be considered
preliminary, since they have not been formally verified. We hope that 
formal proofs of these results can be given in the near future.


\appendix

\newcommand{\tinygap}{\e'}
\section{Locality Gap for Potential \texorpdfstring{$\Phi_2$}{[Phi2]}}
\label{sec:lb}
\label{sec:phi2lb}

In this section, we give lower bounds on the locality gap, and prove
Theorem~\ref{thm:lb}. We show locality gap examples of $\max\{2,\a\}$,
$3-2\b$, and $1+4\b$ for the potential function $\Phi_2$. Putting these
together, the locality gap is
$\min_{\b \in [0,1],\a\in[1,\infty)} \max\{3-2\b, 1+4\b\},\max\{2,\a\}$. 
Note $\max\{3-2\b,1+4\b\}$ is $2\frac13$, when
we set $\b = \nicefrac{1}{3}$.  Therefore we show a locality gap of $2$.


In this section, we show a locality gap of $2$ for $\Phi_2$.
We divide the cases into three main cases:
\begin{OneLiners}
\item When $\a \leq 2$
\item When $\a > 2$ and $\b \leq 1/3$
\item When $\a >2$ and $\b > 1/3$
\end{OneLiners}

We mainly use two types of example that we call
``bi-clique'' and ``double-bi-clique'' described in \Cref{fig:LB3-2b} \Cref{fig:LB1+4b} respectively.
In bi-clique we have $k+\extras$ local facilities on the right, where  $\extras=O(1)$ is the number of extra local facilities, and
$k$ optimal facilities are on the left.  There is a client between
every (local, optimal) facility pair, at unit distance from the
optimal facility, and at distance $d$ from the local
facility. 

\begin{figure}[H]
\center
\begin{tikzpicture}[scale=0.7, 
		opt/.style={shape=regular polygon,regular polygon sides = 3,draw=black,minimum size=0.5cm,inner sep = 0pt},
		local/.style={shape=rectangle,draw=black,minimum size=0.4cm},
		client/.style={shape=circle,draw=black,minimum size=0.1cm,inner sep = 0pt,fill=black}]
		
	\pgfmathtruncatemacro{\N}{2}
	\pgfmathtruncatemacro{\Nm}{\N-1}	
	
	\pgfmathsetmacro{\LEFT}{0}
	\pgfmathsetmacro{\RIGHT}{\LEFT+5}
	\pgfmathsetmacro{\CLIENT}{\LEFT+1.25}

	\foreach \x in {0,...,\N}{
		\ifthenelse{\x=1 \OR \x=2}{
		\node [opt]  (opt\x) at (\LEFT,\x) {};
		\node[local] (loc\x) at (\RIGHT,\x) {};
		}{
		\node [opt]  (opt\x) at (\LEFT,\x) {};
		\node[local] (loc\x) at (\RIGHT,\x) {};
		}
		\draw (loc\x) edge[-] (opt\x);
		\node[client] (c\x\x) at (\CLIENT,\x){};
	}

	\foreach \x in {0,...,\Nm}
		{\pgfmathtruncatemacro{\label}{\x+1}
		\foreach \y in {\label,...,\N}{
			\draw (loc\x) edge[-]  (opt\y);	
			\draw (loc\y) edge[-]  (opt\x);
			\pgfmathsetmacro{\f}{\x+0.25*(\y-\x)}
			\node[client] (c\x\y) at(\CLIENT,\f) {};
			\pgfmathsetmacro{\g}{\y+0.25*(\x-\y)}
			\node[client] (c\y\x) at (\CLIENT,\g) {};
		}	
	}
	
	\node[opt] (optZ) at (\LEFT,-2) {};
	\node[local] (locZ) at (\RIGHT,-2) {};	
	\draw (locZ) edge[] (optZ) ;
	\node[client] (cZZ) at (\CLIENT,-2){};
	
	\foreach \x in {0,...,\N}{
		\draw (locZ) edge[-] (opt\x);
		\draw (loc\x) edge[-]  (optZ);
		\pgfmathsetmacro{\f}{-1.5+0.25*\x}
		\node[client] (cZ\x) at (\CLIENT,\f){};
		\pgfmathsetmacro{\g}{\x-0.25*(\x+2)}
		\node[client] (cY\x) at (\CLIENT,\g){};
	}

	\node () at (\LEFT,-1) {$\vdots$};
	\node () at (\RIGHT,-1) {$\vdots$};

	\pgfmathsetmacro{\lCLIENT}{\CLIENT-0.02}
	\pgfmathsetmacro{\rCLIENT}{\CLIENT+0.02}
		
	\draw[decoration={brace,mirror,raise=5pt},decorate] 
		(\LEFT,-2.1) -- node[below=7pt] {$1$} (\lCLIENT,-2.1);
	\draw[decoration={brace,mirror,raise=5pt},decorate] 
		(\rCLIENT,-2.1) -- node[below=7pt] {$d$} (\RIGHT,-2.1);
\end{tikzpicture}

\caption{An illustration of the bi-clique example with $\extras=0$.} 
\label{fig:LB3-2b}
\end{figure}

In double-bi-clique, we have two back-to-back bi-cliques as in \Cref{fig:LB1+4b}.
Each bi-clique is constructed the same way
as Figure~\ref{fig:LB3-2b} except the number of facilities are halved.
Consider a client $c$ with an edge going into $f_i$,
create an edge at distance $d$ between $c$ and $i^{th}$ local facility in the 
other bi-clique. 
Now every client has an optimal facility at distance 1,
and two local facilities at distance $d$.

\begin{figure}[H]
\center
\begin{tikzpicture}[scale=1.2, 
		opt/.style={shape=regular polygon,regular polygon sides = 3,draw=black,minimum size=0.5cm,inner sep = 0pt},
		local/.style={shape=rectangle,draw=black,minimum size=0.4cm},
		client/.style={shape=circle,draw=black,minimum size=0.2cm,inner sep = 0pt,fill=black}]
		\pgfmathtruncatemacro{\N}{1}
	\pgfmathtruncatemacro{\Nm}{\N-1}	
	
	\pgfmathsetmacro{\OLEFT}{1}
	\pgfmathsetmacro{\ORIGHT}{\OLEFT+3}
	\pgfmathsetmacro{\OCLIENT}{\OLEFT+(\ORIGHT-\OLEFT)/4}

	\foreach \x in {0,...,\N}{
		\node [opt]  (opt\x) at (\OLEFT,\x) {};
		\node[local] (loc\x) at (\ORIGHT,\x) {};
		\draw (loc\x) edge[-] (opt\x);
		\node[client] (c\x\x) at (\OCLIENT,\x){};
	}

	\foreach \x in {0,...,\Nm}
		{\pgfmathtruncatemacro{\label}{\x+1}
		\foreach \y in {\label,...,\N}{
			\draw (loc\x) edge[-]  (opt\y);	
			\draw (loc\y) edge[-]  (opt\x);
			\pgfmathsetmacro{\f}{\x+0.25*(\y-\x)}
			\node[client] (c\x\y) at(\OCLIENT,\f) {};
			\pgfmathsetmacro{\g}{\y+0.25*(\x-\y)}
			\node[client] (c\y\x) at (\OCLIENT,\g) {};
		}	
	}
	
	


	\pgfmathsetmacro{\lCLIENT}{\OCLIENT-0.02}
	\pgfmathsetmacro{\rCLIENT}{\OCLIENT+0.02}
	\pgfmathsetmacro{\BRACEHEIGHT}{-0.1}
		
	\draw[decoration={brace,mirror,raise=5pt},decorate] 
		(\OLEFT,\BRACEHEIGHT) -- node[below=7pt] {$1$} (\lCLIENT,\BRACEHEIGHT);
	\draw[decoration={brace,mirror,raise=5pt},decorate] 
		(\rCLIENT,\BRACEHEIGHT) -- node[below=7pt] {$d$} (\ORIGHT,\BRACEHEIGHT);
	
	
	\pgfmathtruncatemacro{\N}{1}
	\pgfmathtruncatemacro{\Nm}{\N-1}	
	
	\pgfmathsetmacro{\LEFT}{4.2}
	\pgfmathsetmacro{\RIGHT}{\LEFT+3}
	\pgfmathsetmacro{\CLIENT}{\RIGHT-(\RIGHT-\LEFT)/4}
	\pgfmathsetmacro{\ABOVE}{0.5}
	
	\foreach \tx in {0,...,\N}{
		\pgfmathsetmacro{\x}{\tx+\ABOVE}
		\node [opt]  (2opt\tx) at (\RIGHT,\x) {};
		\node[local] (2loc\tx) at (\LEFT,\x) {};
		\draw (2loc\tx) edge[-] (2opt\tx);
	}
	\foreach \tx in {0,...,\N}
		{\pgfmathsetmacro{\x}{\tx+\ABOVE}
		\node[client] (2c\tx\tx) at (\CLIENT,\x){};
	}
	
	\foreach \tx in {0,...,\Nm}
		{\pgfmathsetmacro{\x}{\tx+\ABOVE}
		\pgfmathtruncatemacro{\label}{\tx+1}
		\foreach \ty in {\label,...,\N}{
			\pgfmathsetmacro{\y}{\ty+\ABOVE}
			\draw (2loc\tx) edge[-]  (2opt\ty);	
			\draw (2loc\ty) edge[-]  (2opt\tx);
			\pgfmathsetmacro{\f}{\x+0.25*(\y-\x)}
			\node[client] (2c\tx\ty) at(\CLIENT,\f) {};
			\pgfmathsetmacro{\g}{\y+0.25*(\x-\y)}
			\node[client] (2c\ty\tx) at (\CLIENT,\g) {};
		}	
	}
	


	\pgfmathsetmacro{\HEIGHT}{-1+\ABOVE}
	
	
	\pgfmathsetmacro{\lCLIENT}{\CLIENT-0.02}
	\pgfmathsetmacro{\rCLIENT}{\CLIENT+0.02}
	\pgfmathsetmacro{\HEIGHT}{\N+\ABOVE+0.1}
	\pgfmathsetmacro{\HEIGHTT}{\N+0.1}	
	\pgfmathsetmacro{\LEFTT}{\LEFT-0.05}		
	\pgfmathsetmacro{\ORIGHTT}{\ORIGHT+0.05}		
	\pgfmathsetmacro{\BRACEHEIGHTT}{\BRACEHEIGHT+\ABOVE}

	\draw[decoration={brace,raise=5pt},decorate] 
		(\LEFT,\HEIGHT) -- node[above=7pt] {$d$} (\lCLIENT,\HEIGHT);
	\draw[decoration={brace,raise=5pt},decorate] 
		(\rCLIENT,\HEIGHT) -- node[above=7pt] {$1$} (\RIGHT,\HEIGHT);

	\draw[decoration={brace,raise=5pt},decorate] 
		(\OCLIENT,\HEIGHTT) -- node[above=7pt] {$d$} (\LEFTT,\HEIGHT);
		
	\draw[decoration={brace,raise=5pt,mirror},decorate] 
		(\ORIGHTT,\BRACEHEIGHT) -- node[below=9pt] {$d$} (\CLIENT,\BRACEHEIGHTT);

	\foreach \x in {0,...,\N}{
		\foreach \y in {0,...,\N}{
			\draw (c\x\y) edge[dashed,gray] (2loc\y);	
			\draw (2c\x\y) edge[dashed,gray] (loc\y);	
		}
	}
\end{tikzpicture}
\caption{The double-bi-clique example with $k=4$ and $\extras=0$.}
\label{fig:LB1+4b}
\end{figure}

For all cases we calculate the potential difference
after performing a swap of size $p$.
There are mainly 4 different types of clients.
\begin{OneLiners}
\item $\cC_o$: the set of clients with their $f^*$ opened.
\item $\cC_1$: Clients with their $f_1$ closed and  $f^*$ not opened 
\item $\cC_2$: Clients with its $f_2$ closed and  $f^*$ not opened 
\item $\cC_3$: Clients with its $f^*$ closed, $f_1$ opened, and $f_2$ opened.
\end{OneLiners}

We use $c_o$, $c_1$, $c_2$, $c_3$ to denote generic
client for sets $\cC_o$, $\cC_1$, $\cC_2$, and $\cC_3$
respectively. We first calculate potential difference
for each client type, then sum them over.

We assume there is no client with their $f_1$ and $f_2$ both closed:
those clients can only hurt the quality of the
solution, and given a swap that closes $f_1$ and $f_2$ of some clients it is easy to construct a strictly better set of swaps with no such client. 

\subsection{When \texorpdfstring{$\a \leq 2$}{[alpha <= 2]}}
\label{sec:smalla-example}
We first give lower bound examples when $\a \leq 2$. 
We divide the case further into two cases:
when $\a \leq 4/3 + \nicefrac{1}{(3\b)}$ and when $\a > 4/3 + \nicefrac{1}{(3\b)}$.

\paragraph*{Subcase I: $\a \leq 4/3 + \nicefrac{1}{(3\b)}$.}
We create a bi-clique presented in \Cref{fig:LB3-2b} with $d=
2-\tinygap$, where $\tinygap \approx O(1/k)$ is a small quantity to be specified later.
Note that every client has $k+\extras-1$ local facilities at distance $2+d$,
thus the second closest facility is never closed for any client.
Then for each client we get the potential differences:
\begin{align*}
\Delta\Phi^{c_o} & =  1+\a\b - (2-\tinygap) - (2-\tinygap)\a\b = -1-\a\b+\tinygap + \tinygap\a\b\\
\Delta\Phi^{c_1} & = (4-\tinygap)+(4-\tinygap)\b - (2-\tinygap) - (2-\tinygap)\a\b \geq 2+4\b - 2\a\b\\
\Delta\Phi^{c_2} & = 0\\
\Delta\Phi^{c_3} & = 0
\end{align*}

Note $|\cC_o| = p(k+\extras)$ and $|\cC_1| = p(k-p)$.

Summing up gives
\begin{align*}
\sum_{c\in \cC} \Delta \Phi^c & \geq p(k+\extras)(-1-\a\b+\tinygap+\tinygap\a\b) + p(k-p)(2+4\b-2\a\b)\\
& \geq pk(1+4\b - 3\a\b) + pk\tinygap + \extras p(-1-\a\b) -p^2(2+4\b-2\a\b)\\
& \geq pk(1+4\b - 3\a\b) \geq 0 \tag{for $\a\leq 4/3+\nicefrac{1}{(3\b)}$}
\end{align*}

The second inequality holds for any $\tinygap \geq \frac{p^2(2+4\b-2\a\b)+\extras p(1+\a\b)}{pk} = O(1/k)$. Hence, this example shows a locality gap of $2 - o(1)$ when $\a \leq \min(2, 4/3 + \nicefrac{1}{(3\beta)})$.

\paragraph*{Subcase II: $\a > 4/3+\nicefrac{1}{(3\b)}$. }
Since we focus on $\a \leq 2$ and $\b \leq 1$, this subcase implies that $\b > 1/2$ and $\a > 5/3$.
To deal with it, we create a double-bi-clique presented in \Cref{fig:LB1+4b}  with $d = 2$.
Note that every client has two local facilities at distance $2$, and $k+\extras-2$ facilities at distance
$4$.
Then for each client, we get the following potential differences:
\begin{align*}
\Delta\Phi^{c_o} & =  1+\a\b - 2- 2\b = -1+\a\b-2\b \\
\Delta\Phi^{c_1} & = 2+2\a\b - 2 - 2\b \geq 2\a\b-2\b\\
\Delta\Phi^{c_2} & = 2+2\a\b - 2 - 2\b \geq  2\a\b-2\b\\
\Delta\Phi^{c_3} & = 0
\end{align*}

Let $p_1$ and $p_2$ be the number of optimal facilities in the first clique that belong to the swap.
Let $p_2$ be the number of optimal facilities in the second clique that belong to the swap.
Then we have $|\cC_o | = \frac{k+r}{2}p_1 + \frac{k+r}{2}p_2 = \frac{k+r}{2}p$.
Similarly let $p'_1$ and $p'_2$ be the number of local facilities in the first and second clique that belong to the swap. Then we have $|\cC_1| = p'_1(\frac{k}{2}-p_1) + p'_2(\frac{k}{2}-p_2)\geq
(\frac{k}{2}-p)p$. Also note that $|\cC_2| \geq (\frac{k}{2}-p)p$.
Summing up gives
\begin{align*}
\sum_{c\in \cC} \Delta \Phi^c & \geq \nicefrac{p}{2}(k+\extras)(-1+\a\b-2\b) +2 p(\nicefrac{k}{2}-p)(2\a\b-2\b)\\
& \geq p\nicefrac{k}{2}(-1+5\a\b-6\b) - \nicefrac{\extras p}{2} (1+2\b-\a\b)- 2p^2(2\a\b-2\b)\\
& \geq pk(1/12)  - \nicefrac{\extras}{2} p(1+2\b-\a\b)- 2p^2(2\a\b-2\b) \geq 0 \tag{for $\a> 5/3$ and $\b > 1/2$.}
\end{align*}

The last inequality holds for $\extras \leq \frac{k}{24(1+2\b-\a\b)}$ and $p \leq \frac{k}{48(2\a\b-2\b)}$. 
Since $\extras$ and $p$ are absolute constant (i.e, $o(k)$), the inequality is valid for big enough $k$. Hence, this example shows a locality gap of $2$ when $4/3 + \nicefrac{1}{3\beta} \leq \a \leq 2$, in particular when $\b > \nicefrac12$ and $2 \geq \a \geq \nicefrac53$ This concludes therefore the case $\a \leq 2$.

\subsection{When \texorpdfstring{$\b \leq 1/2$ and $\a >2$}{[beta <=
    1/2 and alpha > 2]}}

\label{sec:3-2b-example}

In this section we give a bi-clique example showing a locality
gap when $\b\leq 1/3$ and $\a>2$. for constant-sized swap.
Consider the bi-clique graph in \Cref{fig:LB3-2b}
with distance $d = \min\{3-2\b-\tinygap, \a\}$.
We divide the case into two subcases. 
In first case we consider when $3-2\b \leq \a$.
Then we consider when $3-2\b > \a$.


\paragraph*{Subcase I: $3-2\b \leq \a$.}
We will first consider the case when $3-2\b \leq \a$. In that case, the current potential value of a client in the local solution is
$\Phi^c(F) =  (3-2\b-\tinygap) + \b(5-2\b-\tinygap)$ (since $\a > 2 \geq \frac{5-2\b}{3-2\b}$ for $\b \leq \nicefrac{1}{2}$).

%

For the $p(k+\extras)$ clients in $\cC_o$, where $k+\extras$ is the number of local facilities, if $3-2\b-\tinygap \leq \a$ we get the following potential difference:
\[
	\Delta\Phi^{c_o} \geq 1+\b(3-2\b-\tinygap)  - (3-2\b-\tinygap) - \b(5-2\b-\tinygap) = -2+\tinygap
\]
Note that if a client's $f^*$ is opened but its $f_1$ is closed, the client contributes $1+\b(5-2\b-\tinygap) \geq 1+\b(3-2\b-\tinygap)$, and hence the above inequality is still valid for those clients.
 
There are $p(k-p)$ clients in $\cC_1$, and they induce the following potential difference:
\[
	\Delta\Phi^{c_1} = (5-2\b-\tinygap) + \b(5-2\b-\tinygap) - (3-2\b-\tinygap) - \b(5-2\b-\tinygap) = 2
\]
Finally, clients in $\cC_3$  and $\cC_4$ do not induce a change in the potential value. 

%
The sum over all clients yields
\begin{align*}
	\sum_{c\in \cC}  \Delta\Phi^c & \geq p(k+\extras)(\tinygap-2) + 2p(k-p) \\
	& = pk(0) + p(k+\extras)(\tinygap) - 2(p\extras) -2p^2 \geq 0.
\end{align*}
The last inequality holds for any $\tinygap \geq \frac{2pr+2p^2}{p(k+\extras)} = O(1/k)$.

Hence, in the case where $\a \geq 2,~\b \leq \nicefrac{1}{2}$ and $3 - 2\b \leq \a$,  this example shows a locality gap of $3-2\b - o(1) \geq 2$.

\paragraph*{Case II: $2 \leq \a \leq 3-2\b$.} When $2\leq \a \leq 3-2\b$, client's closest distance is now $\a$.
Thus the potential value of a client before any swap is $\Phi^c(F) = \a + (2+\a)\b$. Note that $2+\a  \leq \a^2$ for $\a \geq 2$. 
We have the same number of clients in each set. Furthermore, we get the potential differences:
\begin{align*}
\Delta\Phi^{c_o} & \geq 1+\a\b - \a - (2+\a)\b \geq  1-\a -2\b\\
\Delta\Phi^{c_1} & = (2+\a) + (2+\a)\b - \a - (2+\a)\b = 2\\
\Delta\Phi^{c_2} & = 0\\
\Delta\Phi^{c_3} & = 0
\end{align*}
The sum over all clients yields
\begin{align*}
	\sum_{c\in \cC}\Delta\Phi^c & \geq  p(k+\extras)(1-\a-2\b) + p(k-p)(2)\\
	& \geq pk(3-2\b-\a) + p\extras (1-\a-2\b) -2p^2 \\
	& \geq pk\tinygap -p\extras(\a+2\b-1) -2p^2 \geq 0.\tag{$\a \leq 3-2\b-\tinygap$}
\end{align*}
The last inequality holds for any $\tinygap \geq \frac{p\extras (\a-1-2\b)+2p^2}{pk} = O(1/k)$.

Hence, in the case where $\a \leq 3 - 2\b$, $\a \geq 2$ and $\b \leq \nicefrac{1}{3}$, this example shows a locality gap of $\a \geq 2$.

\subsection{When \texorpdfstring{$\b > \haff$ and $\a > 2$}{[beta >
    1/2 and alpha > 2]}}

Finally we give lower bound examples when $\b > \nicefrac12$ and $\a >2$.
We use double-bi-clique in described \Cref{fig:LB1+4b} with $d=\min\{1+4\b-\tinygap, \a\}$.

\paragraph*{Subcase I: $1+4\b \leq \a$.}
Here, the current potential function value for a client is $\Phi^c(F) =
(1+4\b-\tinygap) + \b(1+4\b-\tinygap)$.


There are $p(\frac{k+\extras}{2})$ clients in $\cC_o$, and the potential difference for a client $c_o\in \cC_o$ is
\[
\Delta\Phi^{c_o} = 1+\b(1+4\b-\tinygap) - (1+\b)(1+4\b-\tinygap) = -4\b + \tinygap.
\]

There are at least $p(\frac{k}{2}-p)$ clients in $\cC_1$. Recall $\a > 2 \geq \frac{3+4\b-\tinygap}{1+4\b-\tinygap}$ for $\b > \nicefrac12$. The potential difference for $c_1 \in \cC_1$ is
\[
	\Delta\Phi^{c_1} = (1+4\b-\tinygap)+\b(3+4\b-\tinygap) - (1+\b)(1+4\b-\tinygap) = 2\b
\]

%
%

There are $p(\frac{k}{2}-p)$ clients in $\cC_2$, and they get the same swap value as clients in $\cC_1$.
Clients in $\cC_3$ do not induce any change in the potential.

Then sum over all clients yields
\begin{align*}
 \sum_{c\in \cC}\Delta\Phi^c \geq -4\b p\left(\frac{k+\extras}{2}\right) + \tinygap p\left(\frac{k+\extras}{2}\right) + 4\b\left(\frac{pk}{2}-p^2\right) = \tinygap\left(\frac{p(k+\extras)}{2}\right)- 4\b p\left(p+\nicefrac{\extras}{2}\right) \geq 0
\end{align*}
The difference in potential function is therefore positive for all $\tinygap >
\frac{4\b\extras+8\b p}{(k+\extras)} = O\left(\frac{1}{k}\right)$.  Hence, this example shows a locality gap of $1+4\b -o(1) \geq 2 - o(1)$ when $\a \geq \max(2, 1+4\b)$ and $\b \geq \nicefrac 12$.

\paragraph*{Case II: $2\leq \a \leq 1+4\b$.} When $\a \leq 1+4\b$, clients'  closest and the second closest local facilities are both at distance $\a$.
Thus the current potential value for a client  is $\Phi^c(F) = \a + \a\b$.
We have the same number of clients in each set. We get the following potential differences:
\begin{align*}
\Delta\Phi^{c_o} & = 1+\a\b -\a -\a\b \geq 1-\a \\
\Delta\Phi^{c_1} & = \a + (2+\a)\b - \a - \a\b = 2\b \tag{Note $2\leq \a$ implies $(2+\a) \leq \a^2$.}\\
\Delta\Phi^{c_2} & = \a + (2+\a)\b - \a - \a\b = 2\b\\
\Delta\Phi^{c_3} & = 0
\end{align*}
The sum over all clients yields
\begin{align*}
	\sum_{c\in \cC}\Delta\Phi^c  & \geq p\left(\frac{k+\extras}{2}\right )(1-\a) +  2(\nicefrac{pk}{2}-p^2)(2\b)\\
	& \geq \frac{pk}{2}(1-\a+4\b)  - \frac{p\extras}{2} (\a-1)-2p^2(2\b )\\
	& \geq \frac{pk}{2}(\tinygap)   - \frac{p\extras}{2} (\a-1)-2p^2(2\b ) \geq 0  \tag{$\a \leq 1+4\b -\tinygap$}
\end{align*}
The last inequality holds for any $\tinygap \geq \frac{p\extras (\a-1)+8p^2\b}{pk} = O(1/k)$.
Hence, this example shows a locality gap of $\a -o(1) \geq 2 - o(1)$ when $2 \leq \a \leq 1+4\b$ and $\b \geq \nicefrac 12$.
\newcommand{\lgap}{x}


\section{Motivating our Swaps}
\label{sec:fig-examples}

In this section we present examples that motivate our choice of
potential function and our swaps. In particular they show that the
swap structures defined in previous works are not powerful enough to
prove our results. 

\begin{figure}[H]
\centering
\includegraphics[scale=0.5]{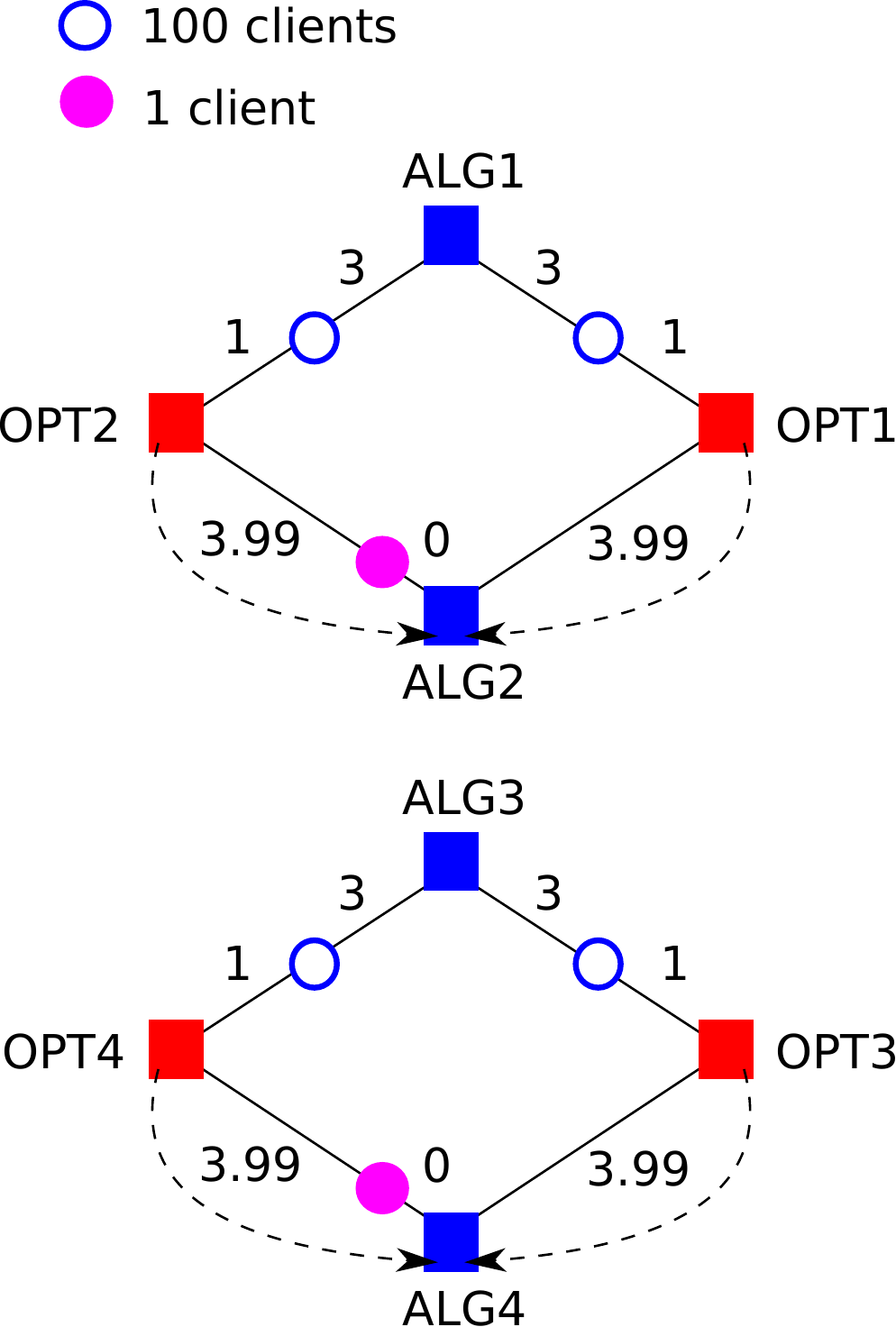}
\caption{A bad scenario for the swap structures defined
  by~\cite{Gupta2008SimplerLS}. 
}
\label{fig:badcasesingleGupta}
\end{figure}

The analysis in~\cite{Gupta2008SimplerLS} matches each optimal facility to its
closest local facility. So it matches both OPT1, OPT2 to ALG2 and both OPT3, OPT4 to
ALG4, leaving ALG1 and ALG3 with no facility of $\Opt$ matched to
them. Hence, two swaps are defined: (1) swapping in OPT1 and OPT2 and
removing ALG2 and ALG$j$ for some $j \in \{1,3\}$ that remains
unspecified in their analysis, and (2) swapping in OPT3 and OPT4 and
removing ALG4 and ALG(4-$j$).  Now, if we consider the swaps defined
by choosing $j = 3$, the set of equations obtained does not allow us
to deduce that the solution is not a local optimum, as long as
$\alpha > 5/3$ and for any $\beta < 1$.

\begin{figure}[H]
\centering
\includegraphics[scale=0.5]{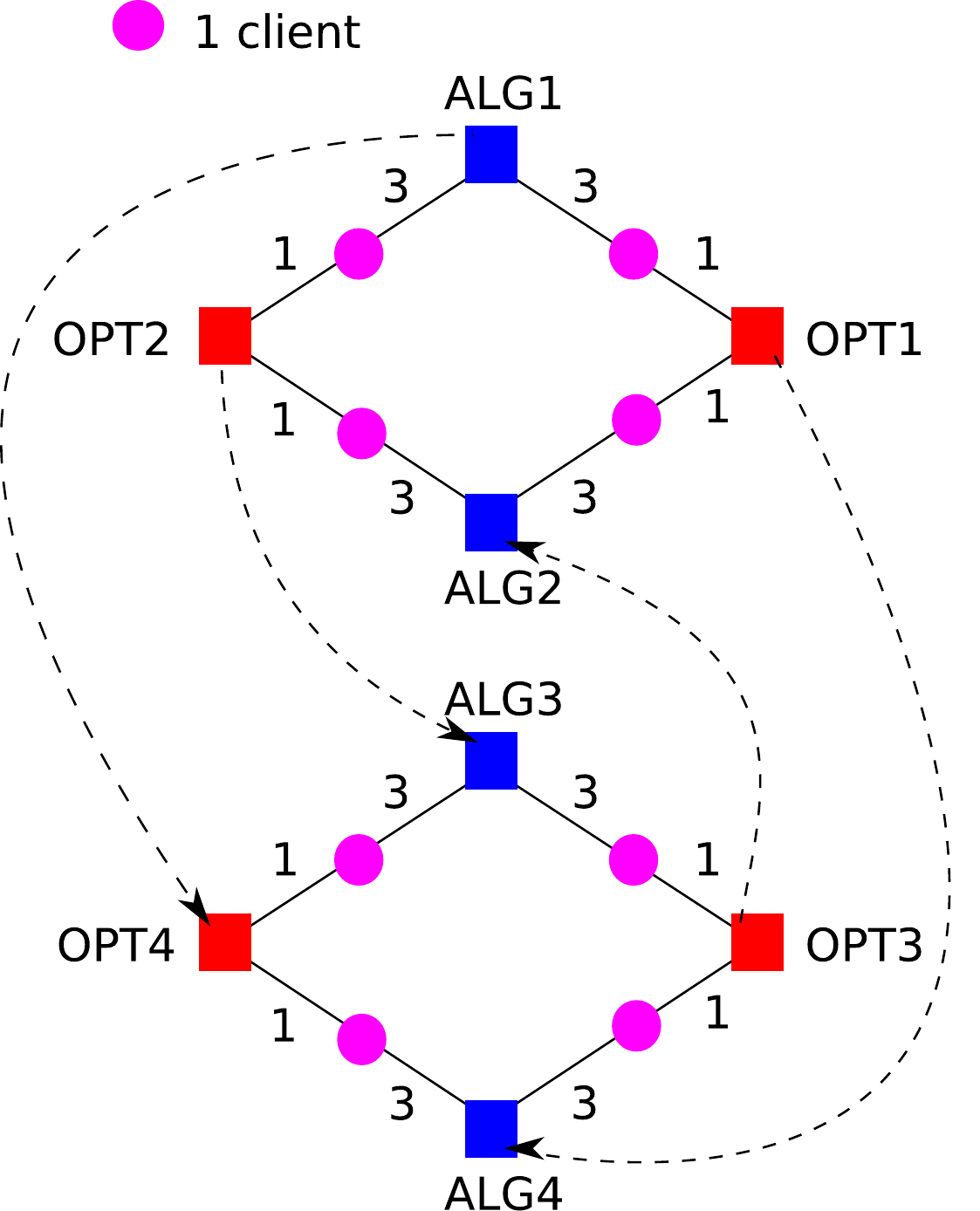}
\caption{A bad scenario for the swap structures defined by Arya et al.~\cite{Arya2001LocalSearch}.}          
\label{fig:badcasesingleArya}
\end{figure}

The definition of the swap structure in~\cite{Arya2001LocalSearch}
does not uniquely identify which local facility is matched to which
optimal facility. Hence, if the analysis matches ALG1 with OPT4, ALG2
with OPT3, ALG3 with OPT2 and ALG4 with OPT1, the set of linear
equations obtained does not allow us to deduce that the instance is not a
local optimum.


%

\section{Useful Inequalities}\label{ap:proofUseful}

In this section, we prove the inequalities in Table~\ref{tab:ineqs}. 
We also give some more inequalities in Table~\ref{tab:apx:ineqs};
these will be used in \S\ref{apx:bounds}.

{\small
\begin{table}[h]
\begin{center}
\begin{tabular}{| c l| c | }
\hline
    \rowcolor{black!90} \textcolor{white}{Bound}& & \textcolor{white}{Coundition} \\ 
     \hline
    $d_2 \geq (\nicefrac{1}{\rho})\, d_1 - (1+\nicefrac{1}{\rho})\, d^*$ &$\inlineeqnum\label{eqn:apx:lb3} $ & $\eta_1(f^*) = f_1$\\
     \hline
  	$d_2 \leq d^* + (\nicefrac{1}{\rho})(d_1+d^*)$ & $\inlineeqnum\label{eqn:apx:ub3}$ & $\eta_1(f^*) = f_1$\\
     \hline
    $d(c,\eta_2(f^*)) \leq d^* + \nicefrac{1}{\rho}(d^*+d_1)$ & $\inlineeqnum\label{eqn:apx:eta2rho}$ & $\eta_1(f^*) = f_1$\\
     \hline
     $d_2 \leq d^* + \rho(d_1+d^*)$ & $\inlineeqnum\label{eqn:apx:ub4}$  & $\eta_1(f^*) \neq f_1$\\
     \hline
    $d(c,\eta_1(f^*))\leq d^*+\rho(d^*+d_1)\leq 2d^*+d_1$ & $\inlineeqnum\label{eqn:apx:eta1}$  & $\eta_1(f^*)\neq f_1$\\
    \hline
   $d(c,\eta_2(f^*)) \leq 2d^*+d_1$ & $\inlineeqnum\label{eqn:apx:eta2}$  & $\eta_1(f^*) \neq f_1$\\
     \hline
    $d(c,\pi(f_2)) \leq 2d_2+d^*$ & $\inlineeqnum\label{eqn:pif2}$ & $\eta_1(f^*) \neq f_1$\\
     \hline
\end{tabular}
\end{center}
\caption{More useful inequalities}\label{tab:apx:ineqs}
\end{table}
}


For clients $c$ with $\eta_1(f^*) = f_1$:
\begin{flalign*}
&& d_2  & \geq \ts  d(f^*,\eta_2) - d^*
 = \frac{1}{\rho}d(f^*,\eta_1) - d^*
  \geq (\nicefrac{1}{\rho})\, d_1 - (1+\nicefrac{1}{\rho})\, d^*, 
       && \text{(proving \eqref{eqn:apx:lb3})} \\
&& d_2  & \ts \leq  d^* + d(f^*,\eta_2) = d^* +
\frac{1}{\rho}d(f^*,\eta_1) \leq d^* + \frac{1}{\rho}(d_1 + d^*). 
&& \text{(proving \eqref{eqn:apx:ub3})} \\
&& d(c,\eta_2) &\leq d(c,f^*) + d(f^*,\eta_2) = d^* +
(\nicefrac{1}{\rho})
d(f^*,f_1)\leq d^* +\nicefrac{1}{\rho}(d^*+ d_1).        && \text{(proving \eqref{eqn:apx:eta2rho})}
     \end{flalign*}
Else when $\eta_1(f^*) \neq f_1$:
\begin{flalign*}
&& d(c,\eta_1) &\leq d(c,f^*) + d(f^*,\eta_1) = d(c,f^*) +
\rho\, d(f^*,\eta_2) \leq d^* + \rho (d^* + d_1) && \text{(proving \eqref{eqn:apx:eta1})}\\
&& d(c,\eta_2) &\leq d(c,f^*) + d(f^*,\eta_2)\leq d^* + (d^*+d_1). &&
\text{(proving \eqref{eqn:apx:eta2})} \\
&& d_2  & \leq d(c,\eta_1) \stackrel{\eqref{eqn:apx:eta1}}{\leq}
d^* + \rho(d_1 + d^*), && \text{(proving 
  \eqref{eqn:apx:ub4} and \eqref{eqn:ub5})}
\end{flalign*}
Combining \eqref{eqn:apx:eta1} and \eqref{eqn:apx:eta2} gives \eqref{eqn:ubeta}.


Recalling that $\pi(f)$ is the closest optimal facility to $f$, we get
for any client $c$,
\begin{flalign*}
&& d(c,\pi(f_1))  & \leq d(c,f_1) + d(f_1, \pi(f_1)) \leq d_1 + d(f_1,
f^*) \leq 2d_1 + d^* && \text{(proving \eqref{eqn:ubpartner})} \\
&& d(c,\pi(f_2))  & \leq d(c,f_2) + d(f_2, \pi(f_2)) \leq d_2 + d(f_2,
f^*) \leq 2d_2 + d^*. && \text{(proving \eqref{eqn:pif2})}
\end{flalign*}

To prove \eqref{eqn:bestBalance}, we use that for any $a,b \geq 0$ and
$\beta \leq 1$, the expression
$\min(a,b) + \beta\, \max(a,b) = \min(a+\b b, b + \b a)$ is smaller
than any convex combination
$(1-\lambda)(a + \b b) + \lambda (b + \b a)$ with $\lambda \in [0,1]$. Setting
$\lambda = \frac{\b}{1-\b}$ and simplifying gives $(1-\b)a + 2\b
b$. Using $a = d^*$ and $b = d_1$ completes the proof.

\section{Proof of \Cref{lem:boundsAll}}
\label{apx:bounds}

We now present the proof of \Cref{lem:boundsAll}, giving bounds for
all the client types other than type \xE. The idea is the same for
each one: First we fix a client $c$ of some type. We partition the
amenable event into some sub-events, and look on some sub-event $\cE$.
We consider a generic swap set $\cP$ generated under that event, and
give an upper bound for the maximum potential change for client $c$
due to these swaps. Combining over all sub-events (with the correct
probability values) gives the expected potential change. The largest
such change for each client type is then shown to be the one recorded
in \Cref{lem:boundsAll}.

When we prove upper bounds for the potential change caused by a  swap set $\cP$, 
we assume that both $\move{\neg f_1}$ and $\move{\neg f_2}$ exist in $\cP$
(if $f$ is heavy $\move{\neg f}$ does not exist). 
As we mentioned in \Cref{sec:bounds-notation}, our bounds also hold in 
cases where either of them does not exist, because  our upper bounds for 
$\WCchange_{\move{\neg f}}$ is non-negative  as long as $\move{f^*} \neq \move{\neg f}$ for 
any $f \in \{f_1, f_2\}$.

%
%
%
%
%
%
%
%

In the rest of this section, we prove each 
inequality from \Cref{lem:boundsAll}.


\boundsAll*

\subsection{Proof of \eqref{eqn:bndFar}: Far Clients of Type \xA and \xB}

In this section, we show that for any far case client $c$ of type \xA or \xB, we have
\[
\Gsum{c} \leq 2.467\, d^*(c) - 1.13085\, d_1(c) + \epsd.
\]
We give different analysis depending on whether $f^*$ points to $\eta_1$ or $\eta_2$; this is different from our type \xE analysis, where our bounds are the same in both cases. Formally, we partition the amenable event $\cA$ as the union of $\cS_1\cap\cA$, $\cS_2\cap\cA$, $\cT_1\cap\cA$, and $\cT_2\cap\cA$. We upper-bound $\change_\cA(c)$ by
\begin{align}
\change_\cA(c) & \leq \Pr[\cS_1\cap\cA]\WCchange_{\cS_1\cap\cA}(c) + \Pr[\cS_2\cap\cA]\WCchange_{\cS_2\cap\cA}(c) + \Pr[\cT_1\cap\cA]\WCchange_{\cT_1\cap\cA}(c) + \Pr[\cT_2\cap\cA]\WCchange_{\cT_2\cap\cA}(c)\nonumber \\
& \leq  \Pr[\cS_1]\WCchange_{\cS_1\cap\cA}(c) + \Pr[\cS_2]\WCchange_{\cS_2\cap\cA}(c) + \Pr[\cT_1]\WCchange_{\cT_1\cap\cA}(c) + \Pr[\cT_2]\WCchange_{\cT_2\cap\cA}(c) + \epsd.\label{eq:typeAfar}
\end{align}
The probabilities $\Pr[\cS_1],\Pr[\cS_2],\Pr[\cT_1],\Pr[\cT_2]$ are given in \Cref{tab:prob}. We proceed by showing upper-bounds for the $\WCchange$ values, the potential changes of client $c$ on the worst-case swap set $\cP$, for far clients of type \xA and \xB in the following subsections.
%

\subsubsection{Far clients of type \xA: \texorpdfstring{$f_1 = \eta_1$}{[f1 = eta1]}}
\label{subsec:farA}
\paragraph{Simple swaps with $\tau(f^*) = \eta_1$} Type \xA clients have $f_1 = \eta_1$, which is the same as $\tau(f^*)$, so we have $\move{f^*} = \move{\neg f_1}$ in $\cP$ by implication (ii) of amenability. On that swap, the client can be served by $f^*$. Therefore,
\begin{align*}
  \dSwapHa\leq \hoono{(1+\a\b)\, d^* - (1+\a\b)\, d_1}. \tag{$\WCchange_{\move{f^*, \neg f_1}}$}
\end{align*}

\paragraph{Simple swaps with $\tau(f^*) = \eta_2$} Since $\tau(f^*) = \eta_2\neq f_1$, we know $\move{f^*}\neq \move{\neg f_1}$ by implication (Siv) of amenability. On swap $\move{f^*}$, $c$ can be served by both $f^*$ and $f_1$. On swap $\move{\neg f_1}$, $c$ can be served by $\eta_2$ (by implication (Siii) of amenability). 
Note that $d(c,\eta_2) \leq  d^* + \nicefrac 1\rho\cdot (d^* + d_1)$ by \eqref{eqn:apx:eta2rho}. Therefore,
\begin{align*}
  \dSwapHb & \leq  d_1 + \b\, d^* - (1+\a\b)\, d_1 \tag{$\WCchange_{\move{f^*}}$}\\
  & ~~~ + (1+\a\b)(d^*+\nicefrac 1\rho\cdot (d^* + d_1)) - (1+\a\b)\, d_1 \tag{$\WCchange_{\move{\neg f_1}}$}\\
  & = \hoono{((1+\nicefrac{1}{\rho})(1+\a\b)+\b)\, d^* 
    -((1-\nicefrac{1}{\rho})(1+\a\b)+\a\b)\, d_1}.
\end{align*}

\paragraph{Tree swaps with $\tau(f^*) = \eta_1$} We have $\move{f^*} = \move{\neg f_1}$ by implication (ii) of amenability. On that swap, the client can be served by $\move{f^*}$. Therefore,
\begin{align*}
  \dSwapTa\leq \hoono{(1+\a\b)\, d^* - (1+\a\b)\, d_1}. \tag{$\WCchange_{\move{f^*, \neg f_1}}$}
\end{align*}

\paragraph{Tree swaps with $\tau(f^*) = \eta_2$} If $\move{f^*} = \move{\neg f_1}$, then we have the same bound as above:
\begin{align}
  \sum_{(P,Q)\in\cP}\WCchange_{(P,Q)}(c) \leq (1+\a\b)\, d^* - (1+\a\b)\, d_1. \label{eq:A-far-eta2-tree1}
\end{align}
If $\move{f^*}\neq\move{\neg f_1}$, then on swap $\move{\neg f_1}$, $c$ can be served by $\eta_2$ and $\pi(f_1)$, by implications (ii) and  (Tii) of amenability. We already showed $d(c,\eta_2)\leq d^* + \nicefrac 1\rho\cdot (d^* + d_1)$. We also have $d(c,\pi(f_1))\leq 2d_1 + d^*$ by \eqref{eqn:ubpartner}. Therefore,
\begin{align}
  \sum_{(P,Q)\in\cP}\WCchange_{(P,Q)}(c) & \leq d_1 + \b\, d^* - (1+\a\b)\, d_1 \tag{$\WCchange_{\move{f^*}}$}\nonumber \\
  & ~~~ +  (d^*+\nicefrac 1\rho\cdot (d^* + d_1)) + \b(2d_1 + d^*) - (1+\a\b)\, d_1 \tag{$\WCchange_{\move{\neg f_1}}$}\nonumber \\
  & = (1+\nicefrac{1}{\rho} + 2\b)\, d^* - (1+2\a\b -2\b - \nicefrac{1}{\rho} )\, d_1. \label{eq:A-far-eta2-tree2}
\end{align}
For our choice of $\a,\b$, \eqref{eq:A-far-eta2-tree2} is larger than \eqref{eq:A-far-eta2-tree1}, so we have
\begin{equation*}
\dSwapTb \leq \hoono{(1+\nicefrac{1}{\rho} + 2\b)\, d^* - (1+2\a\b -2\b - \nicefrac{1}{\rho} )\, d_1}.
\end{equation*}
Summarizing, we have

\begin{mybox}
\begin{align*}
\dSwapHa \leq  &(1+\a\b)\, d^* - (1+\a\b)\, d_1
&	=  1.6\, d^* - 1.6\, d_1\\
\dSwapHb \leq  &((1+\nicefrac{1}{\rho})(1+\a\b)+\b)\, d^* -((1-\nicefrac{1}{\rho})(1+\a\b)+\a\b)\, d_1\span\\
&	&\quad  =  (1.8+1.6/\rho)\, d^* - (2.2-1.6/\rho)\, d_1\\
\dSwapTa\leq  &(1+\a\b)\, d^* - (1+\a\b)\, d_1
&	= 1.6\, d^* - 1.6\, d_1 \\
\dSwapTb\leq  &(1+\nicefrac{1}{\rho} + 2\b)\, d^* - (1+2\a\b -2\b - \nicefrac{1}{\rho} )\, d_1
& 	 = (1.4 + {1}/{\rho})\, d^* - (1.8-{1}/{\rho})\, d_1
\end{align*}
\end{mybox}

We now combine these inequalities using \eqref{eq:typeAfar}. If $\rho(f^*)\leq\nicefrac 23$, we have $\Pr[\cS_1] = \Pr[\cT_1] = \nicefrac 12$ and $\Pr[\cS_2] = \Pr[\cT_2] = 0$. Therefore,
\begin{align*}
\Gsum{c} & \leq  \nicefrac 12\cdot \WCchange_{\cS_1\cap\cA}(c) + \nicefrac 12\cdot \WCchange_{\cT_1\cap\cA}(c) + \epsd\\
&\leq \hoono{1.6\, d^* - 1.6\, d_1} + \epsd.
\end{align*}
If $\nicefrac 23 < \rho(f^*)\leq \nicefrac 34$, we have $\Pr[\cS_1] = \nicefrac 12,\Pr[\cS_2] = 0,\Pr[\cT_1] = \nicefrac 14,\Pr[\cT_2] = \nicefrac 14$. Therefore,
\begin{align*}
\Gsum{c} & \leq \nicefrac 12\cdot\WCchange_{\cS_1\cap\cA}(c) + \nicefrac 14\cdot\WCchange_{\cT_1\cap\cA}(c) + \nicefrac 14\cdot\WCchange_{\cT_2\cap\cA}(c) + \epsd\\
& \leq \nicefrac 34\cdot(1.6\, d^* - 1.6\, d_1) + \nicefrac 14 \cdot ((1.4 + \nicefrac 32)d^* - (1.8 - \nicefrac 32)d_1) + \epsd\\
& = \hoono{1.925\, d^* - 1.275\, d_1} + \epsd.
\end{align*}
If $\rho(f^*) > \nicefrac 34$, we have $\Pr[\cS_1] = \nicefrac 54 - \rho,\Pr[\cS_2] = \rho - \nicefrac 34, \Pr[\cT_1] = \Pr[\cT_2] = \nicefrac 14$. Therefore,
\begin{align*}
\Gsum{c} & \leq (\nicefrac 54 - \rho)\cdot\WCchange_{\cS_1\cap\cA}(c) + (\rho - \nicefrac 34)\cdot \WCchange_{\cS_2\cap\cA}(c) + \nicefrac 14\cdot\WCchange_{\cT_1\cap\cA}(c) + \nicefrac 14\cdot\WCchange_{\cT_2\cap\cA}(c) + \epsd\\
& \leq (3 + 0.2 \rho - 0.95 / \rho)\,d^* - (0.6\rho + 0.95/\rho - 0.4)\,d_1 + \epsd\\
& \leq (3 + 0.2 - 0.95)\, d^* - (0.6 + 0.95 - 0.4)\,d_1 + \epsd\\
& = \hoono{2.25\, d^* - 1.15\, d_1} + \epsd.
\end{align*}

\subsubsection{Far clients of type \xB: \texorpdfstring{$f_1 = \eta_2$}{[f1 = eta2]}}
When $c$ is a far client of type \xB, we have $\move{f^*} = \move{\neg f_1}$ on $\cS_2\cap\cA$ and $\cT_2\cap\cA$. This is exactly the situation for type \xA clients on $\cS_1\cap\cA$ and $\cT_1\cap\cA$. Therefore, we have the same bound for all of these cases:
\begin{align*}
  \dSwapHb &\leq \hoono{(1+\a\b)\, d^* - (1+\a\b)\, d_1}, \tag{$\WCchange_{\move{f^*, \neg f_1}}$}\\
  \dSwapTb &\leq \hoono{(1+\a\b)\, d^* - (1+\a\b)\, d_1}. \tag{$\WCchange_{\move{f^*, \neg f_1}}$}
\end{align*}
We continue to bound $\WCchange_{\cS_1\cap\cA}(c)$ and $\WCchange_{\cT_1\cap\cA}(c)$.

\paragraph{Simple swaps with $\tau(f^*) = \eta_1$}By implication (Siv) of amenability, we have $\move{f^*}\neq \move{\neg f_1}$. 
On swap $\move{f^*}$, the client can be served by $f^*$ and $f_1$. On swap $\move{\neg f_1}$, the client can be served by $\eta_1$ (by implication (Siii) of amenability). Also, $d(c,\eta_1)\leq d^* + \rho (d^* + d_1)$ by \eqref{eqn:apx:eta1}. Therefore,
\begin{align*}
  \dSwapHa & \leq  d_1 + \b\, d^* - (1+\a\b)\, d_1 \tag{$\WCchange_{\move{f^*}}$}\\
  & ~~~ + (1+\a\b)(d^*+\rho(d^* + d_1)) - (1+\a\b)\, d_1 \tag{$\WCchange_{\move{\neg f_1}}$}\\
  & = \hoono{((1+\rho)(1+\a\b)+\b)\, d^* -((1-\rho)(1+\a\b)+\a\b)\, d_1}.
\end{align*}

\paragraph{Tree swaps with $\tau(f^*) = \eta_1$} If $\move{f^*} = \move{\neg f_1}$, then we have
\begin{align}
  \sum_{(P,Q)\in\cP}\WCchange_{(P,Q)}(c) \leq (1+\a\b)\, d^* - (1+\a\b)\, d_1. \label{eq:B-far-eta2-tree1}
\end{align}
If $\move{f^*}\neq\move{\neg f_1}$, then on swap $\move{\neg f_1}$, $c$ can be served by $\eta_1$ and $\pi(f_1)$ by implications (ii) and (Tii) of amenability. We  showed $d(c,\eta_1)\leq d^* + \rho(d^* + d_1)$. We also have $d(c,\pi(f_1)) \leq 2d_1 + d^*$ by \eqref{eqn:ubpartner}. Therefore,
\begin{align}
  \sum_{(P,Q)\in\cP}\WCchange_{(P,Q)}(c) & \leq d_1 + \b\, d^* - (1+\a\b)\, d_1 \tag{$\WCchange_{\move{f^*}}$}\nonumber \\
  & ~~~ +  (d^*+\rho (d^* + d_1)) + \b(2d_1 + d^*) - (1+\a\b)\, d_1 \tag{$\WCchange_{\move{\neg f_1}}$}\nonumber \\
  & = (1+\rho + 2\b)\, d^* - (1+2\a\b -2\b - \rho )\, d_1. \label{eq:B-far-eta2-tree2}
\end{align}
Taking the maximum of \eqref{eq:B-far-eta2-tree1} and \eqref{eq:B-far-eta2-tree2} using $\alpha = 3,\beta = 0.2$, we have
\begin{equation*}
\dSwapTa \leq \hoono{(1.4+ \max\{\rho,0.2\})\, d^* - (1.8 - \max\{\rho,0.2\} )\, d_1}.
\end{equation*}

Summarizing, we have

\begin{mybox}
\begin{align*}
\dSwapHa & \leq  ((1+\rho)(1+\a\b)+\b)\, d^* -((1-\rho)(1+\a\b)+\a\b)\, d_1 \span\\
	& &  = (1.8+1.6\rho)\, d^* - (2.2-1.6\rho)\, d_1\\
\dSwapHb & \leq   (1+\a\b)\, d^* - (1+\a\b)\, d_1
	& =  1.6\, d^* - 1.6\, d_1\\
\dSwapTa  &  
	&  \leq  (1.4+\max\{\rho,0.2\})\, d^* - (1.8-\max\{\rho,0.2\})\, d_1\\
\dSwapTb & \leq  (1+\a\b)\, d^* - (1+\a\b)\, d_1
	& = 1.6\, d^* - 1.6\, d_1
\end{align*}
\end{mybox}
We now combine these inequalities using \eqref{eq:typeAfar}. If $\rho(f^*)\leq\nicefrac 23$, we have $\Pr[\cS_1] = \Pr[\cT_1] = \nicefrac 12$ and $\Pr[\cS_2] = \Pr[\cT_2] = 0$. Therefore,
\begin{align*}
\change_{\cA}(c) & \leq \nicefrac 12\cdot\WCchange_{\cS_1\cap\cA}(c) + \nicefrac 12\cdot\WCchange_{\cT_1\cap\cA}(c) + \epsd\\
& \leq \nicefrac 12 \cdot ((1.8 + 1.6 \times \nicefrac 23)d^* - (2.2 - 1.6 \times \nicefrac 23)d_1)\\
& ~~~ + \nicefrac 12\cdot ((1.4 + \nicefrac 23)d^*  - (1.8 - \nicefrac 23)d_1)\\
& ~~~ + \epsd\\
& \leq \hoono{2.46667 d^* - 1.13333 d_1} + \epsd.
\end{align*}
If $\nicefrac 23 < \rho(f^*)\leq \nicefrac 34$, we have $\Pr[\cS_1] = \nicefrac 12,\Pr[\cS_2] = 0,\Pr[\cT_1] = \nicefrac 14,\Pr[\cT_2] = \nicefrac 14$. Therefore,
\begin{align*}
\change_{\cA}(c) & \leq \nicefrac 12\cdot\WCchange_{\cS_1\cap\cA}(c) + \nicefrac 14\cdot\WCchange_{\cT_1\cap\cA}(c) + \nicefrac 14\cdot \WCchange_{\cT_2\cap\cA}(c) + \epsd\\
& \leq \nicefrac 12 \cdot ((1.8 + 1.6 \times \nicefrac 34)d^* - (2.2 - 1.6 \times \nicefrac 34)d_1)\\
& ~~~ + \nicefrac 14\cdot ((1.4 + \nicefrac 34)d^*  - (1.8 - \nicefrac 34)d_1)\\
& ~~~ + \nicefrac 14\cdot (1.6 d^* - 1.6d_1)\\
& ~~~ + \epsd\\
& = \hoono{2.4375 d^* - 1.1625 d_1} + \epsd.
\end{align*}
If $\rho(f^*) > \nicefrac 34$, we have $\Pr[\cS_1] = \nicefrac 54 - \rho,\Pr[\cS_2] = \rho - \nicefrac 34, \Pr[\cT_1] = \Pr[\cT_2] = \nicefrac 14$. Therefore,
\begin{align*}
\Gsum{c} & \leq (\nicefrac 54 - \rho)\cdot\WCchange_{\cS_1\cap\cA}(c) + (\rho - \nicefrac 34)\cdot \WCchange_{\cS_2\cap\cA}(c) + \nicefrac 14\cdot\WCchange_{\cT_1\cap\cA}(c) + \nicefrac 14\cdot\WCchange_{\cT_2\cap\cA}(c) + \epsd\\
& \leq (1.8 + 2.05\rho - 1.6\rho^2)d^* - (2.4 - 2.85\rho + 1.6\rho^2)d_1 + \epsd\\
& \leq (1.8 + 2.05 \cdot \nicefrac 34 - 1.6 \cdot (\nicefrac 34)^2) d^* - (2.4 - 2.85 \cdot \nicefrac{2.85}{3.2} + 1.6 \cdot (\nicefrac{2.85}{3.2})^2)d_1 + \epsd\\
& \leq \hoono{2.4375 d^* - 1.13085 d_1} + \epsd.
\end{align*}


\subsection{Proof of \eqref{eqn:bndA}: Close Clients of Type \xA}

In this section, we show that for any close case client $c$ with type \xA, we have
\[
\Gsum{c} \leq 2.375\, d^*(c) - 0.9\, d_1(c).
\]


%
%
%
%
%

\subsubsection{Clients with \texorpdfstring{$\rho(f^*)\leq\nicefrac
    23$}{[rho(f*) <= 2/3]}} 
\label{sec:typeAlow}

We first consider the case where $\rho(f^*)\leq \nicefrac 23$. Our analysis for this case is very simple: we directly use $\change_{\cA}(c)\leq \Pr[\cA]\WCchange_{\cA}(c) \leq \WCchange_{\cA}(c) + \epsd$ without considering sub-events of $\cA$. 

Let us fix a generic swap set $\cP$ generated on the amenable event $\cA$.
$\rho(f^*) \leq \nicefrac{2}{3}$ implies that $\tau(f^*)$ always equals to $\eta_1 = f_1$. By implication (ii) of amenability, we have $\move{f^*} = \move{\neg f_1}$ in $\cP$. Therefore,
\begin{align*}
\WCchange_{\cA}(c) \leq {} & (1+\a\b)\, d^* - d_1 - \b\, d_2 \tag{$\WCchange_{\move{f^*, \neg f_1}}$}\\
&  + (1+\a\b)\, d_1 - d_1 - \b\, d_2 \tag{$\WCchange_{\move{\neg f_2}}$}\\
= {} & (1 + \a\b)\, d^* - (1 - \a\b)\, d_1 - 2\b\, d_2.
\end{align*}

Note that this bound also holds when $\move{\neg f_2} = \move{f^*,\neg f_1}$, because our bound for $\move{f^*,\neg f_1}$ does not require $f_2$ to remain open after the swap and $\WCchange_{\move{\neg f_2}}$ is non-negative.

If $d_1 \leq d^*$, then we have
\begin{align*}
\WCchange_{\cA}(c) \leq {} & (1+\a\b+\b)\, d^* - (1- \a\b + \b)\, d_1 - 2\b\, d_2 \tag{$d_1 \leq d^*$}\\
\leq {} & (1+\a\b+\b)\, d^* - (1- \a\b + 3\b)\, d_1  \tag{$d_2 \geq d_1$}\\
= {} & \hoono{1.8\, d^* - d_1}.
\end{align*}

If $d_1 \geq d^*$, we have
\begin{align*}
\WCchange_{\cA}(c) \leq {}& (1+\a\b)\, d^* - (1-\a\b)\, d_1 - 2\b\, d_2\\
\leq {} & (1+\a\b)\, d^* - (1-\a\b)\, d_1 -
          2\b\left(\frac{1+\nicefrac{1}{\rho}}{2}d_1 -
          \frac{1+\nicefrac{1}{\rho}}{2}d^*\right)\tag{averaging
          \eqref{eqn:apx:lb3} with $d_2\geq d_1$}\\
\leq {} & (1+\a\b)\, d^* - (1-\a\b)\, d_1 -2\b\left(\frac{5}{4}d_1 - \frac{5}{4}d^*\right) \tag{$\rho\leq \nicefrac 23$ and $d^* \leq d_1$}\\
= {} & (1+\a\b+2.5\b)\, d^* - (1-\a\b+2.5\b)\, d_1\\
= {} & \hoono{2.1\, d^* - 0.9\, d_1}.
\end{align*}

\subsubsection {Clients with \texorpdfstring{$\rho(f^*)>\nicefrac
    23$}{[rho(f*) > 2/3]}} 
Now we turn to close clients of type \xA with $\rho(f^*)>\nicefrac 23$. Our analysis for simple swaps adopts the usual strategy:
\begin{align*}
\change_{\cS\cap\cA}(c) & \leq \Pr[\cS_1\cap\cA]\WCchange_{\cS_1\cap\cA}(c) + \Pr[\cS_2\cap\cA]\WCchange_{\cS_2\cap\cA}(c)\\
& \leq \Pr[\cS_1]\WCchange_{\cS_1\cap\cA}(c) + \Pr[\cS_2]\WCchange_{\cS_2\cap\cA}(c) + \epsd.
\end{align*}
However, we will be a little more careful in our tree swaps analysis. We further partition the tree events $\cT_1$ and $\cT_2$ as $\cT_1 = \cT_{11}\cup\cT_{12}$ and $\cT_2 = \cT_{21}\cup\cT_{22}$ in the following way. 
$\cT_{11}$ is defined as the intersection of $\cT_1$ and the event that $\move{f^*}$ is the only
 swap closing any facility in $\{f_1, f_2\}$.
 $\cT_{21}$ is defined as the intersection of $\cT_2$ 
 and the event that there is a swap which closes both $f_1$ and $f_2$ but does not open the original copy of $f^*$. $\cT_{12}$ and $\cT_{22}$ are defined accordingly: $\cT_{12} = \cT_1\backslash\cT_{11}$ and $\cT_{22} = \cT_2\backslash\cT_{21}$.


Recall that $\rho(f^*)>\nicefrac 23$ implies $\Pr[\cT_1] = \Pr[\cT_2] = \nicefrac 14$. The naive way to bound $\change_{\cT\cap\cA}(c)$ is by the following:
\begin{align*}
\change_{\cT\cap\cA}(c) & \leq \Pr[\cT_1]\WCchange_{\cT_1\cap\cA}(c) + \Pr[\cT_2]\WCchange_{\cT_2\cap\cA}(c) + \epsd\\
& = \nicefrac 14 \cdot \max\{\WCchange_{\cT_{11}\cap\cA}(c), \WCchange_{\cT_{12}\cap\cA}(c)\} + \nicefrac 14 \cdot \max\{\WCchange_{\cT_{21}\cap\cA}(c), \WCchange_{\cT_{22}\cap\cA}(c)\} + \epsd.
\end{align*}

If we ignore the $\epsd$ term, the above bound is equal to $\nicefrac 14$ times the maximum of all four sums: $\WCchange_{\cT_{11}\cap\cA}(c) + \WCchange_{\cT_{21}\cap\cA}(c), \WCchange_{\cT_{11}\cap\cA}(c) + \WCchange_{\cT_{22}\cap\cA}(c), \WCchange_{\cT_{12}\cap\cA}(c) + \WCchange_{\cT_{21}\cap\cA}(c), \WCchange_{\cT_{12}\cap\cA}(c) + \WCchange_{\cT_{22}\cap\cA}(c)$. However, by relating the probabilities of $\cT_{11}$ and $\cT_{21}$, we have the following lemma (proved in \Cref{subsec:typeAaveraging}),
which gives an improved bound by not taking $\WCchange_{\cT_{12}\cap\cA}(c) + \WCchange_{\cT_{21}\cap\cA}(c)$ into the maximum.

\begin{restatable}[Type \xA averaging]{lemma}{typeA}
\label{lm:typeA}
For a close client of type \xA with $\rho(f^*)>\nicefrac 23$, we have
\begin{equation*}
\change_{\cT\cap\cA}(c) \leq \nicefrac 14 \cdot \max\{\WCchange_{\cT_{11}\cap\cA} + \WCchange_{\cT_{21}\cap\cA}, \WCchange_{\cT_{11}\cap\cA} + \WCchange_{\cT_{22}\cap\cA}, \WCchange_{\cT_{12}\cap\cA} + \WCchange_{\cT_{22}\cap\cA}\} + \epsd.
\end{equation*}
\end{restatable}

 
We now proceed to show upper bounds for the worst-case potential change on each event.
\paragraph{Simple swaps with $\tau(f^*) = \eta_1$} We have $\move{f^*} = \move{\neg f_1}\neq\move{\neg f_2}$ by implications (ii) and (Siv) of amenability. On swap $\move{f^*}$, the client can be served by $f^*$, and on $\move{\neg f_2}$, the client can be served by $f_1$ (at distance $d_1$) and $\eta_2$ (at distance $\leq d^* + \nicefrac 1\rho \cdot(d^* + d_1)$), by implication (Siii) of amenability. Therefore,
\begin{align*}
  \dSwapHa & \leq (1+\a\b)\, d^* - d_1 - \b\, d_2 \tag{$\WCchange_{\move{f^*, \neg f_1}}$}\\
  & ~~~ +  d_1 + \b(d^* + \nicefrac 1\rho \cdot (d^* + d_1)) - d_1 - \b\, d_2 \tag{$\WCchange_{\move{\neg f_2}}$} \\
  & = \hoono{(1+\a\b + \b + \nicefrac{\b}{\rho} )\, d^*
    - (1- \nicefrac{\b}{\rho})\, d_1 - 2\b\, d_2}.
\end{align*}

\paragraph{Simple swaps with $\tau(f^*) = \eta_2$} By implications (Siii) and (Siv) of amenability, the three swaps $\move{f^*}$,$\move{\neg f_1}$, $\move{\neg f_2}$ are all different. On swap $\move{f^*}$, the client can be served by $f^*$ and $f_1$. On swap $\move{\neg f_1}$, the client can be served by $f_2$ and $\eta_2$. On swap $\move{\neg f_2}$, the client can be served by $f_1$ and $\eta_2$. Therefore,
\begin{align*}
  \dSwapHb & \leq d^* + \b\, d_1 - d_1 - \b\, d_2 \tag{$\WCchange_{\move{f^*}}$}\\
  & ~~~ + d_2 + \b(d^* + \nicefrac 1\rho\cdot(d^* + d_1)) - d_1 - \b\, d_2 \tag{$\WCchange_{\move{\neg f_1}}$}\\
  & ~~~ + d_1 + \b(d^* + \nicefrac 1\rho\cdot(d^* + d_1)) - d_1 - \b\, d_2 \tag{$\WCchange_{\move{\neg f_2}}$}\\
  & = \hoono{(1+2\b+\nicefrac{2\b}{\rho} )\, d^* 
    - (2-\b-\nicefrac{2\b}{\rho})\, d_1 +(1-3\b)\, d_2}.
\end{align*}

\paragraph{Tree swaps with $\tau(f^*) = \eta_1$} On $\cT_{11}\cap\cA$, $\move{f^*}$ is the only swap closing any facility in $\{f_1,f_2\}$ by the definition of $\cT_{11}$. In other words, both $\move{\neg f_1}$ and $\move{\neg f_2}$ coincide with $\move{f^*}$ as long as they exist. 
Therefore,
\begin{align*}
  \WCchange_{\cT_{11}\cap\cA}(c)\leq \hoono{(1+\a\b)\, d^* - d_1 - \b\, d_2}.
  \tag{$\WCchange_{\move{f^*, \neg f_1,\neg f_2}}$}
\end{align*}

On $\cT_{12}\cap\cA$, we have $\move{f^*} = \move{\neg f_1}\neq\move{\neg f_2}$ by implication (ii) of amenability. On swap $\move{f^*}$, the client can be served by $f^*$. On swap $\move{\neg f_2}$, the client can be served by $f_1$ and $\pi(f_2)$ by implication (Tii) of amenability. We have $d(c,\pi(f_2)) \leq  2d_2 + d^*$ by \eqref{eqn:pif2}. Therefore,
\begin{align*}
  \WCchange_{\cT_{12}\cap\cA}(c) & \leq (1+\a\b)\, d^* - d_1 - \b\, d_2 \tag{$\WCchange_{\move{f^*, \neg f_1}}$}\\
  & ~~~ + d_1 + \b(2\,d_2 + d^*) -d_1 - \b\, d_2 \tag{$\WCchange_{\move{\neg f_2}}$}\\
  & = \hoono{(1+\a\b+\b)\, d^* - d_1 }. 
\end{align*}

\paragraph{Tree swaps with $\tau(f^*) = \eta_2$} On $\cT_{21}\cap\cA$, we have $\move{\neg f_1} = \move{\neg f_2}\neq\move{f^*}$. On swap $\move{f^*}$, the client can be served by $f^*$ and $f_1$. On swap $\move{\neg f_1,\neg f_2}$, the client can be served by $\eta_2$ and $\pi(f_1)$ by implications (ii) and (Tii) of amenability. We have $d(c,\eta_2)\leq d^* + \nicefrac 1\rho\cdot(d^* + d_1)$ and $d(c,\pi(f_1))\leq 2d_1 + d^*$. Therefore,
\begin{align}
  \WCchange_{\cT_{21}\cap\cA}(c) 
  &
    \leq d^* + \b\, d_1 - d_1 - \b\, d_2 \tag{$\WCchange_{\move{f^*}}$}\\
  & ~~~ +  (d^* + \nicefrac{1}{\rho} (d^* +d_1)) + \b(2d_1 + d^*) - d_1 - \b\, d_2
    \tag{$\WCchange_{\move{\neg f_1,\neg f_2}}$}\\
  & = \hoono{(2+\b+\nicefrac{1}{\rho})\, d^* - (2-3\b-\nicefrac{1}{\rho})\, d_1 - 2\b\, d_2}. \nonumber
\end{align}

On $\cT_{22}\cap\cA$, we first consider the case where all three swaps $\move{f^*},\move{\neg f_1},\move{\neg f_2}$ are different.
 On swap $\move{f^*}$, the client can be served by $f^*$. On swap $\move{\neg f_1}$, the client can be served by $f_2$ and $\pi(f_1)$ (at distance $\leq 2d_1 + d^*$), by implication (Tii) of amenability. 
 On swap $\move{\neg f_2}$, the client can be served by $f_1$ and $\pi(f_2)$ (at distance $\leq 2d_2 + d^*$), again by implication (Tii) of amenability. Therefore,
\begin{align*}
  \WCchange_{\cT_{22}\cap\cA}(c) & \leq (1 + \a\b)\, d^* - d_1 - \b\, d_2 \tag{$\WCchange_{\move{f^*}}$}\\
  & ~~~ + d_2 + \b(2d_1 + d^*) - d_1 - \b\, d_2 \tag{$\WCchange_{\move{\neg f_1}}$}\\
  & ~~~ + d_1 + \b(2d_2 + d^*) - d_1 - \b\, d_2 \tag{$\WCchange_{\move{\neg f_2}}$}\\
  & \leq \hoono{(1+ \a\b + 2\b )\, d^* 
   - (2-2\b)\, d_1 +(1-\b)\, d_2}.
\end{align*}
Since our bound for $\WCchange_{\move{f^*}}(c)$ doesn't require either $f_1$ or $f_2$ to remain open after the swap,
and both $\WCchange_{\move{\neg f_1}}$ and $\WCchange_{\move{\neg f_2}}$ are non-negative,
 the above bound also holds when $\move{\neg f_1}$ and/or $\move{\neg f_2}$ coincides with $\move{f^*}$.

Summarizing, we have
\begin{mybox}
\begin{align*}
\dSwapHa\leq  & (1+\a\b + \b+\nicefrac{\b}{\rho})\, d^* - (1-\nicefrac{\b}{\rho})\, d_1 - 2\b\, d_2 \span\\
 & & = (1.8+0.2/\rho)\, d^* - (1-0.2/\rho)\, d_1 - 0.4\, d_2\\
\dSwapHb\leq  & (1+2\b+\nicefrac{2\b}{\rho})\, d^* - (2-\b-\nicefrac{2\b}{\rho})\, d_1 +(1-3\b)\, d_2\span \\
& &  = (1.4+0.4/\rho)\, d^* - (1.8-0.4/\rho)\, d_1 + 0.4\, d_2\\
\WCchange_{\cT_{11}\cap\cA}(c)\leq  & (1+\a\b)\, d^* - d_1 - \b\, d_2
& = 1.6\, d^* - d_1 - 0.2\, d_2\\
\WCchange_{\cT_{12}\cap\cA}(c)\leq  & (1+\a\b+\b)\, d^* - d_1 
 &  = 1.8\, d^* - d_1\\
\WCchange_{\cT_{21}\cap\cA}(c)\leq  & (2+ \b + \nicefrac{1}{\rho})\, d^* - (2-3\b-\nicefrac{1}{\rho})\, d_1 - 2\b\, d_2 \span \\
 & & = (2.2+\nicefrac{1}{\rho})\, d^* - (1.4-\nicefrac{1}{\rho})\, d_1 - 0.4\, d_2\\
\WCchange_{\cT_{22}\cap\cA}(c)\leq  & (1+\a\b + 2\b )\, d^* 
   - (2-2\b)\, d_1 +(1-\b)\, d_2 \span \\
& & = 2\, d^* - 1.6\, d_1 + 0.8\, d_2
 \end{align*}
\end{mybox}
We now combine these bounds to show an upper bound for $\change_\cA(c)$ using \Cref{lm:typeA}. Note that our bound for $\WCchange_{\cT_{11}\cap\cA}(c)$ is smaller than our bound for $\WCchange_{\cT_{12}\cap\cA}(c)$, so we only need to consider cases where the maximum in \Cref{lm:typeA} is attained at either $\WCchange_{\cT_{11}\cap\cA}(c) + \WCchange_{\cT_{21}\cap\cA}(c)$ or $\WCchange_{\cT_{12}\cap\cA}(c) + \WCchange_{\cT_{22}\cap\cA}(c)$.

When $\nicefrac 23 < \rho(f^*)\leq \nicefrac 34$, we have $\Pr[\cS_1] = \nicefrac 12$ and $\Pr[\cS_2] = 0$. Therefore, if the maximum in \Cref{lm:typeA} is attained at $\WCchange_{\cT_{11}\cap\cA}(c) + \WCchange_{\cT_{21}\cap\cA}(c)$, we have
\begin{align*}
\change_\cA(c) & \leq \nicefrac 12\cdot\WCchange_{\cS_1\cap\cA}(c) + \nicefrac 14\cdot(\WCchange_{\cT_{11}\cap\cA}(c) + \WCchange_{\cT_{21}\cap\cA}(c)) + \epsd\\
& \leq (1.85 + 0.35/\rho)d^* - (1.1 - 0.35/\rho)d_1 - 0.35 d_2+ \epsd\\
& \leq (1.85 + 0.35\times\nicefrac 32)d^* - (1.45 - 0.35 \times \nicefrac 32)d_1+ \epsd\tag{$d_2\geq d_1$}\\
& = \hoono{2.375d^* - 0.925d_1} + \epsd.
\end{align*}
If the maximum in \Cref{lm:typeA} is attained at $\WCchange_{\cT_{12}\cap\cA}(c) + \WCchange_{\cT_{22}\cap\cA}(c)$, we have
\begin{align*}
\change_\cA(c) & \leq \nicefrac 12\cdot\WCchange_{\cS_1\cap\cA}(c) + \nicefrac 14\cdot(\WCchange_{\cT_{12}\cap\cA}(c) + \WCchange_{\cT_{22}\cap\cA}(c)) + \epsd\\
& \leq (1.85 + 0.1 / \rho)d^* - (1.15 - 0.1 / \rho)d_1 + \epsd\\
& \leq (1.85 + 0.1\times\nicefrac 32)d^* - (1.15 - 0.1 \times \nicefrac 32)d_1+ \epsd\\
& = \hoono{2\, d^* - d_1} + \epsd.
\end{align*}
When $\rho(f^*)>\nicefrac 34$, we have $\Pr[\cS_1] = \nicefrac 54 - \rho$ and $\Pr[\cS_2] = \rho - \nicefrac 34$. Therefore, if the maximum in \Cref{lm:typeA} is attained at $\WCchange_{\cT_{11}\cap\cA}(c) + \WCchange_{\cT_{21}\cap\cA}(c)$, we have
\begin{align*}
\change_\cA(c) & \leq (\nicefrac 54 - \rho) \cdot\WCchange_{\cS_1\cap\cA}(c) + (\rho - \nicefrac 34) \cdot\WCchange_{\cS_2\cap\cA}(c) + \nicefrac 14\cdot(\WCchange_{\cT_{11}\cap\cA}(c) + \WCchange_{\cT_{21}\cap\cA}(c)) + \epsd\\
& \leq (2.35 + 0.2 / \rho - 0.4 \rho)d^* - (0.3 - 0.2 / \rho + 0.8 \rho)d_1 - (0.95 - 0.8\rho) d_2+ \epsd\\
& \leq (2.35 + 0.2 / \rho - 0.4 \rho)d^* - (1.25 - 0.2 / \rho)d_1 + \epsd\tag{$d_2\geq d_1$}\\
& \leq (2.35 + 0.2\times \nicefrac 43 - 0.4 \times\nicefrac 34)d^* - (1.25 - 0.2 \times \nicefrac 43)d_1+ \epsd\\
& \leq \hoono{2.31667\, d^* - 0.98333\, d_1} + \epsd.
\end{align*}
If the maximum in \Cref{lm:typeA} is attained at $\WCchange_{\cT_{12}\cap\cA}(c) + \WCchange_{\cT_{22}\cap\cA}(c)$, we have
\begin{align*}
\change_\cA(c) & \leq (\nicefrac 54 - \rho) \cdot\WCchange_{\cS_1\cap\cA}(c) + (\rho - \nicefrac 34) \cdot\WCchange_{\cS_2\cap\cA}(c) + \nicefrac 14\cdot(\WCchange_{\cT_{12}\cap\cA}(c) + \WCchange_{\cT_{22}\cap\cA}(c)) + \epsd\\
& \leq (2.35 - 0.05/\rho - 0.4 \rho)d^* - (0.35 + 0.05/\rho + 0.8 \rho)d_1 + (0.8\rho - 0.6) d_2+ \epsd\\
& \leq (2.55 - 0.65/\rho + 0.4\rho)d^* - (0.65/\rho + 0.8\rho - 0.45)d_1 + \epsd\tag{$d_2\leq d^* + \nicefrac 1\rho\cdot(d^* + d_1)$}\\
& \leq (2.55 - 0.65 + 0.4)d^* - (0.65\cdot \nicefrac{4}{\sqrt{13}} + 0.8\cdot \nicefrac{\sqrt{13}}{4} - 0.45)d_1+ \epsd\\
& \leq \hoono{2.3\, d^* - 0.99222\, d_1} + \epsd.
\end{align*}

\subsection{Proof of \eqref{eqn:bndB}: Close Clients of Type \xB}

In this section, we show that for any close case client $c$ with type \xB, we have
\[
\Gsum{c} \leq 2.4\, d^*(c) - 0.9\, d_1(c).
\]

In our type \xA analysis, we further partitioned the tree events $\cT_1$ and $\cT_2$ as $\cT_1 = \cT_{11}\cup\cT_{12}$ and $\cT_2 = \cT_{21}\cup\cT_{22}$. 
We require this partitioning also in our type \xB analysis, 
 with the roles of $\eta_1$ and $\eta_2$ flipped. 
Specifically, we define $\cT_{11}$ as the intersection of $\cT_1$ and the event that there is a swap which closes both $f_1$ and $f_2$ but does not open the original copy of $f^*$. We define $\cT_{21}$ as the intersection of $\cT_2$ and the event that $\move{f^*}$ is the only swap closing any facility in $\{f_1, f_2\}$. We define $\cT_{12}$ and $\cT_{22}$ accordingly as $\cT_{12} = \cT_1\backslash\cT_{11}$ and $\cT_{22} = \cT_2\backslash\cT_{21}$. Similar to \Cref{lm:typeA}, we have the following lemma for type \xB:

\begin{restatable}[Type \xB averaging]{lemma}{typeB}
\label{lm:typeB}
For a close client of type \xB with $\rho(f^*)>\nicefrac 23$, we have
\begin{equation*}
\change_{\cT\cap\cA}(c) \leq \nicefrac 14 \cdot \max\{\WCchange_{\cT_{11}\cap\cA} + \WCchange_{\cT_{21}\cap\cA}, \WCchange_{\cT_{12}\cap\cA} + \WCchange_{\cT_{21}\cap\cA}, \WCchange_{\cT_{12}\cap\cA} + \WCchange_{\cT_{22}\cap\cA}\} + \epsd.
\end{equation*}
\end{restatable}



We now proceed to bound the worst-case potential changes in different events.

\paragraph{Simple swaps with $\tau(f^*) = \eta_1$} By implications (Siii) and (Siv), all three swaps $\move{f^*},\move{\neg f_1},\move{\neg f_2}$ are different. On swap $\move{f^*}$, the client can be served by $f^*$ and $f_1$. On swap $\move{\neg f_1}$, the client can be served by $f_2$ and $\eta_1$ by implication (Siii) of amenability. On swap $\move{\neg f_2}$, the client can be served by $f_1$ and $\eta_1$, again by implication (Siii) of amenability. Note that $d(c,\eta_1)\leq d^* + \rho(d^* + d_1)$ by \eqref{eqn:apx:eta1}. Therefore,
\begin{align*}
  \dSwapHa & \leq d^* + \b\, d_1 - d_1 - \b\, d_2 \tag{$\WCchange_{\move{f^*}}$}\\
  & ~~~ + d_2 + \b(d^* + \rho(d^* + d_1)) - d_1 - \b\, d_2 \tag{$\WCchange_{\move{\neg f_1}}$}\\
  & ~~~ + d_1 + \b(d^* + \rho(d^* + d_1)) - d_1 - \b\, d_2 \tag{$\WCchange_{\move{\neg f_2}}$}\\
  & = \hoono{(1+2\b+2\rho\b )\, d^* - (2-\b-2\rho\b)\, d_1 +(1-3\b)\, d_2}.
\end{align*}

\paragraph{Simple swaps with $\tau(f^*) = \eta_2$} By implications (ii) and (Siv), we have $\move{f^*} = \move{\neg f_1} \neq \move{\neg f_2}$. On swap $\move{f^*}$, the client can be served by $f^*$. On swap $\move{\neg f_2}$, the client can be served by $f_1$ and $\eta_1$, by implication (Siii) of amenability. Therefore,
\begin{align*}
  \dSwapHb & \leq (1+\a\b)\, d^* - d_1 - \b\, d_2 \tag{$\WCchange_{\move{f^*, \neg f_1}}$}\\
  & ~~~ +  d_1 + \b(d^* + \rho(d^* + d_1)) - d_1 - \b\, d_2 \tag{$\WCchange_{\move{\neg f_2}}$} \\
  & = \hoono{(1+\a\b + \b + \rho\b)\, d^* - (1- \rho\b)\, d_1 - 2\b\, d_2}. 
\end{align*}

\paragraph{Tree swaps with $\tau(f^*) = \eta_1$} On $\cT_{11}\cap\cA$, we have $\move{\neg f_1} = \move{\neg f_2}\neq\move{f^*}$. On swap $\move{f^*}$, the client can be served by $f^*$ and $f_1$. On swap $\move{\neg f_1,\neg f_2}$, the client can be served by $\eta_1$ and $\pi(f_1)$ by implication (ii) and (Tii) of amenability. We have $d(c,\eta_1)\leq d^* + \rho(d^* + d_1)$ and $d(c,\pi(f_1))\leq 2d_1 + d^*$ by \eqref{eqn:ubpartner}. Therefore,
\begin{align*}
  \WCchange_{\cT_{11}\cap\cA}(c)
  &
    \leq d^* + \b\, d_1 - d_1 - \b\, d_2 \tag{$\WCchange_{\move{f^*}}$}\\
  & ~~~ +  (d^* + \rho (d^* +d_1)) + \b(2d_1 + d^*) - d_1 - \b\, d_2
    \tag{$\WCchange_{\move{\neg f_1,\neg f_2}}$}\\
  & \leq \hoono{(2+\b + \rho)\, d^* - (2-3\b - \rho)\, d_1 - 2\b\, d_2}.
\end{align*}

On $\cT_{12}\cap\cA$, we first consider the case where all three swaps $\move{f^*},\move{\neg f_1},\move{\neg f_2}$ are different. On swap $\move{f^*}$, the client can be served by $f^*$. On swap $\move{\neg f_1}$, the client can be served by $f_2$ and $\eta_1$. On swap $\move{\neg f_2}$, the client can be served by $f_1$ and $\eta_1$. After both $\move{\neg f_1}$ and $\move{\neg f_2}$, $\eta_1$ is open by implication (ii) of amenability. Therefore,
\begin{align*}
  \WCchange_{\cT_{12}\cap\cA}(c)
  &
    \leq (1 + \a\b)\,d^* - d_1 - \b\, d_2 \tag{$\WCchange_{\move{f^*}}$}\\
  & ~~~ +  d_2  + \b(d^* + \rho(d^* + d_1)) - d_1 - \b\, d_2
    \tag{$\WCchange_{\move{\neg f_1}}$}\\
  & ~~~ +  d_1  + \b(d^* + \rho(d^* + d_1)) - d_1 - \b\, d_2
    \tag{$\WCchange_{\move{\neg f_2}}$}\\
  & \leq \hoono{(1 + \a\b + 2\b + 2\rho\b)\, d^* - (2-2\rho\b)\, d_1 + (1 - 3\b)\, d_2}. 
\end{align*}
Since our bound for $\WCchange_{\move{f^*}}(c)$ doesn't require either $f_1$ or $f_2$ to remain open after the swap, the above bound also holds when $\move{\neg f_1}$ and/or $\move{\neg f_2}$ coincides with $\move{f^*}$.

\paragraph{Tree swaps with $\tau(f^*) = \eta_2$} On $\cT_{21}\cap\cA$, $\move{f^*}$ is the only swap closing any facility in $\{f_1,f_2\}$ by the definition of $\cT_{21}$. In other words, both $\move{\neg f_1}$ and $\move{\neg f_2}$ coincide with $\move{f^*}$ as long as they exist. 
Therefore,
\begin{equation*}
  \WCchange_{\cT_{21}\cap\cA}(c)\leq \hoono{(1+\a\b)\, d^* - d_1 - \b\, d_2}.\tag{$\WCchange_{\move{f^*, \neg f_1,\neg f_2}}$}
\end{equation*}

On $\cT_{22}\cap\cA$, we have $\move{f^*} = \move{\neg f_1} \neq \move{\neg f_2}$ by implication (ii) of amenability. On swap $\move{f^*}$, the client can be served by $f^*$. On swap $\move{\neg f_2}$, the client can be served by $f_1$ and $\pi(f_2)$. Therefore,
\begin{align*}
  \WCchange_{\cT_{22}\cap\cA}(c) & \leq (1+\a\b)\, d^* - d_1 - \b\, d_2 \tag{$\WCchange_{\move{f^*, \neg f_1}}$}\\
  & ~~~ + d_1 + \b(2d_2+d^*) -d_1 - \b\, d_2 \tag{$\WCchange_{\move{\neg f_2}}$}\\
  & = \hoono{(1+\a\b+\b)\, d^* - d_1}. 
\end{align*}

Summarizing, we have

\begin{mybox}
\begin{align*}
\dSwapHa\leq  & (1+2\b+2\rho\b)\, d^* -(2-\b-2\rho\b)\, d_1 + (1-3\b)\, d_2 \span \\
& &  = (1.4+0.4\rho)\, d^* - (1.8-0.4\rho)\, d_1 + 0.4\, d_2\\
\dSwapHb\leq  & (1+\a\b+\b+\rho\b)\, d^* - (1-\rho\b)\, d_1 -2\b\, d_2
&  = (1.8+0.2\rho)\, d^* - (1-0.2\rho)\, d_1 - 0.4\, d_2\\
\WCchange_{\cT_{11}\cap\cA}(c) \leq  & (2+\b + \rho)\, d^* - (2-3\b - \rho)\, d_1 - 2\b\, d_2
& = (2.2 + \rho)\, d^* - (1.4 - \rho)\, d_1 - 0.4\, d_2\\
\WCchange_{\cT_{12}\cap\cA}(c) \leq  & (1 + \a\b + 2\b + 2\rho\b)\, d^* - (2-2\rho\b)\, d_1 + (1 - 3\b)\, d_2\span\\
& & = (2 + 0.4\rho)\, d^* - (2 - 0.4\rho)\, d_1 + 0.4\, d_2\\
\WCchange_{\cT_{21}\cap\cA}(c) \leq  & (1+\a\b)\, d^* - d_1 - \b\, d_2
& = 1.6\, d^* - d_1 - 0.2\, d_2\\
\WCchange_{\cT_{22}\cap\cA}(c)\leq  & (1+\a\b+\b)\, d^* - d_1
&  = 1.8\, d^* - d_1
\end{align*}
\end{mybox}
Now we combine these inequalities to get an upper bound for $\change_{\cA}(c)$. When $\rho(f^*)\leq \nicefrac 23$, we have $\Pr[\cS_1] = \Pr[\cT_1] = \nicefrac 12$ and $\Pr[\cS_2] = \Pr[\cT_2] = 0$. Therefore, 
\begin{equation*}
\change_{\cA}(c)\leq \nicefrac 12 \cdot \WCchange_{\cS_1\cap\cA}(c) + \nicefrac 12 \cdot \max\{ \WCchange_{\cT_{11}\cap\cA}(c), \WCchange_{\cT_{12}\cap\cA}(c) \} + \epsd.
\end{equation*} 
If the maximum is attained at $\WCchange_{\cT_{11}\cap\cA}(c)$, we have
\begin{align*}
\change_{\cA}(c) & \leq \nicefrac 12\cdot \WCchange_{\cS_1\cap\cA}(c) + \nicefrac 12 \cdot  \WCchange_{\cT_{11}\cap\cA}(c) + \epsd\\
& \leq (1.8 + 0.7\rho)d^* - (1.6 - 0.7\rho)d_1 + \epsd\\
& \leq (1.8 + 0.7 \times \nicefrac 23)d^* - (1.6 - 0.7 \times \nicefrac 23)d_1 + \epsd\\
& = \hoono{2.26667d^* - 1.13333d_1} + \epsd.
\end{align*}
If the maximum is attained at $\WCchange_{\cT_{12}\cap\cA}(c)$, we have
\begin{align*}
\change_{\cA}(c) & \leq \nicefrac 12\cdot \WCchange_{\cS_1\cap\cA}(c) + \nicefrac 12 \cdot  \WCchange_{\cT_{12}\cap\cA}(c) + \epsd\\
& \leq (1.7 + 0.4 \rho)d^* - (1.9 - 0.4 \rho)d_1 + 0.4d_2 + \epsd\\
& \leq (1.9 + 0.6 \rho)d^* - (1.3 - 0.6 \rho)d_1 + \epsd\tag{$d_2\leq \nicefrac 12 \cdot(d^* + \rho(d^* + d_1)) + \nicefrac 12\cdot\alpha d_1$}\\
& \leq (1.9 + 0.6 \times \nicefrac 23)d^* - (1.3 - 0.6 \times \nicefrac 23)d_1 + \epsd\\
& = \hoono{2.3\,d^* - 0.9\,d_1} + \epsd.
\end{align*}

When $\rho(f^*)>\nicefrac 23$, we apply \Cref{lm:typeB} to combine the inequalities. Note that our bound for $\WCchange_{\cT_{21}\cap\cA}(c)$ is smaller than our bound for $\WCchange_{\cT_{22}\cap\cA}(c)$, so we only need to consider cases where the maximum in \Cref{lm:typeB} is attained at either $\WCchange_{\cT_{11}\cap\cA}(c) + \WCchange_{\cT_{21}\cap\cA}(c)$ or $\WCchange_{\cT_{12}\cap\cA}(c) + \WCchange_{\cT_{22}\cap\cA}(c)$.

When $\nicefrac 23 < \rho(f^*)\leq \nicefrac 34$, we have $\Pr[\cS_1] = \nicefrac 12,\Pr[\cS_2] = 0,\Pr[\cT_1] = \Pr[\cT_2] = \nicefrac 14$. If the maximum in \Cref{lm:typeB} is attained at $\WCchange_{\cT_{11}\cap\cA}(c) + \WCchange_{\cT_{21}\cap\cA}(c)$, we have
\begin{align*}
\change_{\cA}(c) & \leq \nicefrac 12\cdot \WCchange_{\cS_1\cap\cA}(c) + \nicefrac 14 \cdot  (\WCchange_{\cT_{11}\cap\cA}(c) +   \WCchange_{\cT_{21}\cap\cA}(c)) + \epsd\\
& \leq (1.65 + 0.45\rho)\,d^* - (1.5 - 0.45\rho)\,d_1 + 0.05\, d_2 + \epsd\\
& \leq (1.65 + 0.45\rho)\,d^* - (1.35 - 0.45\rho)\,d_1 + \epsd\tag{$d_2\leq\alpha d_1$}\\
& \leq (1.65 + 0.45 \times \nicefrac 34)d^* - (1.35 - 0.45 \times \nicefrac 34)d_1 + \epsd\\
& = \hoono{1.9875d^* - 1.0125d_1} + \epsd.
\end{align*}

If the maximum in \Cref{lm:typeB} is attained at $\WCchange_{\cT_{12}\cap\cA}(c) + \WCchange_{\cT_{22}\cap\cA}(c)$, we have
\begin{align*}
\change_{\cA}(c) & \leq \nicefrac 12\cdot \WCchange_{\cS_1\cap\cA}(c) + \nicefrac 14 \cdot  (\WCchange_{\cT_{12}\cap\cA}(c) +   \WCchange_{\cT_{22}\cap\cA}(c)) + \epsd\\
& \leq (1.65 + 0.3\rho)\,d^* - (1.65 - 0.3\rho)\,d_1 + 0.3\, d_2 + \epsd\\
& \leq (1.95 + 0.6\rho)\,d^* - (1.65 - 0.6\rho)\,d_1 + \epsd\tag{$d_2\leq d^* + \rho(d^* + d_1)$}\\
& \leq (1.95 + 0.6 \times \nicefrac 34)d^* - (1.65 - 0.6 \times \nicefrac 34)d_1 + \epsd\\
& = \hoono{2.4d^* - 1.2d_1} + \epsd.
\end{align*}

When $\rho(f^*)> \nicefrac 34$, we have $\Pr[\cS_1] = \nicefrac 54 - \rho,\Pr[\cS_2] = \rho - \nicefrac 34,\Pr[\cT_1] = \Pr[\cT_2] = \nicefrac 14$. If the maximum in \Cref{lm:typeB} is attained at $\WCchange_{\cT_{11}\cap\cA}(c) + \WCchange_{\cT_{21}\cap\cA}(c)$, we have
\begin{align*}
\change_{\cA}(c) & \leq (\nicefrac 54 - \rho)\cdot \WCchange_{\cS_1\cap\cA}(c) + (\rho - \nicefrac 34)\cdot \WCchange_{\cS_2\cap\cA}(c) + \nicefrac 14 \cdot  (\WCchange_{\cT_{11}\cap\cA}(c) +   \WCchange_{\cT_{21}\cap\cA}(c)) + \epsd\\
& \leq (1.35 + \rho - 0.2\rho^2)\,d^* - (2.1 - 1.4\rho + 0.2\rho^2)\,d_1 + (0.65 - 0.8\rho)\, d_2 + \epsd.
\end{align*}
When $\rho < 0.8125 = \nicefrac{0.65}{0.8}$, we use $d_2\leq\alpha d_1$:
\begin{align*}
\change_{\cA}(c) & \leq (1.35 + \rho - 0.2\rho^2)\,d^* - (0.15 + \rho + 0.2\rho^2)\,d_1 + \epsd\\
& \leq (1.35 + 0.8125 - 0.2 \times 0.8125^2)d^* - (0.15 + \nicefrac 34 + 0.2 \times (\nicefrac 34)^2)d_1 + \epsd\\
& = \hoono{2.03047d^* - 1.0125d_1} + \epsd.
\end{align*}
When $\rho \geq 0.8125$, we use $d_2 \geq d_1$:
\begin{align*}
\change_{\cA}(c) & \leq (1.35 + \rho - 0.2\rho^2)\,d^* - (1.45 - 0.6\rho + 0.2\rho^2)\,d_1 + \epsd\\
& \leq (1.35 + 1 - 0.2)d^* - (1.45 - 0.6 + 0.2 )d_1 + \epsd\\
& = \hoono{2.15\,d^* - 1.05\,d_1} + \epsd.
\end{align*}

If the maximum in \Cref{lm:typeB} is attained at $\WCchange_{\cT_{12}\cap\cA}(c) + \WCchange_{\cT_{22}\cap\cA}(c)$, we have
\begin{align*}
\change_{\cA}(c) & \leq (\nicefrac 54 - \rho)\cdot \WCchange_{\cS_1\cap\cA}(c) + (\rho - \nicefrac 34)\cdot \WCchange_{\cS_2\cap\cA}(c) + \nicefrac 14 \cdot  (\WCchange_{\cT_{12}\cap\cA}(c) +   \WCchange_{\cT_{22}\cap\cA}(c)) + \epsd\\
& \leq (1.35 + 0.85\rho - 0.2\rho^2)\,d^* - (2.25 - 1.25\rho + 0.2\rho^2)\,d_1 + (0.9 - 0.8\rho)\, d_2 + \epsd\\
& \leq (2.25 + 0.95\rho - \rho^2)\,d^* - (2.25 - 2.15\rho + \rho^2)\,d_1 + \epsd\tag{$d_2\leq d^* + \rho(d^* + d_1)$}\\
& \leq (2.25 + 0.95 \times \nicefrac 34  - (\nicefrac 34)^2)d^* - (2.25 - 2.15  + 1)d_1 + \epsd\\
& = \hoono{2.4d^* - 1.1d_1} + \epsd.
\end{align*}

\subsection{Proof of \eqref{eqn:bndC}:  Clients of Type \xC}

In this section, we show that for any client $c$ with type \xC, we have
\[
\Gsum{c} \leq 2.2\, d^*(c) - 0.8888\, d_1(c).
\]
%
%
%
If the client $c$ satisfies $\rho(f^*) \leq \nicefrac{2}{3}$, we have the same bound as in the type \xA case in \Cref{sec:typeAlow}, where our analysis was independent of whether $f_2 = \eta_2$ or not. That is
\begin{equation*}
\change_\cA(c) \leq \hoono{2.1\, d^*(c) - 0.9\, d_1(c)} + \epsd.
\end{equation*}

We thus focus on clients with $\rho(f^*) > \nicefrac 23$. Compared to our analysis for other client types, our analysis for type \xC involves a larger neighborhood of the client. In particular, the optimal facility $g^*:=\pi(f_1)$ and the local facilities close to it play a crucial role in our analysis. This makes it important to consider finer-grained events. Recall that we used $\cS_1,\cS_2,\cT_1,\cT_2$ to denote simple/tree events restricted to $f^*$ pointing to $\eta_1$ or $\eta_2$. We now also define events $\cS_1',\cS_2',\cT_1',\cT_2'$ similarly, except that they depend on where $g^*$ points to, rather than $f^*$. We classify clients into subtypes according to the characteristics of the swap sets generated on these events:

\begin{restatable}[Subtypes within type \xC]{claim}{typeThreeSwapcases}
  \label{claim:type3Swapcases}
  For a client $c$ of type \xC, one of the following is true:
  \begin{enumerate}[(a)]
    \item  $f_1$ is \heavy.
    \item  $f_2$ is \heavy.
    \item A facility $h$ is open near $c$ after the simple swap closing $f_1$. Formally, a facility $h\neq f_2$ is open after swap $\move{\neg f_1}$ at distance $d(c,h)\leq 3d_1 + 2d^*$ on $\cS\cap\cA$.
	\item $g^* \neq f^*$, $\rho(g^*) > \nicefrac 34$, and for all $b = 1,2$, any swap set $\cP$ generated on $\cS_b'\cap\cA$, a facility $h\neq f_2$ is open after swap $\move{\neg f_1}$ at distance $d(c,h)\leq \left\{\begin{array}{ll}2d_1 + d^*,& \textup{if}~b = 1\\ 2d_1 + d^* + \nicefrac 43(d_1 + d^*),& \textup{if}~b=2\end{array}\right.$.
    \item For any swap set $\cP$ generated on $\cT_2\cap\cA$, $\move{f^*}$ closes both $f_1$ and $f_2$.
\item $g^*\neq f^*$, $\rho(g^*)>\nicefrac{2}{3}$, and there exists $b\in\{1,2\}$ such that for any swap set $\cP$ generated on $\cT_{b}'\cap\cA$, $\move{f^*}$ closes both $f_1$ and $f_2$.
  \end{enumerate}
\end{restatable}

We prove this claim in \Cref{sec:proof-typeC}. Below we present our bounds for each of these subtypes.

%

\subsubsection{When \texorpdfstring{$f_1$}{[f1]} is a heavy facility}
$f_1$ being heavy implies that the swap $\move{\neg f_1}$ doesn't exist. We thus focus on $\move{f^*}$ and $\move{\neg f_2}$.

\paragraph{Simple swaps with $\tau(f^*) = \eta_1$} By implication (Siv) of amenability, we have $\move{f^*}\neq \move{\neg f_2}$. On swap $\move{f^*}$, the client can be served by $f^*$, and on swap $\move{\neg f_2}$, the client can be served by $f_1$. Therefore,
\begin{align*}
\dSwapHa  \leq {} & (1+\a\b)d^* - d_1 - \b\, d_2 \tag{$\WCchange_{\move{f^*}}$}\\
& +  (1+\a\b)d_1 - d_1 -\b\, d_2 \tag{$\WCchange_{\move{\neg f_2}}$}\\
= {} & \hoono{(1+\a\b)d^*-(1-\a\b)d_1 -2\b d_2}.
\end{align*}

\paragraph{Tree swaps with $\tau(f^*) = \eta_1$}Let us first assume $\move{f^*}\neq \move{\neg f_2}$. On swap $\move{f^*}$, the client can be served by $f^*$, and on swap $\move{\neg f_2}$, the client can be served by $f_1$ and $\pi(f_2)$ by implication (Tii) of amenability. Note that $d(c,\pi(f_2))\leq 2d_2 + d^*$ by \eqref{eqn:pif2}. Therefore,
\begin{align*}
\dSwapTa \leq {} & (1+\a\b)d^* - d_1 - \b\, d_2 \tag{$\WCchange_{\move{f^*}}$}\\
& + d_1+ \b(2d_2+d^*)- d_1 -\b\, d_2 \tag{$\WCchange_{\move{\neg f_2}}$}\\
= {} & \hoono{(1+\a\b+\b)d^* -  d_1}.
\end{align*}

The inequality also holds when $\move{f^*} = \move{\neg f_2}$ since our bound for $\WCchange_{\move{f^*}}$ does not require $f_2$ to remain open after the swap.

\paragraph{Simple \& tree swaps with $\tau(f^*) = \eta_2$}We have $\move{f^*} = \move{\neg f_2}$ by implication (ii) of amenability. On that swap, the client can be served by $f^*$. Therefore,
\begin{align*}
\dSwapHb \leq {} & \hoono{(1+\a\b)d^* - d_1 -  \b\, d_2}, \tag{$\WCchange_{\move{f^*,\neg f_2}}$}\\
\dSwapTb \leq {} & \hoono{(1+\a\b)d^* - d_1 -  \b\, d_2}. \tag{$\WCchange_{\move{f^*,\neg f_2}}$}
\end{align*}

Summarizing, we have

\begin{mybox}
\begin{align*}
\dSwapHa\leq  & (1+\a\b)\, d^* - (1-\a\b) d_1 - 2 \b\, d_2
&  = 1.6\, d^* - 0.4 d_1 - 0.4\, d_2\\
\dSwapHb\leq  & (1+\a\b)\, d^* -  d_1 - \b\, d_2
&  = 1.6\, d^* -  d_1 - 0.2\, d_2\\
\dSwapTa\leq  & (1+\a\b+\b)\, d^* -d_1
&  = 1.8\, d^* - d_1\\
\dSwapTb\leq  & (1+\a\b)\, d^*-  d_1 -\b\, d_2
&  = 1.6\, d^* -  d_1 - 0.2\, d_2
\end{align*}
\end{mybox}
We now combine these inequalities to get an upper bound for $\change_\cA(c)$.

When $\nicefrac 23 < \rho(f^*) \leq \nicefrac 34$, we have $\Pr[\cS_1] = \nicefrac 12,\Pr[\cS_2] = 0,\Pr[\cT_1] = \Pr[\cT_2] = \nicefrac 14$. Therefore,
\begin{align*}
\change_\cA(c) & \leq \nicefrac 12\cdot \WCchange_{\cS_1\cap\cA}(c) + \nicefrac 14 \cdot \WCchange_{\cT_1\cap\cA}(c) + \nicefrac 14 \cdot \WCchange_{\cT_2\cap\cA}(c) + \epsd \\
& \leq 1.65\, d^* - 0.7\, d_1 - 0.25\, d_2 + \epsd\\
& \leq \hoono{1.65\, d^* - 0.95\, d_1} + \epsd.\tag{$d_2\geq d_1$}
\end{align*}

When $\rho(f^*) > \nicefrac 34$, we have $\Pr[\cS_1] = \nicefrac 54 - \rho, \Pr[\cS_2] = \rho - \nicefrac 34, \Pr[\cT_1] = \Pr[\cT_2] = \nicefrac 14$. Therefore,
\begin{align*}
\change_\cA(c) & \leq (\nicefrac 54 - \rho)\cdot \WCchange_{\cS_1\cap\cA}(c) + (\rho - \nicefrac 34)\cdot \WCchange_{\cS_1\cap\cA}(c) + \nicefrac 14 \cdot \WCchange_{\cT_1\cap\cA}(c) + \nicefrac 14 \cdot \WCchange_{\cT_2\cap\cA}(c) + \epsd \\
& \leq 1.65\, d^* - (0.6\rho + 0.25)\, d_1 - (0.4 - 0.2\rho)\, d_2 + \epsd\\
& \leq 1.65\, d^* - (0.4\rho + 0.65)\, d_1 + \epsd\tag{$d_2\geq d_1$}\\
& \leq 1.65\, d^* - (0.4 \times \nicefrac 34 + 0.65)\, d_1 + \epsd\\
& = \hoono{1.65\, d^* - 0.95\, d_1} + \epsd.
\end{align*}
%
%
%
%
%
%
%
%
%
\subsubsection{When \texorpdfstring{$f_2$}{[f2]} is a heavy facility}
$f_2$ being heavy implies that $\move{\neg f_2}$ does not exist. We thus focus on $\move{f^*}$ and $\move{\neg f_1}$.

\paragraph{Simple \& tree swaps with $\tau(f^*) = \eta_1$}We have $\move{f^*} = \move{\neg f_1}$ by implication (ii) of amenability. On that swap, the client can be served by $f^*$. Therefore,
\begin{align*}
\dSwapHa \leq {} & \hoono{(1+\a\b)d^* - d_1 -  \b\, d_2}, \tag{$\WCchange_{\move{f^*,\neg f_1}}$}\\
\dSwapTa \leq {} & \hoono{(1+\a\b)d^* - d_1 -  \b\, d_2}. \tag{$\WCchange_{\move{f^*,\neg f_1}}$}
\end{align*}

\paragraph{Simple swaps with $\tau(f^*) = \eta_2$} Implication (Siv) of amenability implies that $\move{f^*}\neq \move{\neg f_1}$. On swap $\move{f^*}$, the client can be served by $f^*$. On swap $\move{\neg f_1}$, the client can be served by $f_2$. Therefore,
\begin{align}
\dSwapHb \leq {} & (1+\a\b)d^* - d_1 -  \b\, d_2 \tag{$\WCchange_{\move{f^*}}$}\\
& + (1+\a\b)d_2 - d_1 -\b\, d_2 \tag{$\WCchange_{\move{\neg f_1}}$}\notag\\
= {} & \hoono{(1+\a\b)d^* - 2\, d_1 + (1+\a\b-2\b)d_2}.\label{eq:typeCsimple2}
\end{align}

\paragraph{Tree swaps with $\tau(f^*) = \eta_2$} We first assume that $\move{f^*} \neq \move{\neg f_1}$. On swap $\move{f^*}$, the client can be served by $f^*$. On swap $\move{\neg f_1}$, the client can be served by $f_2$ and $\pi(f_1)$, by implication (Tii) of amenability. Note that $d(c,\pi(f_1))\leq d_1 + d(f_1,\pi(f_1))\leq 2d_1 + d^*$. Therefore,
\begin{align*}
\dSwapTb \leq {} & (1+\a\b)d^* - d_1 - \b\, d_2 \tag{$\WCchange_{\move{f^*}}$}\\
& + d_2+ \b(2d_1+d^*)- d_1 -\b\, d_2 \tag{$\WCchange_{\move{\neg f_1}}$}\\
= {} & \hoono{(1+\a\b+\b)d^* - (2-2\b)\, d_1 + (1-2\b)d_2}.
\end{align*}
The above inequality also holds when $\move{f^*} = \move{\neg f_1}$ because our bound for $\WCchange_{\move{f^*}}$ does not require $f_1$ to remain open after the swap.
and $\WCchange_{\move{\neg f_1}}$ is  non-negative.

Summarizing, we have
\begin{mybox}
\begin{align*}
\dSwapHa\leq  & (1+\a\b)\, d^* -  d_1 - \b\, d_2
&  = 1.6\, d^* - d_1 - 0.2\, d_2\\
\dSwapHb\leq  & (1+\a\b)\, d^* - 2\, d_1 + (1+\a\b-2\b)\, d_2
&  = 1.6\, d^* - 2\, d_1 + 1.2\, d_2\\
\dSwapTa\leq  & (1+\a\b)\, d^* -d_1 - \b\, d_2
&  = 1.6\, d^* - d_1-0.2\, d_2\\
\dSwapTb\leq  & (1+\a\b+\b)\, d^*-(2-2\b)\, d_1 + (1-2\b)\, d_2
&  = 1.8\, d^* - 1.6\, d_1 + 0.6\, d_2
\end{align*}
\end{mybox}
We now combine these inequalities to get an upper bound for $\change_\cA(c)$.

When $\nicefrac 23 < \rho(f^*) \leq \nicefrac 34$, we have $\Pr[\cS_1] = \nicefrac 12,\Pr[\cS_2] = 0,\Pr[\cT_1] = \Pr[\cT_2] = \nicefrac 14$. Therefore,
\begin{align*}
\change_\cA(c) & \leq \nicefrac 12\cdot \WCchange_{\cS_1\cap\cA}(c) + \nicefrac 14 \cdot \WCchange_{\cT_1\cap\cA}(c) + \nicefrac 14 \cdot \WCchange_{\cT_2\cap\cA}(c) + \epsd \\
& \leq \hoono{1.65\, d^* - 1.15\, d_1} + \epsd.
\end{align*}

When $\rho(f^*) > \nicefrac 34$, we have $\Pr[\cS_1] = \nicefrac 54 - \rho,\Pr[\cS_2] = \rho - \nicefrac 34,\Pr[\cT_1] = \Pr[\cT_2] = \nicefrac 14$. Therefore,
\begin{align*}
\change_\cA(c) & \leq (\nicefrac 54 - \rho)\cdot \WCchange_{\cS_1\cap\cA}(c) + (\rho - \nicefrac 34)\cdot \WCchange_{\cS_2\cap\cA}(c) + \nicefrac 14 \cdot \WCchange_{\cT_1\cap\cA}(c) + \nicefrac 14 \cdot \WCchange_{\cT_2\cap\cA}(c) + \epsd \\
& \leq 1.65\, d^* - (\rho + 0.4)\, d_1 + (1.4\rho - 1.05)\, d_2 + \epsd\\
& \leq (2 + 1.4\rho - 1.05 / \rho)\, d^* - (\rho + 1.05/\rho - 1)\, d_1 + \epsd\tag{$\rho >\nicefrac 34$ and $d_2\leq d^* + \nicefrac 1\rho (d^* + d_1)$}\\
& \leq (2 + 1.4 - 1.05)\, d^* - (1 + 1.05 - 1)\, d_1 + \epsd\\
& = \hoono{2.35\, d^* - 1.05\, d_1} + \epsd.
\end{align*}
%
%
%
%
%
%
%
%
%

\subsubsection{\texorpdfstring{There exists $h$ such that $d(c,h)\leq 3d_1 + 2d^*$}{d(c,h) <= 3d1
    + 2d*]} in simple swaps}
\label{sec:typeCsimple}
\paragraph{Simple swaps with $\tau(f^*) = \eta_1$} By implications (ii) and (Siv) of amenablity, we know $\move{f^*} = \move{\neg f_1}\neq \move{\neg f_2}$. On swap $\move{f^*,\neg f_1}$, the client can be served by $f^*$. On swap $\move{\neg f_2}$, the client can be served by $f_1$. Therefore,
\begin{align*}
  \dSwapHa\leq {} & (1+\a\b)\, d^* - d_1 - \b\, d_2 \tag{$\WCchange_{\move{f^*, \neg f_1}}$}\\
  & +  (1+\a\b)\, d_1  - d_1 - \b\, d_2 \tag{$\WCchange_{\move{\neg f_2}}$}\\
  = {} & \hoono{(1+\a\b)\, d^* - (1- \a\b)\, d_1 - 2\b\, d_2}.
\end{align*}

\paragraph{Simple swaps with $\tau(f^*) = \eta_2$} By implications (ii) and (Siv) of amenablity, we know $\move{f^*} = \move{\neg f_2}\neq \move{\neg f_1}$. On swap $\move{f^*,\neg f_2}$, the client can be served by $f^*$. On swap $\move{\neg f_1}$, the client can be served by $f_2$ and $h$. Therefore,
\begin{align*}
  \dSwapHb\leq {} & (1+\a\b)\, d^* - d_1 - \b\, d_2 \tag{$\WCchange_{\move{f^*, \neg f_2}}$}\\
  & + d_2 + \b(3d_1 + 2d^*)-d_1 - \b\, d_2 \tag{$\WCchange_{\move{\neg f_1}}$}\\
  = {} & \hoono{(1+\a\b+2\b)\, d^* - (2-3\b)\, d_1 + (1-2\b)\, d_2}.
\end{align*}

\paragraph{Tree swaps with $\tau(f^*) = \eta_1$} By implication (ii) of amenability, we know $\move{f^*} = \move{\neg f_1}$. Let us first assume that $\move{\neg f_2}\neq \move{f^*,\neg f_1}$. On swap $\move{f^*,\neg f_1}$, the client can be served by $f^*$. On swap $\move{\neg f_2}$, the client can be served by $f_1$ and $\pi(f_2)$ by implication (Tii) of amenability. We have $d(c,\pi(f_2))\leq 2d_2 + d^*$ by \eqref{eqn:pif2}. Therefore,
\begin{align*}
  \dSwapTa\leq {} & (1+\a\b)\, d^* - d_1  - \b\, d_2 \tag{$\WCchange_{\move{f^*, \neg f_1}}$}\\
  & +  d_1 + \b(2d_2 + d^*) - d_1 - \b\, d_2 \tag{$\WCchange_{\move{\neg f_2}}$}\\
  = {} & \hoono{(1+\a\b+\b)\, d^* -d_1}.
\end{align*}
This inequality also holds when $\move{\neg f_2} = \move{f^*,\neg f_1}$, because our bound for $\WCchange_{\move{f^*,\neg f_1}}$ does not require $f_2$ to remain open after the swap
and $\WCchange_{\move{\neg f_2}}$ is non-negative.

\paragraph{Tree swaps with $\tau(f^*) = \eta_2$} By implication (ii) of amenability, we know $\move{f^*} = \move{\neg f_2}$. Again, let us first assume that $\move{\neg f_1}\neq \move{f^*,\neg f_2}$. On swap $\move{f^*,\neg f_2}$, the client can be served by $f^*$. On swap $\move{\neg f_1}$, the client can be served by $f_2$ and $\pi(f_1)$ by implication (ii) of amenability. We have $d(c,\pi(f_1))\leq 2d_1 + d^*$ by \eqref{eqn:ubpartner}. Therefore,
\begin{align*}
  \dSwapTb\leq {} & (1+\a\b)\, d^* - d_1 - \b\, d_2 \tag{$\WCchange_{\move{f^*, \neg f_2}}$}\\
  & + d_2 + \b(2d_1 + d^*) - d_1 - \b\, d_2 \tag{$\WCchange_{\move{\neg f_1}}$}\\
  = {} & \hoono{(1+\a\b+\b)\, d^* - (2-2\b)\, d_1 + (1-2\b)\, d_2}.
\end{align*}
This inequality also holds when $\move{\neg f_1} = \move{f^*,\neg f_2}$, because our bound for $\WCchange_{\move{f^*,\neg f_2}}$ does not require $f_1$ to remain open after the swap
and $\WCchange_{\move{\neg f_1}}$ is non-negative.

Summarizing, we have

\begin{mybox}
\begin{align*}
\dSwapHa\leq  & (1+\a\b)\, d^* - (1-\a\b)\, d_1 - 2\b\, d_2
& = 1.6\, d^* - 0.4\, d_1 - 0.4\, d_2\\
\dSwapHb\leq  & (1+\a\b+2\b)\, d^* - (2-3\b)\, d_1 + (1-2\b)\, d_2
&  = 2d^* -1.4\, d_1 + 0.6\, d_2\\
\dSwapTa\leq  & (1+\a\b+\b)\, d^* -d_1
&  = 1.8\, d^* - d_1\\
\dSwapTb\leq  & (1+\a\b+\b)\, d^* - (2-2\b)\, d_1 + (1-2\b)\, d_2
&  = 1.8\, d^* - 1.6\, d_1 + 0.6\, d_2
\end{align*}
\end{mybox}
We now combine these inequalities to get an upper bound for $\change_\cA(c)$.

When $\nicefrac 23 < \rho(f^*) \leq \nicefrac 34$, we have $\Pr[\cS_1] = \nicefrac 12,\Pr[\cS_2] = 0,\Pr[\cT_1] = \Pr[\cT_2] = \nicefrac 14$. Therefore,
\begin{align*}
\change_\cA(c) & \leq \nicefrac 12\cdot \WCchange_{\cS_1\cap\cA}(c) + \nicefrac 14 \cdot \WCchange_{\cT_1\cap\cA}(c) + \nicefrac 14 \cdot \WCchange_{\cT_2\cap\cA}(c) + \epsd \\
& \leq 1.7\, d^* - 0.85\, d_1 - 0.05\, d_2 + \epsd\\
& \leq \hoono{1.7\, d^* - 0.9\, d_1} + \epsd.\tag{$d_2\geq d_1$}
\end{align*}

When $\rho(f^*) > \nicefrac 34$, we have $\Pr[\cS_1] = \nicefrac 54 - \rho,\Pr[\cS_2] = \rho - \nicefrac 34, \Pr[\cT_1] = \Pr[\cT_2] = \nicefrac 14$. Therefore,
\begin{align*}
\change_\cA(c) & \leq (\nicefrac 54 - \rho )\cdot \WCchange_{\cS_1\cap\cA}(c) + (\rho - \nicefrac 34)\cdot \WCchange_{\cS_2\cap\cA}(c) + \nicefrac 14 \cdot \WCchange_{\cT_1\cap\cA}(c) + \nicefrac 14 \cdot \WCchange_{\cT_2\cap\cA}(c) + \epsd\\
& \leq (1.4+0.4\rho)\, d^* - (0.1+\rho)\, d_1 + (\rho-0.8)\, d_2 + \epsd.
\end{align*}
When $\rho \leq 0.8$, we use $d_2\geq d_1$:
\begin{align*}
\change_\cA(c) & \leq (1.4 + 0.4\rho)\, d^* - 0.9\, d_1 + \epsd\\
& \leq \hoono{1.72\, d^* - 0.9\, d_1} + \epsd.
\end{align*}
When $\rho > 0.8$, we use $d_2 \leq d^* + d(f^*,d_2)\leq d^* + \nicefrac 1\rho\cdot (d^* + d_1)$:
\begin{align*}
\change_\cA(c) & \leq (1.6 + 1.4\rho - 0.8/\rho)\, d^* - (\rho + 0.8/\rho - 0.9)\, d_1 + \epsd\\
& \leq (1.6 + 1.4 - 0.8)\, d^* - (2\sqrt{0.8} - 0.9)\, d_1 + \epsd\\
& \leq \hoono{2.2\, d^* - 0.88885\, d_1} + \epsd.
\end{align*}
%
%
%
%
%
%

\subsubsection{\texorpdfstring{$d(c,h)\leq 2d_1 + d^*$ or $d(c,h)\leq
    2d_1 + d^* + \nicefrac 43(d^* + d_1)$ in simple swaps}{[d(c,h) <=
    2d1 + d* or d(c,h) <= 2d1 + d* + 4/3(d* + d1) in simple swaps]}}
\label{sec:typeCsimple2}
We have the same bound for $\change_\cA(c)$ in this case as the previous case. Indeed, our previous bounds for $\WCchange_{\cS_1\cap\cA}(c),\WCchange_{\cT_1\cap\cA}(c)$ and $\WCchange_{\cT_2\cap\cA}(c)$ remain valid. We replace our bound for $\WCchange_{\cS_2\cap\cA}(c)$ by a bound for
\begin{equation}
\label{eq:typeCsimple-expectation}
\WCchange'_{\cS_2\cap\cA}(c):=\Pr[\cS_1'|\cS_2]\WCchange_{\cS_1'\cap\cS_2\cap\cA}(c) + \Pr[\cS_2'|\cS_2]\WCchange_{\cS_2'\cap\cS_2\cap\cA}(c).
\end{equation}
We show that we can upper-bound $\WCchange'_{\cS_2\cap\cA}(c)$ by the same expression as in \eqref{eq:typeCsimple2}. Our previous bound for $\WCchange_{\cS_2\cap\cA}(c)$ is linear in $d(c,h)$ with a non-negative coefficient: $\WCchange_{\cS_2\cap\cA}(c) \leq A\cdot d(c,h) + B$ with $A\geq 0$, so
\begin{align*}
\WCchange_{\cS_1'\cap\cS_2\cap\cA}(c) & \leq A\cdot (2d_1 + d^*) + B\\
\WCchange_{\cS_2'\cap\cS_2\cap\cA}(c) & \leq A\cdot (2d_1 + d^* + \nicefrac 43(d^* + d_1)) + B.
\end{align*}
Plugging them into \eqref{eq:typeCsimple-expectation}, we have
\begin{align*}
\WCchange'_{\cS_2\cap\cA}(c) & \leq A\cdot(\Pr[\cS_1'|\cS_2]\cdot (2d_1 + d^*) + \Pr[\cS_2'|\cS_2]\cdot (2d_1 + d^* + \nicefrac 43(d^* + d_1))) + B\\
& \leq A\cdot (\nicefrac 12\cdot (2d_1 + d^*) + \nicefrac 12 \cdot (2d_1 + d^* + \nicefrac 43(d^* + d_1))) + B\\
& \leq A\cdot (3d_1 + 2d^*) + B.
\end{align*}

\subsubsection{\texorpdfstring{$\move{f^*}$}{[move(f*)]} closes
  \texorpdfstring{$f_1$}{[f1]} and \texorpdfstring{$f_2$}{[f2]} on
  \texorpdfstring{$\cT_2\cap\cA$}{[event T2 intersect A]}}
If $\tau(f^*) = \eta_1$, we get the same bounds as before:
\begin{align*}
\WCchange_{\cS_1\cap\cA}(c) & \leq \hoono{(1+\a\b)\, d^* - (1- \a\b)\, d_1 - 2\b\, d_2},\\
\WCchange_{\cT_1\cap\cA}(c) & \leq \hoono{(1+\a\b+\b)\, d^* -d_1}.
\end{align*}
We continue to bound $\WCchange_{\cS_2\cap\cA}(c)$ and $\WCchange_{\cT_2\cap\cA}(c)$.

\paragraph{Simple swaps with $\tau(f^*) = \eta_2$} By implications (ii) and (Siv), we have $\move{f^*} = \move{\neg f_2} \neq \move{\neg f_1}$. On swap $\move{f^*,\neg f_2}$, the client can be served by $f^*$. On swap $\move{\neg f_1}$, the client can be served by $f_2$. Therefore,
\begin{align*}
  \dSwapHb\leq & (1+\a\b)\, d^* - d_1 - \b\, d_2 \tag{$\WCchange_{\move{f^*, \neg f_2}}$}\\
  & + (1+ \a\b)\, d_2-d_1 - \b\, d_2 \tag{$\WCchange_{\move{\neg f_1}}$}\\
  = {} & \hoono{(1+\a\b)\, d^* - 2\, d_1 + (1+\a\b-2\b)\, d_2}.
\end{align*}

\paragraph{Tree swaps with $\tau(f^*) = \eta_2$} On $\cT_2\cap\cA$, we know $\move{f^*}$ closes both $f_1$ and $f_2$. Therefore,
\begin{align*}
  \dSwapTb\leq \hoono{(1+\a\b)\, d^* - d_1 - \b\, d_2}. \tag{$\WCchange_{\move{f^*, \neg f_1,\neg f_2}}$}
\end{align*}

Summarizing, we have

\begin{mybox}
\begin{align*}
\dSwapHa\leq  & (1+\a\b)\, d^* - (1-\a\b)\, d_1 - 2\b\, d_2
&  = 1.6\, d^* - 0.4\, d_1 - 0.4\, d_2\\
\dSwapHb\leq  & (1+\a\b)\, d^* - 2\, d_1 + (1+\a\b-2\b)\, d_2
&  = 1.6\, d^* - 2\, d_1 + 1.2\, d_2\\
\dSwapTa\leq  & (1+\a\b+\b)\, d^* -d_1
&  = 1.8\, d^* - d_1\\
\dSwapTb\leq  & (1+\a\b)\, d^*-d_1 - \b\, d_2
&  = 1.6\, d^* - d_1 - 0.2\, d_2
\end{align*}
\end{mybox}
We now combine these inequalities to get an upper bound for $\change_\cA(c)$.

When $\nicefrac 23 < \rho(f^*) \leq \nicefrac 34$, we have $\Pr[\cS_1] = \nicefrac 12,\Pr[\cS_2] = 0,\Pr[\cT_1] = \Pr[\cT_2] = \nicefrac 14$. Therefore,
\begin{align*}
\change_\cA(c) & \leq \nicefrac 12\cdot \WCchange_{\cS_1\cap\cA}(c) + \nicefrac 14 \cdot \WCchange_{\cT_1\cap\cA}(c) + \nicefrac 14 \cdot \WCchange_{\cT_2\cap\cA}(c) + \epsd \\
& \leq 1.65\, d^* - 0.7\, d_1 - 0.25\, d_2 + \epsd\\
& \leq \hoono{1.65\, d^* - 0.95\, d_1} + \epsd.\tag{$d_2\geq d_1$}
\end{align*}

When $\rho(f^*) > \nicefrac 34$, we have $\Pr[\cS_1] = \nicefrac 54 - \rho,\Pr[\cS_2] = \rho - \nicefrac 34,\Pr[\cT_1] = \Pr[\cT_2] = \nicefrac 14$. Therefore,
\begin{align*}
\change_\cA(c) & \leq (\nicefrac 54 - \rho) \cdot \WCchange_{\cS_1\cap\cA}(c) + (\rho - \nicefrac 34) \cdot \WCchange_{\cS_2\cap\cA}(c) + \nicefrac 14 \cdot \WCchange_{\cT_1\cap\cA}(c) + \nicefrac 14 \cdot \WCchange_{\cT_2\cap\cA}(c) + \epsd \\
& \leq 1.65\, d^* - (1.6\rho - 0.5)\, d_1 + (1.6\rho - 1.45)\, d_2 + \epsd.
\end{align*}
When $\rho \leq 1.45/1.6$, we use $d_2\geq d_1$:
\begin{align*}
\change_\cA(c) & \leq \hoono{1.65\, d^* - 0.95\, d_1} + \epsd.
\end{align*}
When $\rho > 1.45/1.6$, we use $d_2\leq d^* + \nicefrac 1\rho\cdot(d^* + d_1)$:
\begin{align*}
\change_\cA(c) & \leq (1.8 + 1.6\rho - 1.45/\rho)\, d^* - (1.6\rho + 1.45/\rho - 2.1)\, d_1 + \epsd\\
& \leq (1.8 + 1.6 - 1.45)\, d^* - (2\sqrt{1.6\times 1.45} - 2.1)\, d_1 + \epsd\\
& \leq \hoono{1.95\, d^* - 0.94630\, d_1} + \epsd.
\end{align*}
%
%
%
%
%
%
%
%
\subsubsection{\texorpdfstring{$\move{f^*}$ closes $f_1$ and $f_2$ on
    $\cT_b'\cap\cA$ for some $b \in \{1,2\}$}{[move(f*) closes f1 and f2 on event Tb' cap A]}}
Bounds for simple swaps remain the same as before:
\begin{align*}
\WCchange_{\cS_1\cap\cA}(c) & \leq \hoono{(1+\a\b)\, d^* - (1- \a\b)\, d_1 - 2\b\, d_2},\\
\WCchange_{\cS_2\cap\cA}(c) & \leq \hoono{(1+\a\b)\, d^* - 2\, d_1 + (1+\a\b-2\b)\, d_2}.
\end{align*}

For tree swaps, we partition $\cT\cap\cA$ as the union of $\cT_1\cap\cT_{3-b}'\cap\cA$, $\cT_2\cap\cT_{3-b}'\cap\cA$ and $\cT_b'\cap\cA$. On the first two events, our bounds are the same as in \Cref{sec:typeCsimple}:
\begin{align*}
\WCchange_{\cT_1\cap\cT_{3-b}'\cap\cA}(c) & \leq \hoono{(1+\a\b+\b)\, d^* -d_1},\\
\WCchange_{\cT_2\cap\cT_{3-b}'\cap\cA}(c) & \leq \hoono{(1+\a\b+\b)\, d^* - (2-2\b)\, d_1 + (1-2\b)\, d_2}.
\end{align*}
On $\cT_b'\cap\cA$, we have $\move{f^*}$ closes both $f_1$ and $f_2$. Therefore,
\begin{align*}
  \WCchange_{\cT_b'\cap\cA}(c) \leq \hoono{(1+\a\b)\, d^* - d_1 - \b\, d_2}. \tag{$\WCchange_{\move{f^*, \neg f_1,\neg f_2}}$}
\end{align*}

Summarizing, we have

\begin{mybox}
\begin{align*}
\dSwapHa\leq  & (1+\a\b)\, d^* - (1-\a\b)\, d_1 - 2\b\, d_2
&  = 1.6\, d^* - 0.4\, d_1 - 0.4\, d_2\\
\dSwapHb\leq  & (1+\a\b)\, d^* - 2\, d_1 + (1+\a\b-2\b)\, d_2
&  = 1.6\, d^* - 2\, d_1 + 1.2\, d_2\\
\WCchange_{\cT_1\cap\cT_{3-b}'\cap\cA}(c) \leq  & (1+\a\b+\b)\, d^* -d_1
&  = 1.8\, d^* - d_1 \\
\WCchange_{\cT_2\cap\cT_{3-b}'\cap\cA}(c) \leq  & (1+\a\b+\b)\, d^* - (2-2\b)\, d_1 + (1-2\b)\, d_2
&  = 1.8\, d^* - 1.6\, d_1 + 0.6\, d_2\\
\WCchange_{\cT_b'\cap\cA}(c) \leq & (1+\a\b)\, d^* - d_1 - \b\, d_2
&  = 1.6\, d^* - d_1 - 0.2\, d_2
\end{align*}
\end{mybox}

We now combine these inequalities to get an upper bound for $\change_\cA(c)$.

When $\nicefrac 23 < \rho(f^*) \leq \nicefrac 34$, we have $\Pr[\cS_1] = \nicefrac 12,\Pr[\cS_2] = 0,\Pr[\cT_1\cap\cT_{3-b}'] = \Pr[\cT_2\cap\cT_{3-b}'] = \nicefrac 18, \Pr[\cT_b'] = \nicefrac 14$. Therefore,
\begin{align*}
\change_\cA(c) & \leq \nicefrac 12\cdot \WCchange_{\cS_1\cap\cA}(c) + \nicefrac 18 \cdot \WCchange_{\cT_1\cap\cT_{3-b}'\cap\cA}(c) + \nicefrac 18 \cdot \WCchange_{\cT_2\cap\cT_{3-b}'\cap\cA}(c) + \nicefrac 14\cdot \WCchange_{\cT_b'\cap\cA}(c) + \epsd \\
& \leq 1.65\, d^* - 0.775\, d_1 -0.175\, d_2 + \epsd\\
& \leq \hoono{1.65\, d^* - 0.95\, d_1} + \epsd. \tag{$d_2\geq d_1$}
\end{align*}

When $\rho(f^*) > \nicefrac 34$, we have $\Pr[\cS_1] = \nicefrac 54 - \rho, \Pr[\cS_2] = \rho - \nicefrac 34, \Pr[\cT_1\cap\cT_{3-b}'] = \Pr[\cT_2\cap\cT_{3-b}'] = \nicefrac 18, \Pr[\cT_b'] = \nicefrac 14$. Therefore,
\begin{align*}
\change_\cA(c) & \leq (\nicefrac 54 - \rho)\cdot \WCchange_{\cS_1\cap\cA}(c) + (\rho - \nicefrac 34)\cdot \WCchange_{\cS_2\cap\cA}(c) + \nicefrac 18 \cdot \WCchange_{\cT_1\cap\cT_{3-b}'\cap\cA}(c) + \nicefrac 18 \cdot \WCchange_{\cT_2\cap\cT_{3-b}'\cap\cA}(c) \\ & ~~~ + \nicefrac 14\cdot \WCchange_{\cT_b'\cap\cA}(c) + \epsd \\
& \leq 1.65\, d^* - (1.6\rho - 0.425)\, d_1 + (1.6\rho - 1.375)\, d_2+ \epsd.
\end{align*}
When $\rho \leq 1.375/1.6$, we use $d_2\geq d_1$:
\begin{align*}
\change_\cA(c) & \leq \hoono{1.65\, d^* - 0.95\, d_1} + \epsd.
\end{align*}
When $\rho > 1.375/1.6$, we use $d_2\leq d^* + \nicefrac 1\rho (d^* + d_1)$:
\begin{align*}
\change_\cA(c) & \leq (1.875 + 1.6\rho - 1.375/\rho)\, d^* - (1.6\rho + 1.375/\rho - 2.025)\, d_1 + \epsd\\
& \leq (1.875 + 1.6 - 1.375)\, d^* - (2\sqrt{1.6 \times 1.375} - 2.025)\, d_1 + \epsd\\
& \leq \hoono{2.1\, d^* - 0.94147\, d_1} + \epsd.
\end{align*}
%
%
%
%
%
%
%
%
\subsection{Proof of \eqref{eqn:bndD}:  Clients of Type \xD}

In this section, we show that for any client $c$ with type \xD, we have
\[
\Gsum{c} \leq 2.5203\, d^*(c) - 0.8888\, d_1(c).
\]

Similar to \Cref{claim:type3Swapcases} for type \xC clients, we also have the following claim classifying type \xD clients into subtypes. The only change is in item (e), where we replace $\cT_2$ by $\cT_1$ because the roles of $\eta_1$ and $\eta_2$ are now swapped.

\begin{restatable}[Type \xD subcases]{claim}{typeFourSwapcases}
  \label{claim:type4Swapcases}
  For a client $c$ of type \xD, one of the following is true:
  \begin{enumerate}[(a)]
    \item  $f_1$ is \heavy.
    \item  $f_2$ is \heavy.
    \item A facility $h$ is open near $c$ after the simple swap closing $f_1$. Formally, a facility $h\neq f_2$ is open after swap $\move{\neg f_1}$ at distance $d(c,h)\leq 3d_1 + 2d^*$ on $\cS\cap\cA$.
	\item $g^* \neq f^*$, $\rho(g^*) > \nicefrac 34$, and for all $b = 1,2$, any swap set $\cP$ generated on $\cS_b'\cap\cA$, a facility $h\neq f_2$ is open after swap $\move{\neg f_1}$ at distance $d(c,h)\leq \left\{\begin{array}{ll}2d_1 + d^*,& \textup{if}~b = 1\\ 2d_1 + d^* + \nicefrac 43(d_1 + d^*),& \textup{if}~b=2\end{array}\right.$.
    \item For any swap set $\cP$ generated on $\cT_1\cap\cA$, $\move{f^*}$ closes both $f_1$ and $f_2$;
\item $g^*\neq f^*$, $\rho(g^*)>\nicefrac 23$, and there exists $b\in\{1,2\}$ such that for any swap set $\cP$ generated on $\cT_{b}'\cap\cA$, $\move{f^*}$ closes both $f_1$ and $f_2$.
  \end{enumerate}
\end{restatable}
\subsubsection{When \texorpdfstring{$f_1$}{[f1]} is a heavy facility}
$f_1$ being heavy implies that $\move{\neg f_1}$ doesn't exist. We thus focus on $\move{f^*}$ and $\move{\neg f_2}$.

\paragraph{Simple \& tree swaps with $\tau(f^*) = \eta_1$} We have $\move{f^*} = \move{\neg f_2}$ by implication (ii) of amenability. Therefore,
\begin{align*}
\dSwapHa  \leq {} & \hoono{(1+\a\b)d^* - d_1 - \b\, d_2}, \tag{$\WCchange_{\move{f^*,\neg f_2}}$}\\
\dSwapTa  \leq {} & \hoono{(1+\a\b)d^* - d_1 - \b\, d_2}. \tag{$\WCchange_{\move{f^*,\neg f_2}}$}
\end{align*}

\paragraph{Simple swaps with $\tau(f^*) = \eta_2$} By implication (Siv) of amenability, we have $\move{f^*}\neq \move{\neg f_2}$. On swap $\move{f^*}$ the client can be served by $f^*$. On swap $\move{\neg f_2}$, the client can be served by $f_1$. Therefore,
\begin{align*}
\dSwapHb \leq {} & (1+\a\b)d^* - d_1 -  \b\, d_2 \tag{$\WCchange_{\move{f^*}}$}\\
& + (1+\a\b)d_1 - d_1 -\b\, d_2 \tag{$\WCchange_{\move{\neg f_2}}$}\\
= {} & \hoono{(1+\a\b)d^* - (1-\a\b)\, d_1 -2\b\, d_2}
\end{align*}

\paragraph{Tree swaps with $\tau(f^*) = \eta_2$}Let us first assume that $\move{f^*}\neq \move{\neg f_2}$. On swap $\move{f^*}$, the client can be served by $f^*$. On swap $\move{\neg f_2}$, the client can be served by $f_1$ and $\pi(f_2)$, by implication (Tii) of amenability. 
We have $d(c,\pi(f_2))\leq 
\leq 2d_2 + d^*$ by \eqref{eqn:pif2}.  Therefore,
\begin{align*}
\dSwapTb \leq {} & (1+\a\b)d^* - d_1 - \b\, d_2 \tag{$\WCchange_{\move{f^*}}$}\\
& + d_1+ \b(2d_2+d^*)- d_1 -\b\, d_2 \tag{$\WCchange_{\move{\neg f_2}}$}\\
= {} & \hoono{(1+\a\b+\b)d^* - d_1 }
\end{align*}
This inequality also holds when $\move{\neg f_2} = \move{f^*}$, because our bound for $\WCchange_{\move{f^*}}$ does not require $f_2$ to remain open after the swap and $\WCchange_{\move{\neg f_2}}$ is non-negative.

\begin{mybox}
\begin{align*}
\dSwapHa\leq  & (1+\a\b)\, d^* -d_1 -  \b\, d_2
&  = 1.6\, d^* -  d_1 - 0.2\, d_2\\
\dSwapHb\leq  & (1+\a\b)\, d^* - (1-\a\b)\, d_1 -2\b\, d_2
&  = 1.6\, d^* -  0.4\, d_1 - 0.4\, d_2\\
\dSwapTa\leq  & (1+\a\b)\, d^* -  d_1 - \b\, d_2
&  = 1.6\, d^* - d_1 -0.2\, d_2\\
\dSwapTb\leq  & (1+\a\b+\b)\, d^*-  d_1
&  = 1.8\, d^* -  d_1
\end{align*}
\end{mybox}

We now combine these inequalities to get an upper bound for $\change_\cA(c)$.

When $\rho(f^*)\leq \nicefrac 23$, we have $\Pr[\cS_1] = \Pr[\cT_1] = \nicefrac 12, \Pr[\cS_2] = \Pr[\cT_2] = 0$. Therefore,
\begin{align*}
\change_\cA(c) & \leq \nicefrac 12\cdot \WCchange_{\cS_1\cap\cA}(c) + \nicefrac 12\cdot \WCchange_{\cT_1\cap\cA}(c) + \epsd\\
& \leq\hoono{ 1.6\, d^* -1.2\, d_1} + \epsd. \tag{$d_1 \leq d_2$}
\end{align*}

When $\nicefrac 23 < \rho(f^*) \leq \nicefrac 34$, we have $\Pr[\cS_1] = \nicefrac 12, \Pr[\cS_2] = 0, \Pr[\cT_1] = \Pr[\cT_2] = \nicefrac 14$. Therefore,
\begin{align*}
\change_\cA(c) & \leq \nicefrac 12\cdot \WCchange_{\cS_1\cap\cA}(c) + \nicefrac 14\cdot \WCchange_{\cT_1\cap\cA}(c) + \nicefrac 14\cdot \WCchange_{\cT_2\cap\cA}(c) + \epsd\\
&\leq 1.65\, d^* - d_1 -0.15\, d_2 +\epsd\\
& \leq \hoono{1.65\, d^* - 1.15\, d_1} +\epsd \tag{$d_1 \leq d_2$}.
\end{align*}

When $\rho(f^*) > \nicefrac 34$, we have $\Pr[\cS_1] = \nicefrac 54 - \rho, \Pr[\cS_2] = \rho - \nicefrac 34, \Pr[\cT_1] = \Pr[\cT_2] = \nicefrac 14$. Therefore,
\begin{align*}
\change_\cA(c) & \leq (\nicefrac 54 - \rho)\cdot \WCchange_{\cS_1\cap\cA}(c)+ (\rho - \nicefrac 34)\cdot \WCchange_{\cS_2\cap\cA}(c) + \nicefrac 14\cdot \WCchange_{\cT_1\cap\cA}(c) + \nicefrac 14\cdot \WCchange_{\cT_2\cap\cA}(c) + \epsd\\
& \leq 1.65\, d^* - (1.45-0.6\rho)\, d_1 - (0.2\rho)\, d_2 + \epsd\\
& \leq 1.65\, d^* - (1.45-0.4\rho)\, d_1 + \epsd \tag{$d_1 \leq d_2$}\\
& \leq \hoono{1.65\, d^* - 1.05\, d_1} +\epsd
\end{align*}

\subsubsection{When \texorpdfstring{$f_2$}{[f2]} is a heavy facility}
$f_2$ being heavy implies that the swap $\move{\neg f_2}$ doesn't exist. We thus focus on $\move{f^*}$ and $\move{\neg f_1}$.

\paragraph{Simple swaps with $\tau(f^*) = \eta_1$} We have $\move{f^*} \neq \move{\neg f_1}$ by implication (Siv) of amenability. On swap $\move{f^*}$, the client can be served by $f^*$ and $f_1$. On swap $\move{\neg f_1}$, the client can be served by $f_2$. Therefore,
\begin{align*}
\dSwapHa  \leq {} & d^* + \b d_1 - d_1 - \b\, d_2 \tag{$\WCchange_{\move{f^*}}$}\\
& +  (1+\a\b)d_2 -  d_1 -\b\, d_2 \tag{$\WCchange_{\move{\neg f_1}}$}\\
= {} & \hoono{ d^*-(2-\b)\, d_1 + (1+\a\b-2\b) d_2}.
\end{align*}

We can also use $(1-\b)d^* + 2\b\, d_1$ to upper-bound $\WCchange_{\move{f^*}}$ (by \eqref{eqn:bestBalance}) and get
\begin{align*}
\dSwapHa  \leq {} & \hoono{ (1-\b)\, d^*-(2-2\b)\, d_1 + (1+\a\b-2\b) d_2}.
\end{align*}

\paragraph{Tree swaps with $\tau(f^*) = \eta_1$} Let us first assume that $\move{f^*}\neq \move{\neg f_1}$. On swap $\move{f^*}$, the client can be served by $f^*$ and $f_1$. On swap $\move{\neg f_1}$, the client can be served by $f_2$ and $\pi(f_1)$ by implication (Tii) of amenability. We have $d(c,\pi(f_1))\leq 2d_1 + d^*$ by \eqref{eqn:ubpartner}. Therefore,
\begin{align*}
\dSwapTa \leq {} & d^* +\b\,  d_1- d_1 - \b\, d_2 \tag{$\WCchange_{\move{f^*}}$}\\
& + d_2+ \b(2d_1+d^*)- d_1 -\b\, d_2 \tag{$\WCchange_{\move{\neg f_1}}$}\\
= {} & \hoono{(1+\b)d^* -  (2-3\b)d_1 + (1-2\b)\, d_2}.
\end{align*}

If $\move{\neg f_1} = \move{f^*}$, we still have the same bound:
\begin{align*}
\dSwapTa \leq {} & d^* +\b\,  d_2- d_1 - \b\, d_2 \tag{$\WCchange_{\move{f^*, \neg f_1}}$}\\
& +  (1-2\b)(d_2-d_1) + \b\, d^* + \b\, d_1 \tag{non-negative terms}\\
= {} & (1+\b)d^* -  (2-3\b)d_1 + (1-2\b)\, d_2.
\end{align*}

\paragraph{Simple \& tree swaps with $\tau(f^*) = \eta_2$} We have $\move{f^*} = \move{\neg f_1}$ by implication (ii) of amenability. Therefore,
\begin{align*}
\dSwapHb \leq {} & \hoono{(1+\a\b)d^*  - d_1 -  \b\, d_2}, \tag{$\WCchange_{\move{f^*, \neg f_1}}$}\\
\dSwapTb \leq {} & d^* + \b d_2- d_1 -  \b\, d_2 \tag{$\WCchange_{\move{f^*, \neg  f_1}}$}\\
= {} & \hoono{d^* - d_1}.
\end{align*}

Summarizing, we have
\begin{mybox}
\begin{align*}
\dSwapHa\leq  & d^* - (2-\b)\, d_1 + (1+\a\b-2\b)\, d_2
&  = d^* - 1.8\, d_1 + 1.2\, d_2\\
\dSwapHa\leq  & (1-\b)\, d^*-(2-2\b)\, d_1 + (1+\a\b-2\b) d_2
& = 0.8\, d^* - 1.6\, d_1 + 1.2\, d_2\\\\
\dSwapHb\leq  & (1+\a\b)d^* - d_1 -\b d_2
&  = 1.6\, d^*-d_1 -0.2\, d_2\\
\dSwapTa\leq  & (1+\b)d^* - (2-3\b)d_1 + (1-2\b)d_2
&  = 1.2\, d^* - 1.4\, d_1+0.6\, d_2\\
\dSwapTb\leq  & d^*-d_1
&  = d^*-d_1
\end{align*}
\end{mybox}

We now combine these inequalities to get an upper bound for $\change_\cA(c)$.

When $\rho(f^*)\leq \nicefrac 23$, we have $\Pr[\cS_1] = \Pr[\cT_1] = \nicefrac 12, \Pr[\cS_2] = \Pr[\cT_2] = 0$. In this case we use the second inequality for $\dSwapHa$. Therefore,
\begin{align*}
\change_\cA(c) & \leq \nicefrac 12\cdot \WCchange_{\cS_1\cap\cA}(c) + \nicefrac 12\cdot \WCchange_{\cT_1\cap\cA}(c) + \epsd\\
& \leq d^* - 1.5\, d_1 + 0.9\, d_2+ \epsd\\
& \leq (1.9+0.9\rho)\, d^* - (1.5-0.9\rho)\, d_1 + \epsd \tag{$d_2 \leq d^* + \rho(d^*+d_1)$}\\
& \leq (1.9+0.9\cdot\nicefrac 23)d^* - (1.5-0.9\cdot\nicefrac 23)\, d_1 + \epsd\\
& \leq \hoono{2.5\, d^* - 0.9\, d_1} + \epsd.
\end{align*}

When $\nicefrac 23 < \rho(f^*) \leq \nicefrac 34$, we have $\Pr[\cS_1] = \nicefrac 12, \Pr[\cS_2] = 0, \Pr[\cT_1] = \Pr[\cT_2] = \nicefrac 14$. In this case we use the first inequality for $\dSwapHa$.  Therefore,
\begin{align*}
\change_\cA(c) & \leq \nicefrac 12\cdot \WCchange_{\cS_1\cap\cA}(c) + \nicefrac 14\cdot \WCchange_{\cT_1\cap\cA}(c) + \nicefrac 14\cdot \WCchange_{\cT_2\cap\cA}(c) + \epsd\\
& \leq 1.05\, d^* -1.5\, d_1 +0.75\, d_2 +\epsd\\
& \leq (1.8+0.75\rho)\, d^* - (1.5-0.75\rho)\, d_1 +\epsd\tag{$d_2 \leq d^* + \rho(d^*+d_1)$} \\
& \leq (1.8+0.75\cdot\nicefrac{3}{4}) - (1.5-0.75\cdot\nicefrac{3}{4})\, d_1 +\epsd\\
& \leq \hoono{2.3625\, d^* - 0.9375\, d_1} +\epsd.
\end{align*}

When $\rho(f^*) > \nicefrac 34$, we have $\Pr[\cS_1] = \nicefrac 54 - \rho, \Pr[\cS_2] = \rho - \nicefrac 34, \Pr[\cT_1] = \Pr[\cT_2] = \nicefrac 14$. 
 In this case we use the first inequality for $\dSwapHa$. Therefore,
\begin{align*}
\change_\cA(c) & \leq (\nicefrac 54 - \rho)\cdot \WCchange_{\cS_1\cap\cA}(c)+ (\rho - \nicefrac 34)\cdot \WCchange_{\cS_2\cap\cA}(c) + \nicefrac 14\cdot \WCchange_{\cT_1\cap\cA}(c) + \nicefrac 14\cdot \WCchange_{\cT_2\cap\cA}(c) + \epsd\\
& \leq (0.6+0.6\rho)\, d^* - (2.1-0.8\rho)\, d_1 + (1.8-1.4\rho)\, d_2 +\epsd\\
& \leq (2.4+\rho-1.4\rho^2)\, d^* - (2.1-2.6\rho+1.4\rho^2)\, d_1 + \epsd \tag{$d_2\leq d^*+\rho(d^*+d_1)$}\\ 
& \leq (2.4+\nicefrac{3}{4}-1.4(\nicefrac{3}{4})^2)\, d^* - (2.1-2.6\cdot\nicefrac{13}{14}+1.4(\nicefrac{13}{14})^2)\, d_1 + \epsd\\
& \leq \hoono{2.3625\, d^* - 0.8928\, d_1} +\epsd.
\end{align*}


\subsubsection{\texorpdfstring{There exists a facility $h$ such that $d(c,h)\leq 3d_1 + 2d^*$}{[d(c,h) <= 3d1+2d*]} in simple swaps}
\label{sec:typeDsimple}

\paragraph{Simple swaps with $\tau(f^*) = \eta_1$} Implications (ii) and (Siv) of amenability imply $\move{f^*} = \move{\neg f_2} \neq \move{\neg f_1}$. 
On swap $\move{f^*,\neg f_2}$, the client can be served by $f^*$ and $f_1$. On swap $\move{\neg f_1}$, the client can be served by $f_2$ and $h$. Therefore,
\begin{align*}
\dSwapHa \leq {} &  d^*+ \b\, d_1 - d_1 -\b\, d_2 \tag{$\WCchange_{\move{f^*, \neg f_2}}$}\\
& + d_2 + \b(3d_1+2d^*) - d_1 - \b\, d_2 \tag{$\WCchange_{\move{\neg f_1}}$}\\
= {} & \hoono{(1+2\b)\, d^* -(2-4\b)\, d_1 + (1-2\b)\, d_2}.
\end{align*}

In $\WCchange_{\move{f^*, \neg f_2}}$, we can use $0.776(1+\a\b)d^* + 0.224(d^*+\b\, d_1)$ instead of $d^*+ \b\, d_1$. This gives
\begin{align*}
\dSwapHa \leq {} & \hoono{(1+2\b+0.776\a\b)\, d^* - (2-3.224\b)\, d_1 + (1-2\b)\, d_2}.
\end{align*}


\paragraph{Simple swaps with $\tau(f^*) = \eta_2$} Implications (ii) and (Siv) of amenability imply $\move{f^*} = \move{\neg f_1} \neq \move{\neg f_2}$. On swap $\move{f^*,\neg f_1}$, the client can be served by $f^*$. On swap $\move{\neg f_2}$, the client can be served by $f_1$. Therefore,
\begin{align*}
  \dSwapHb \leq {} & (1+\a\b)\, d^* - d_1 - \b\, d_2 \tag{$\WCchange_{\move{f^*, \neg f_1}}$}\\
  & + (1+\a\b)\, d_1  - d_1 - \b\, d_2 \tag{$\WCchange_{\move{\neg f_2}}$}\\
  = {} & \hoono{(1+\a\b)\, d^* - (1-\a\b)\, d_1 - 2\b\, d_2}.
\end{align*}

\paragraph{Tree swaps with $\tau(f^*) = \eta_1$} We have $\move{f^*} = \move{\neg f_2}$ by implication (ii) of amenability. Let us first assume that $\move{\neg f_1}\neq \move{f^*,\neg f_2}$. On swap $\move{f^*,\neg f_2}$, the client can be served by $f^*$. On swap $\move{\neg f_1}$, the client can be served by $f_2$ and $\pi(f_1)$ by implication (Tii) of amenability. 
Note that $d(c,\pi(f_1))\leq 2d_1 + d^*$ by \eqref{eqn:ubpartner}. Therefore,
\begin{align*}
  \dSwapTa \leq {} & (1+\a\b)\, d^* - d_1 -\b\, d_2 \tag{$\WCchange_{\move{f^*, \neg f_2}}$}\\
  & + d_2 + \b(2d_1+d^*) - d_1 - \b\, d_2 \tag{$\WCchange_{\move{\neg f_1}}$}\\
  = {} & \hoono{(1+\a\b+\b)\, d^* -(2-2\b)\, d_1 + (1-2\b)\, d_2}.
\end{align*}
This inequality also holds when $\move{\neg f_1} = \move{f^*,\neg f_2}$, because our bound for $\WCchange_{\move{f^*,\neg f_2}}$ does not require $f_1$ to remain open after the swap
and $\WCchange_{\move{\neg f_1}}$ is non-negative.

In $\WCchange_{\move{f^*, \neg f_2}}$, we can use $d^* + \b\, d_1$ instead of $(1+\a\b)\, d^*$. This gives
\begin{align*}
  \dSwapTa \leq {} &  \hoono{(1+\b)\,d^* - (2-3\b)\, d_1 + (1-2\b)\, d_2}.
\end{align*}

This bound also holds when $\move{\neg f_1} = \move{f^*, \neg f_2}$ because in this case we have
\begin{align*}
  \dSwapTa \leq {} & d^* + \b\, d_1 - d_1 -\b\, d_2 \tag{$\WCchange_{\move{f^*, \neg f_2, \neg f_1}}$}\\
  & + (1-\b)(d_2 - d_1) + \b d_1 +\b d^*\tag{non-negative terms}\\
  = {} & (1+\b)\, d^* -(2-3\b)\, d_1 +(1-2\b)\, d_2.
\end{align*}
%
%

\paragraph{Tree swaps with $\tau(f^*) = \eta_2$} We have $\move{f^*} = \move{\neg f_1}$ by implication (ii) of amenability. Let us first assume that $\move{\neg f_2}\neq \move{f^*,\neg f_1}$. On swap $\move{f^*,\neg f_1}$, the client can be served by $f^*$. On swap $\move{\neg f_2}$, the client can be served by $f_1$ and $\pi(f_2)$ by implication (Tii) of amenability. Note that $d(c,\pi(f_2))\leq d_2 + d(f_2,\pi(f_2))\leq d_2 + d(f_2,f^*) \leq d_2 + \rho d(f_1,f^*) \leq d_2 + \rho(d_1 + d^*)$. Therefore,
\begin{align*}
  \dSwapTb \leq {} & (1+\a\b)\, d^* - d_1 - \b\, d_2 \tag{$\WCchange_{\move{f^*, \neg f_1}}$}\\
  & + d_1 + \b(d_2 + \rho(d^* + d_1)) - d_1 - \b\, d_2 \tag{$\WCchange_{\move{\neg f_2}}$}\\
  = {} & \hoono{(1+\a\b+\rho\b)\, d^*-(1-\rho\b)\, d_1 - \b\, d_2}.
\end{align*}
This inequality also holds when $\move{\neg f_2} = \move{f^*,\neg f_1}$, because our bound for $\WCchange_{\move{f^*,\neg f_1}}$ does not require $f_2$ to remain open after the swap and $\WCchange_{\move{\neg f_2}}$ is non-negative.

Summarizing, we have

\begin{mybox}
\begin{align*}
\dSwapHa\leq  & (1+2\b)\, d^* - (2-4\b)\, d_1 + (1-2\b)\, d_2
&  = 1.4\, d^* - 1.2\, d_1 + 0.6\, d_2\\
\dSwapHa\leq  & (1+2\b+0.776\a\b)\, d^* - (2-3.224\b)\, d_1 + (1-2\b)\, d_2
&  = 1.8656\, d^* - 1.3552\, d_1 + 0.6\, d_2\\\\
\dSwapHb\leq  & (1+\a\b)\, d^* - (1-\a\b)\, d_1 - 2\b\, d_2
&  = 1.6\, d^* - 0.4\, d_1 - 0.4\, d_2\\\\
\dSwapTa\leq  & (1+\a\b+\b)\, d^* -(2-2\b)\, d_1 + (1-2\b)\, d_2
&  = 1.8\, d^* - 1.6\, d_1 + 0.6\, d_2\\
\dSwapTa\leq  & (1+\b)\, d^* -(2-3\b)\, d_1 + (1-2\b)\, d_2
&  = 1.2\, d^* - 1.4\, d_1 + 0.6\, d_2\\\\
\dSwapTb\leq  & (1+\a\b+\rho\b)\, d^* - (1-\rho\b)\, d_1 - \b\, d_2\quad\quad\quad\quad
  = (1.6+0.2\rho)\, d^* - (1-0.2\rho)\, d_1 - 0.2\, d_2\span
\end{align*}
\end{mybox}

We now combine these inequalities to get an upper bound for $\change_\cA(c)$.

When $\rho(f^*)\leq \nicefrac 23$, we have $\Pr[\cS_1] = \Pr[\cT_1] = \nicefrac 12, \Pr[\cS_2] = \Pr[\cT_2] = 0$. We use the first bound for $\dSwapHa$ and the second bound for $\dSwapTa$. Therefore,

\begin{align*}
\change_\cA(c) & \leq \nicefrac 12\cdot \WCchange_{\cS_1\cap\cA}(c) + \nicefrac 12\cdot \WCchange_{\cT_1\cap\cA}(c) + \epsd\\
& \leq 1.3\, d^* - 1.3\, d_1 + 0.6\, d_2 + \epsd\\
& \leq (1.9 + 0.6\rho)\, d^* - (1.3-0.6\rho)\, d_1 + \epsd \tag{$d_2\leq d^* + \rho(d^* + d_1)$}\\
& \leq (1.9 + 0.6\cdot \nicefrac 23)\, d^* - (1.3-0.6\cdot \nicefrac 23)\, d_1 + \epsd \\
& \leq \hoono{2.3\, d^* - 0.9\, d_1} + \epsd
\end{align*}

When $\nicefrac 23 < \rho(f^*) \leq \nicefrac 34$, we have $\Pr[\cS_1] = \nicefrac 12, \Pr[\cS_2] = 0, \Pr[\cT_1] = \Pr[\cT_2] = \nicefrac 14$.  We use the first bound for both $\dSwapHa$ and $\dSwapTa$. Therefore,
\begin{align*}
\change_\cA(c) & \leq \nicefrac 12\cdot \WCchange_{\cS_1\cap\cA}(c) + \nicefrac 14\cdot \WCchange_{\cT_1\cap\cA}(c) + \nicefrac 14\cdot \WCchange_{\cT_2\cap\cA}(c) + \epsd\\
& \leq (1.55+0.05\rho)\, d^* - (1.25-0.05\rho)\, d_1 + 0.4\, d_2 + \epsd\\
& \leq (1.95+0.45\rho)\, d^* - (1.25-0.45\rho)\, d_1 + \epsd\tag{$d_2\leq d^* + \rho(d^* + d_1)$}\\
& \leq (1.95+0.45\cdot \nicefrac 34)\, d^* - (1.35-0.45\cdot \nicefrac 34)\, d_1 + \epsd\\
& = \hoono{2.2875\, d^* -  0.9125\, d_1} + \epsd.
\end{align*}
 
When $\rho(f^*) > \nicefrac 34$, we have $\Pr[\cS_1] = \nicefrac 54 - \rho, \Pr[\cS_2] = \rho - \nicefrac 34, \Pr[\cT_1] = \Pr[\cT_2] = \nicefrac 14$.  We use the second bound for  $\dSwapHa$ and the first bound for  $\dSwapTa$. Therefore,
\begin{align*}
\change_\cA(c) & \leq (\nicefrac 54 - \rho)\cdot \WCchange_{\cS_1\cap\cA}(c)+ (\rho - \nicefrac 34)\cdot \WCchange_{\cS_2\cap\cA}(c) + \nicefrac 14\cdot \WCchange_{\cT_1\cap\cA}(c) + \nicefrac 14\cdot \WCchange_{\cT_2\cap\cA}(c) + \epsd\\
& \leq   (1.982 - 0.2156\rho)\, d^* -  (2.044-1.0052\rho)\, d_1+ (1.15 -\rho)\, d_2 +\epsd\\
& \leq   (3.132 - 0.0656\rho-\rho^2)\, d^* - (2.044-2.1552\rho+\rho^2)\, d_1 + \epsd\tag{$d_2\leq d^* + \rho(d^* + d_1)$}\\
& \leq   (3.132 - 0.0656\cdot\nicefrac{3}{4}-\nicefrac{3}{4}^2)\, d^* - (2.044-2.1552+1^2)\, d_1 + \epsd\\
& \leq \hoono{2.5203\, d^* - 0.8888\, d_1} + \epsd.
\end{align*}
%
%

%
%
%
\subsubsection{\texorpdfstring{$d(c,h)\leq 2d_1 + d^*$ or $d(c,h)\leq
    2d_1 + d^* + \nicefrac 43(d^* + d_1)$ in simple swaps}{d(c,h) <=
    d1 + d* or d(c,h) <= 2d1 + d* + 4/3(d* + d1)
    in simple swaps}}

Similarly to \Cref{sec:typeCsimple2}, our bound for $\change_\cA(c)$ in the previous case remains valid in this case.

\subsubsection{\texorpdfstring{$\move{f^*}$ closes $f_1$ and $f_2$ on
    $\cT_1\cap\cA$}{[move(f*) closes f1 and f2 on
    event T1 cap A]}}

If $\tau(f^*) = \eta_2$, we get the same bounds as before:
\begin{align*}
\WCchange_{\cS_2\cap\cA}(c) & \leq \hoono{(1+\a\b)\, d^* - (1-\a\b)\, d_1 - 2\b\, d_2},\\
\WCchange_{\cT_2\cap\cA}(c) & \leq \hoono{(1+\a\b+\rho\b)\, d^*-(1-\rho\b)\, d_1 - \b\, d_2}.
\end{align*}
We proceed to bound $\WCchange_{\cS_1\cap\cA}(c)$ and $\WCchange_{\cT_1\cap\cA}(c)$.

\paragraph{Simple swaps with $\tau(f^*) = \eta_1$} Implications (ii) and (Siv) of amenability implies $\move{f^*} = \move{\neg f_2} \neq \move{\neg f_1}$. On swap $\move{f^*,\neg f_2}$, the client can be served by $f^*$ and $f_1$. On swap $\move{\neg f_1}$, the client can be served by $f_2$. Therefore,
\begin{align*}
  \dSwapHa \leq {} & d^* + \b\, d_1 - d_1 -\b\, d_2 \tag{$\WCchange_{\move{f^*, \neg f_2}}$}\\
  & + (1+\a\b)\, d_2 - d_1 - \b\, d_2 \tag{$\WCchange_{\move{\neg f_1}}$}\\
  = {} & \hoono{d^* -(2-\b)\, d_1 + (1+\a\b-2\b)\, d_2}.
\end{align*}

\paragraph{Tree swaps with $\tau(f^*) = \eta_1$} On $\cT_1\cap\cA$, we assumed that $\move{f^*}$ closes $f_1$ and $f_2$. Therefore,
\begin{equation*}
  \dSwapTa \leq \hoono{(1+\a\b)\, d^* - d_1 - \b\, d_2}. \tag{$\WCchange_{\move{f^*,\neg f_1,\neg f_2}}$}
\end{equation*}

Summarizing, we have

\begin{mybox}
\begin{align*}
\dSwapHa\leq  & d^* - (2-\b)\, d_1 + (1+\a\b-2\b)\, d_2
& \quad = d^* - 1.8\, d_1 + 1.2\, d_2\\
\dSwapHb\leq  & (1+\a\b)\, d^* - (1-\a\b)\, d_1 - 2\b\, d_2
& \quad = 1.6\, d^* - 0.4\, d_1 - 0.4\, d_2\\
\dSwapTa\leq  & (1+\a\b)\, d^*-d_1 -\b\, d_2
& \quad = 1.6\, d^* - d_1 - 0.2\, d_2\\
\dSwapTb\leq  & (1+\a\b+\rho\b)\, d^* - (1-\rho\b)\, d_1 - \b\, d_2
& \quad = (1.6+0.2\rho)\, d^* - (1-0.2\rho)\, d_1 - 0.2\, d_2
\end{align*}
\end{mybox}
We now combine these inequalities to get an upper bound for $\change_\cA(c)$.

When $\rho(f^*)\leq \nicefrac 23$, we have $\Pr[\cS_1] = \Pr[\cT_1] = \nicefrac 12, \Pr[\cS_2] = \Pr[\cT_2] = 0$. Therefore,
\begin{align*}
\change_\cA(c) & \leq \nicefrac 12\cdot \WCchange_{\cS_1\cap\cA}(c) + \nicefrac 12\cdot \WCchange_{\cT_1\cap\cA}(c) + \epsd\\
& \leq 1.3\, d^* - 1.4\, d_1 + 0.5\, d_2 + \epsd\\
& \leq (1.8 + 0.5\rho)\, d^* - (1.4 - 0.5\rho)\, d_1 + \epsd\tag{$d_2\leq d^* + \rho(d^* + d_1)$}\\
& \leq (1.8 + 0.5 \times\nicefrac 23)\, d^* - (1.4 - 0.5 \times \nicefrac 23)\, d_1 + \epsd\\
& \leq \hoono{2.13334\, d^* - 1.06666\, d_1} + \epsd.
\end{align*}

When $\nicefrac 23 < \rho(f^*) \leq \nicefrac 34$, we have $\Pr[\cS_1] = \nicefrac 12, \Pr[\cS_2] = 0, \Pr[\cT_1] = \Pr[\cT_2] = \nicefrac 14$. Therefore,
\begin{align*}
\change_\cA(c) & \leq \nicefrac 12\cdot \WCchange_{\cS_1\cap\cA}(c) + \nicefrac 14\cdot \WCchange_{\cT_1\cap\cA}(c) + \nicefrac 14\cdot \WCchange_{\cT_2\cap\cA}(c) + \epsd\\
& \leq (1.3 + 0.05\rho)\, d^* - (1.4 - 0.05\rho)\, d_1 + 0.5\, d_2 + \epsd\\
& \leq (1.8 + 0.55\rho)\, d^* - (1.4 - 0.55\rho)\, d_1 + \epsd \tag{$d_2\leq d^* + \rho(d^* + d_1)$}\\
& \leq (1.8 + 0.55 \times\nicefrac 34)\, d^* - (1.4 - 0.55\times\nicefrac 34)\, d_1 + \epsd\\
& = \hoono{2.2125\, d^* - 0.9875\, d_1} + \epsd.
\end{align*}

When $\rho(f^*) > \nicefrac 34$, we have $\Pr[\cS_1] = \nicefrac 54 - \rho, \Pr[\cS_2] = \rho - \nicefrac 34, \Pr[\cT_1] = \Pr[\cT_2] = \nicefrac 14$. Therefore,
\begin{align*}
\change_\cA(c) & \leq (\nicefrac 54 - \rho)\cdot \WCchange_{\cS_1\cap\cA}(c)+ (\rho - \nicefrac 34)\cdot \WCchange_{\cS_2\cap\cA}(c) + \nicefrac 14\cdot \WCchange_{\cT_1\cap\cA}(c) + \nicefrac 14\cdot \WCchange_{\cT_2\cap\cA}(c) + \epsd\\
& \leq (0.85 + 0.65\rho)\, d^* - (2.45 - 1.45\rho)\, d_1 + (1.7 - 1.6\rho)\, d_2 + \epsd\\
& \leq (2.55 + 0.75\rho - 1.6\rho^2)\, d^* - (2.45 - 3.15\rho + 1.6\rho^2)\, d_1 + \epsd\tag{$d_2\leq d^* + \rho(d^* + d_1)$}\\
& \leq (2.55 + 0.75 \times \nicefrac 34 - 1.6 \times (\nicefrac 34)^2)\, d^* - (2.45 - 3.15\times\nicefrac {3.15}{3.2} + 1.6 \times (\nicefrac{3.15}{3.2})^2)\, d_1\\ & ~~~ + \epsd\\
& \leq \hoono{2.2125\, d^* - 0.89960\, d_1} + \epsd.
\end{align*}
%
%
%
%
%
%
%
\subsubsection{\texorpdfstring{$\move{f^*}$ closes $f_1$ and $f_2$ on
    $\cT_b'\cap\cA$}{[move(f*) closes f1 and f2 on event (Tb' cap A)]}} 
Bounds for simple swaps remain the same as before:
\begin{align*}
\WCchange_{\cS_1\cap\cA}(c) & \leq \hoono{d^* -(2-\b)\, d_1 + (1+\a\b-2\b)\, d_2},\\
\WCchange_{\cS_2\cap\cA}(c) & \leq \hoono{(1+\a\b)\, d^* - (1-\a\b)\, d_1 - 2\b\, d_2}.
\end{align*}

For tree swaps, we partition $\cT\cap\cA$ as the union of $\cT_1\cap\cT_{3-b}'\cap\cA$, $\cT_2\cap\cT_{3-b}'\cap\cA$ and $\cT_b'\cap\cA$. On the first two events, our bounds are the same as in \Cref{sec:typeDsimple}:
\begin{align*}
\WCchange_{\cT_1\cap\cT_{3-b}'\cap\cA}(c) & \leq \hoono{(1+\a\b+\b)\, d^* -(2-2\b)\, d_1 + (1-2\b)\, d_2},\\
\WCchange_{\cT_2\cap\cT_{3-b}'\cap\cA}(c) & \leq \hoono{(1+\a\b+\rho\b)\, d^*-(1-\rho\b)\, d_1 - \b\, d_2}.
\end{align*}

On $\cT_b'\cap\cA$, we assumed that $\move{f^*}$ closes $f_1$ and $f_2$. Therefore,
\begin{equation*}
  \WCchange_{\cT_b'\cap\cA} \leq \hoono{(1+\a\b)\, d^* - d_1 - \b\, d_2}. \tag{$\WCchange_{\move{f^*,\neg f_1,\neg f_2}}$}
\end{equation*}

Summarizing, we have
\begin{mybox}
\begin{align*}
\dSwapHa\leq  & d^* - (2-\b)\, d_1 + (1+\a\b-2\b)\, d_2
&  = d^* - 1.8\, d_1 + 1.2\, d_2\\
\dSwapHb\leq  & (1+\a\b)\, d^* - (1-\a\b)\, d_1 - 2\b\, d_2
&  = 1.6\, d^* - 0.4\, d_1 - 0.4\, d_2\\
\WCchange_{\cT_1\cap\cT_{3-b}'\cap\cA}(c) \leq &  (1+\a\b+\b)\, d^* -(2-2\b)\, d_1 + (1-2\b)\, d_2 \span\\
& &  = 1.8\, d^* - 1.6\, d_1 + 0.6\, d_2\\
\WCchange_{\cT_2\cap\cT_{3-b}'\cap\cA}(c) \leq & (1+\a\b+\rho\b)\, d^*-(1-\rho\b)\, d_1 - \b\, d_2
&  = (1.6 + 0.2\rho)\, d^* - (1 - 0.2\rho)\, d_1 - 0.2\, d_2\\
\WCchange_{\cT_b'\cap\cA} (c) \leq & (1+\a\b)\, d^* - d_1 - \b\, d_2
&  = 1.6\, d^* - d_1 - 0.2\, d_2
\end{align*}
\end{mybox}
We now combine these inequalities to get an upper bound for $\change_\cA(c)$.

When $\rho(f^*)\leq \nicefrac 23$, we have $\Pr[\cS_1] = \nicefrac 12, \Pr[\cT_1\cap\cT_{3-b}'] = \Pr[\cT_b'] = \nicefrac 14, \Pr[\cS_2] = \Pr[\cT_2\cap\cT_{3-b}'] = 0$. Therefore,
\begin{align*}
\change_\cA(c) & \leq \nicefrac 12\cdot \WCchange_{\cS_1\cap\cA}(c) + \nicefrac 14\cdot \WCchange_{\cT_1\cap\cT_{3-b}'\cap\cA}(c) + \nicefrac 14\cdot \WCchange_{\cT_b'\cap\cA}(c) + \epsd\\
& \leq 1.35\, d^* - 1.55\, d_1 + 0.7\, d_2 + \epsd\\
& \leq (2.05 + 0.7\rho)\, d^* - (1.55 - 0.7\rho)\, d_1 + \epsd\tag{$d_2\leq d^* + \rho(d^* + d_1)$}\\
& \leq (2.05 + 0.7\times\nicefrac 23)\, d^* - (1.55 - 0.7 \times \nicefrac 23)\, d_1 + \epsd\\
& \leq \hoono{2.51667\, d^* - 1.08333\, d_1} + \epsd.
\end{align*}

When $\nicefrac 23 < \rho \leq \nicefrac 34$, we have $\Pr[\cS_1] = \nicefrac 12, \Pr[\cS_2] = 0, \Pr[\cT_1\cap\cT_{3-b}'] = \Pr[\cT_2\cap\cT_{3-b}'] = \nicefrac 18, \Pr[\cT_b'] = \nicefrac 14$. Therefore,
\begin{align*}
\change_\cA(c) & \leq \nicefrac 12\cdot \WCchange_{\cS_1\cap\cA}(c) + \nicefrac 18\cdot \WCchange_{\cT_2\cap\cT_{3-b}'\cap\cA}(c)+ \nicefrac 18\cdot \WCchange_{\cT_2\cap\cT_{3-b}'\cap\cA}(c) + \nicefrac 14\cdot \WCchange_{\cT_b'\cap\cA}(c) + \epsd\\
& \leq (1.325 + 0.025\rho)\, d^* -(1.475 - 0.025\rho)\, d_1 + 0.6\, d_2 + \epsd\\
& \leq (1.925 + 0.625\rho)\, d^* - (1.475 - 0.625\rho)\, d_1 + \epsd\tag{$d_2\leq d^* + \rho(d^* + d_1)$}\\
& \leq (1.925 + 0.625 \times \nicefrac 34)\, d^* - (1.475 - 0.625\times\nicefrac 34)\, d_1 + \epsd\\
& = \hoono{2.39375\, d^* - 1.00625\, d_1} + \epsd.
\end{align*}

When $\rho > \nicefrac 34$, we have $\Pr[\cS_1] = \nicefrac 54 - \rho, \Pr[\cS_2] = \rho - \nicefrac 34, \Pr[\cT_1\cap\cT_{3-b}'] = \Pr[\cT_2\cap\cT_{3-b}'] = \nicefrac 18, \Pr[\cT_b'] = \nicefrac 14$. Therefore,
\begin{align*}
\change_\cA(c) & \leq (\nicefrac 54 - \rho)\cdot \WCchange_{\cS_1\cap\cA}(c)+ (\rho - \nicefrac 34)\cdot \WCchange_{\cS_2\cap\cA}(c) + \nicefrac 18\cdot \WCchange_{\cT_2\cap\cT_{3-b}'\cap\cA}(c)+ \nicefrac 18\cdot \WCchange_{\cT_2\cap\cT_{3-b}'\cap\cA}(c) \\& ~~~ + \nicefrac 14\cdot \WCchange_{\cT_b'\cap\cA}(c) + \epsd\\
& \leq (0.875 + 0.625\rho)\, d^* - (2.525 - 1.425\rho)\, d_1 + (1.8 - 1.6\rho)\, d_2 + \epsd\\
& \leq (2.675 + 0.825\rho - 1.6\rho^2)\, d^* - (2.525 - 3.225\rho + 1.6\rho^2)\, d_1 + \epsd\tag{$d_2\leq d^* + \rho(d^* + d_1)$}\\
& \leq (2.675 + 0.825\times \nicefrac 34 - 1.6\times(\nicefrac 34)^2)\, d^* - (2.525 - 3.225 + 1.6)\, d_1 + \epsd\\
& = \hoono{2.39375\, d^* - 0.9\, d_1} + \epsd.
\end{align*}
%
%
%
%
%
%


\section{Omitted Proofs}
\label{sec:omitted-proofs}

\subsection{Proof of \Cref{clm:light-simple}: There are enough local candidates}
\label{sec:proof-local-candidate}
\lightclients*

\begin{proof}
  Let $F_h$ be the set of \heavy local facilities,
  $F_p \sse F \setminus F_h$ be the set of local facilities pointed to by at
  least one optimal facility with no \heavy local neighbor, and $F_c$ be the
  remaining local facilities, which are exactly the local candidates. $|F_h| + |F_p| + |F_c| = $
  the number of local facilities, which in turn is at least the number of
  optimal facilities. There are at least $(\thd+2)|F_h|/2$
  many optimal facilities having a \heavy local neighbor because 1) a \heavy local facility is a neighbor of at least $\thd+2$ optimal facilities, and 2) each optimal facility has at most $2$ local neighbors. Finally, each local facility in $F_p$ is pointed to by an optimal facility with no \heavy local neighbor, so the total
  number of optimal facilities is at least $(\thd+2)|F_h|/2 + |F_p|$. In other
  words, $|F_c| \geq \frac{\thd}{2} |F_h|$.

\end{proof}

\subsection{Proof of \Cref{claim:balancedGroup}: Balancing Procedure}
\label{sec:proof-balancedGroup}
\balancedGroup*

\begin{proof}[Proof of \Cref{claim:balancedGroup}]
  Recall $|G| = |R|+ \extras$, where $\extras\geq \frac{16x^5\theta^2(\theta+1)}{\eps}$
  suffices.
  For each integer $s \in \{-x, \ldots, x\}$ let $D_s$ be the sets $S$
  with discrepancy $|S \cap G| - |S \cap R|$. Each set in $D_0$ can
  be output immediately. If for some $i, j$ we have $|D_i| \geq j/\eps$ and
  $|D_{-j}| \geq i/\eps$, and there is no edge in $H$, then we
  can choose some $j$ sets uniformly at random from $D_i$, 
  and $i$ sets from $D_{-j}$, and merge these together.

  However, since there are forbidden sets (a set $S_1$ and $S_2$ are forbidden
  if there is an edge between them in $H$), we need one more
  ingredient.  We claim that if some $D_i, D_{-j}$ have
  $\geq 8 x^2 \theta$ sets, then we can find $j$ sets from $D_i$ and
  $i$ sets from $D_{-j}$ that are not forbidden for each
  other. Indeed, pick a random collection of $j$ sets from $D_i$ and
  $i$ sets from $D_{-j}$. The probability that any one set has an edge
  to any of the other $i+j-1$ sets is
  $\leq \frac{(i+j-1)\theta}{8x^2 \theta} < \frac1{4x}$. Hence, a
  union bound over all the $i+j$ sets says that with probability at
  least a half, this collection does not have any edges of $H$ within
  it, and hence we can merge this collection together.
  
  However, above procedure does not ensure two sets are combined with
  probability at most $\eps$. To do so, if we find some pair
  $D_i, D_{-j}$ with $\geq \frac{8x^2\theta}{\eps}$ sets, then we can randomly
  partition each of $D_i$ and $D_{-j}$ into $1/\eps$ equal-sized subgroups
  with $8x^2\theta$ sets each. Now we can merge some $j$ sets from any
  subgroup from $D_i$ with some $i$ sets from a randomly chosen
  subgroup of $D_{-j}$ to form a set with equal number of greens and
  reds, exactly as above. Henceforth, we assume that for each 
  $D_i, D_{-j}$, at least one has fewer than $\frac{8x^2 \theta}{\eps}$ sets.

  Finally, since the greens outnumber the reds by $\extras$, we know there
  exists a value $j > 0$ such that
  $|D_j| \geq \extras/x = {16x^4\theta^2(\theta+1)}/\eps$.  Thus, we know each
  $D_s$ with $s<0$ has at most $\frac{8x^2\theta}{\eps}$ sets each. We randomly
  divide $D_j$ into $\frac{16x^3\theta^2}{\eps}$ parts of of size $x(\theta+1)$
  sets each.  Note any two sets $S_a$ and $S_b$ fall in the same part
  with probability at most
  $\frac{\eps}{16 x^3 \theta^2} \leq \eps$. From each part
  pick $x$ sets that have no edge in $H$ between themselves and call them a
  positive group; this can be done because the maximum degree of $H$
  is at most $\theta$. Each such positive group has at least $x$ extra
  green points. On the other hand, there are at most $x \cdot \frac{8x^2
  \theta}{\eps} = \frac{8x^3\theta}{\eps}$
  negative sets, i.e., in $\{D_{i}\}_{i<0}$.  Each negative set has
  edges to at most $\theta$ sets, so there are at most $\frac{8x^3\theta^2}{\eps}$
  sets with an edge to some negative set. Since there are
  $\frac{16x^3\theta^2}{\eps}$ positive groups,  there are at least
  $\frac{8x^3\theta^2}{\eps}$ positive groups with no edge to any negative set, so
  we can merge each negative set with a randomly-chosen such positive
  group.  This ensures that each new set has more green points than
  red, and two sets are combined with probability at most
  $\frac{\eps}{8x^3\theta^2} \leq \eps$. The newly-created sets
  have of size at most $O(x^2)$. Finally, each remaining set can form
  a group by itself, because they have more green points.
\end{proof}

\subsection{Proof of \Cref{clm:crude}: Crude Upper Bound of Potential Change}
\label{sec:defiant}

\Crude*
\begin{proof}
Since every local facility is closed by at most 3 swaps in $\cP$, there are at most 6 swaps in $\cP$ that closes any facility in $\{f_1,f_2\}$. Thus, it suffices to show that $\WCchange_{(P,Q)}(c) \leq O(d^* + d_1)$ for these 6 swaps $(P,Q)$. 

If $f^*$ has a heavy local neighbor $h$, the client can be served by $h$ at distance $\leq d^* + \nicefrac 32(d^* + d_1)$. We assume henceforth that $f^*$ has no heavy local neighbor, which means $\tau(f^*)$ is not \heavy and never closed as a local surrogate.

When $\cP$ is a simple swap set, the client can be served by either $f^*$ (at distance $\leq d^*$) or
 $\tau(f^*)$ (at distance $\leq d^* + \nicefrac 43(d^* + d_1)$).
When $\cP$ is a tree swap set, we show that one of the following facilities must be open after every swap in $\cP$:
\begin{align*}
f^*             ~ & \textup{at distance} ~ \leq d^*,\\
\tau(f^*)       ~ & \textup{at distance} ~ \leq d^* + \nicefrac 32 (d^* + d_1),\\
\pi(\tau(f^*))  ~ & \textup{at distance} ~ \leq d^* + 2\cdot \nicefrac 32 (d^* + d_1).
\end{align*}
It suffices to show that any swap closing $\tau(f^*)$ must open either $f^*$ or $\pi(\tau(f^*))$. 
If $\tau(f^*)$ is closed as an optimal surrogate, $\pi(\tau(f^*))$ must be open because edges on short cycles are not deleted in the edge deletion step (\Cref{cor:edge-deletion}). We thus focus on the swap closing the original copy of $\tau(f^*)$ henceforth.

Consider the 1-forest $G_1$ before edge deletion. The edges in $G_1$ from $f^*$ to $\tau(f^*)$ and from $\tau(f^*)$ to $\pi(\tau(f^*))$ cannot both be deleted in the edge deletion step, because we always choose $\thh$ as an even number and $G_1$ is bipartite (when self-loops are ignored). Therefore, either $f^*$ or $\pi(\tau(f^*))$ must be in the same swap with  $\tau(f^*)$, as desired.
\end{proof}

\subsection{Proof of \Cref{lm:typeA}: Combining Type \xA Inequalities}
\label{subsec:typeAaveraging}
\typeA*
We first prove \Cref{lm:typeA} assuming the following lemma, which we prove later.
\begin{lemma}
\label{lm:averaging}
For a close client of type \xA with $\rho(f^*) > \nicefrac 23$, we have
\begin{equation*}
\Pr[\cT_{21}] \leq \Pr[\cT_{11}] + O(\varepsilon).
\end{equation*}
\end{lemma}

\begin{proof}[Proof of \Cref{lm:typeA}]
Define $p_{ij} := \Pr[\cT_{ij}]$. Note that $\rho(f^*)>\nicefrac 23$ implies that $p_{11} + p_{12} = p_{21} + p_{22} = \nicefrac 14$. Define $p_\Delta := p_{11} - p_{21} = p_{22} - p_{12}$. \Cref{lm:averaging} implies $p_\Delta\geq -O(\varepsilon)$. Define $\WCchange_{\max}:= \max\{\WCchange_{\cT_{11}\cap\cA} + \WCchange_{\cT_{21}\cap\cA}, \WCchange_{\cT_{11}\cap\cA} + \WCchange_{\cT_{22}\cap\cA}, \WCchange_{\cT_{12}\cap\cA} + \WCchange_{\cT_{22}\cap\cA}\}$. \Cref{clm:crude} implies $\WCchange_{\max}\leq O(d^* + d_1)$.
\Cref{lm:typeA} is proved by the following chain of inequalities:
\begin{align*}
\change_{\cT\cap\cA}(c) & \leq \Pr[\cT_{11}\cap\cA]\WCchange_{\cT_{11}\cap\cA}(c) + \Pr[\cT_{12}\cap\cA]\WCchange_{\cT_{12}\cap\cA}(c)\\ 
& ~~~ + \Pr[\cT_{21}\cap\cA]\WCchange_{\cT_{21}\cap\cA}(c) + \Pr[\cT_{22}\cap\cA]\WCchange_{\cT_{22}\cap\cA}(c)\\
& \leq p_{11}\WCchange_{\cT_{11}\cap\cA}(c) + p_{12}\WCchange_{\cT_{12}\cap\cA}(c)\\ 
& ~~~ + p_{21}\WCchange_{\cT_{21}\cap\cA}(c) + p_{22}\WCchange_{\cT_{22}\cap\cA}(c)\\ & ~~~ + \epsd\tag{\Cref{clm:crude2} and $\WCchange_\cE(c)\geq \WClb$}\\
& = p_{21}\WCchange_{\cT_{11}\cap\cA}(c) + p_\Delta\WCchange_{\cT_{11}\cap\cA}(c) + p_{12}\WCchange_{\cT_{12}\cap\cA}(c)\\ 
& ~~~ + p_{21}\WCchange_{\cT_{21}\cap\cA}(c) + p_\Delta \WCchange_{\cT_{22}\cap\cA}(c) + p_{12}\WCchange_{\cT_{22}\cap\cA}(c)\\ & ~~~ + \epsd\\
& \leq p_{21}\WCchange_{\max} + p_\Delta(\WCchange_{\cT_{11}\cap\cA}(c) + \WCchange_{\cT_{22}\cap\cA}(c)) + p_{12}\WCchange_{\max} + \epsd\\
& \leq p_{21}\WCchange_{\max} + p_\Delta\WCchange_{\max} + p_{12}\WCchange_{\max} + \epsd \tag{$p_\Delta \geq -O(\varepsilon)$, $\WCchange_\cE(c)\geq \WClb$ and $\WCchange_{\max}\leq O(d^* + d_1)$}\\
& = (p_{21} + p_\Delta + p_{12})\WCchange_{\max} + \epsd\\
& = \nicefrac 14\cdot \WCchange_{\max} + \epsd. 
\end{align*}
\end{proof}

We now turn to proving \Cref{lm:averaging}. Before doing so, we need some deeper understandings of the edge deletion procedure, which we establish in \Cref{sec:survive}. The proof of \Cref{lm:averaging} is presented in \Cref{sec:averaging}.

\subsubsection{Probability of Surviving Edge Deletion}
\label{sec:survive}
Let $T$ be a 1-tree in the 1-forest $G_1$ before the edge deletion procedure. The edge deletion procedure splits $T$ into several connected components by deleting some edges from $T$. In this section, we prove upper and lower bounds on the probabilities that paths in $T$ remain connected after edge deletion.

Let $\ell>0$ denote the cycle length of $T$. Condition on the height threshold $\thh$ being fixed. We prove the following two lemmas:

\begin{lemma}[Upper bound]
\label{lm:survive-upper}
Suppose $p$ is a directed simple path in $T$ of length $s$. If $\ell\geq\thh$, then the probability that no edge in $p$ is deleted is at most $\max\{\frac{\thh - s}{\thh},0\}(1+\thh/\ell)$. If $\ell \leq \thh$, and we further assume that $p$ doesn't contain any cycle edge, then the probability is exactly $\max\{\frac{\thh - s}{\thh},0\}$.
\end{lemma}
\begin{proof}
If $s\geq \thh$, the lemma is trivial because any path after edge deletion has length at most $\thh - 1$. We assume $s<\thh$ henceforth.

Suppose vertices on $p$ are $v_0\leftarrow v_1\leftarrow\cdots\leftarrow v_s$. We first consider the case where $\ell\geq \thh$. We prove that as long as the (unique) simple path $p^*$ from $v_0$ to $r$ has length equal to $-1,-2,\cdots,-s$ modulo $\thh$, some edge on path $p$ is deleted. Indeed, suppose $p^*$ has length $-i$ modulo $\thh$. If $p^*$ doesn't contain any vertex in $\{v_1,\cdots,v_s\}$, then the edge out of $v_i$ is deleted by \Cref{claim:edge-deletion}. Otherwise, $r$ must be one of $v_1,v_2,\cdots,v_s$, in which case the edge out of $r$ is deleted.

Suppose $\ell = u\thh + w$ for $u,w\in\mathbb Z$ where $0\leq w < \thh$. There are at most $(\thh - s)(u+1)$ choices of $r$ such that $p^*$ has length not in $\{-1,\cdots,-s\}$ modulo $\thh$. Therefore, when $\ell\geq\thh$, the probability that no edge in $p$ is deleted is at most $(\thh - s)(u + 1)/\ell = \frac{\thh - s}{\thh}\cdot (\frac{u\thh}{\ell} + \frac{\thh}{\ell})\leq \frac{\thh - s}{\thh}\cdot (1 + \thh/\ell)$.

When $\ell\leq t_h$ and $p$ doesn't contain a cycle edge, an edge on the path $p$ is deleted if and only if $p^*$ has length $-1,-2,\cdots,-s$ modulo $\thh$ by \Cref{claim:edge-deletion}. Since the cycle length is exactly $\thh$ after dummy vertices are inserted on it, the probability that no edge on $p$ is deleted is exactly $\frac{\thh - s}{\thh}$.
\end{proof}

\begin{lemma}[Lower bound]
\label{lm:survive-lower}
Let $v_1,v_2,v^*$ be vertices in $T$ and $p_1,p_2$ be directed simple paths in $T$ from $v_1$ and $v_2$ to $v^*$, respectively. Suppose both $p_1$ and $p_2$ have lengths no greater than $s$. If $\ell\geq \thh$, then the probability that no edge on either path $p_1,p_2$ is deleted is at least $\max\{\frac{\thh - s}{\thh},0\}(1-2\thh/\ell)$. If $\ell \leq \thh$, and we further assume that $v_1$ is on the cycle, then the probability is at least $\max\{\frac{\thh - s}{\thh},0\}$.
\end{lemma}
\begin{proof}
Again, the lemma is trivial if $s\geq \thh$. Assume $s<\thh$ henceforth.

Let us first consider the case where $\ell\geq\thh$. Consider the vertices on the cycle that are different from $v^*$ but have paths to $v^*$ with length at most $s$. There are at most $s$ such vertices, and they form a contiguous part of the cycle. If $r$ is not among these vertices, then the simple path $p^*$ from $v^*$ to $r$ contains no vertex on $p_1$ or $p_2$ except $v^*$ itself. If we further assume that $p^*$ has length not in $-1,-2,\cdots,-s$ modulo $\thh$, then by \Cref{claim:edge-deletion} no edge on either path $p_1,p_2$ is deleted. Therefore, assuming $\ell - s = u\thh + w$ for $u,w\in\mathbb Z$ where $0\leq w < \thh$, the probability that no edge on either path is deleted is at least $u(\thh - s)/\ell = \frac{\thh - s}\thh\cdot \frac{u\thh}{\ell} = \frac{\thh - s}\thh\cdot (1 - \frac{s + w}\ell)\geq \frac{\thh - s}\thh\cdot (1 - 2\thh/\ell)$.

When $\ell \leq\thh$ and $v_1$ is on the cycle, every edge on $p_1$ must be on the cycle. Since no edge on the cycle is deleted by our convention, the probability that no edge on either path is deleted is lower bounded by the probability that no edge on the shortest path $p'$ from $v_2$ to the cycle is deleted. $p'$ is a part of $p_2$, so $p'$ has length at most $s$. By the second part of the previous lemma, the probability that no edge on $p'$ is deleted is at least $\max\{\frac{\thh - s}{\thh},0\}$.
\end{proof}

\subsubsection{Proof of \Cref{lm:averaging}}
\label{sec:averaging}

We are now ready to prove \Cref{lm:averaging}. Define $\cD'$ as the union of the defiant event $\cD$ (\Cref{def:amenable-defiant-swaps}) and the following events:
\begin{OneLiners}
\item[(i)] $\cP$ is a tree swap set, and, before edge deletion, the cycle in the 1-tree containing the original copy of $f^*$ has length $\ell$ in the range $(\thh, \lceil \nicefrac 1\varepsilon\rceil\cdot\thh)$;
\item[(ii)] $\cP$ is a tree swap set, and two connected components each containing a facility in $\{f_1,f_2\}$ are combined in the balancing procedure.
\end{OneLiners}
Event (i) happens with probability $O(\varepsilon)$ because our height threshold $\thh$ is chosen uniformly at random from $2\lceil \nicefrac 1 \varepsilon \rceil, 2\lceil \nicefrac 1 \varepsilon \rceil^2, \cdots, 2\lceil \nicefrac 1 \varepsilon \rceil^{\lceil \nicefrac 1 \varepsilon \rceil}$. Event (ii) happens with probability $O(\varepsilon)$ as well due to \Cref{clm:tree-degree-reduction,claim:balancedGroup}. By a union bound with \Cref{clm:crude2}, we have
\begin{claim}
\label{claim:extended-defiant}
The event $\cD'$ happens with probability $O(\varepsilon)$.
\end{claim}
\begin{proof}[Proof of \Cref{lm:averaging}]
If either $f_1$ or $f_2$ is heavy, then $\cT_{21}$ never happens. Indeed, $\cT_{21}$ assumes the existence of a swap closing both $f_1$ and $f_2$, but heavy local facilities are never closed. Hence, we assume neither $f_1$ nor $f_2$ is heavy.

By \Cref{claim:extended-defiant} and the union bound, it suffices to prove $\Pr[\cT_{21}\backslash\cD']\leq (1 + O(\varepsilon))\Pr[\cT_{11}\cup \cD']$. By law of total probability, it suffices to prove 
\begin{equation}
\label{eq:conditional}
\Pr[\cT_{21}\backslash\cD'|\cE_i]\leq (1 + O(\varepsilon))\Pr[\cT_{11}\cup \cD'|\cE_i]
\end{equation}
for a partition $\cE_1,\cE_2,\cdots,\cE_t$ of the entire probability space.

If $\cE_i = \cS$, then both sides of (\ref{eq:conditional}) become zero. Let us condition on the tree event $\cT$ henceforth. Conditioned on $\cT$, the probabilities of  $\tau(f^*) = \eta_1$ and $\tau(f^*) = \eta_2$ are both $\nicefrac 12$ since $\rho(f^*) > \nicefrac 23$. Note that the set of heavy local/optimal facilities doesn't depend on the random function $\tau$. Therefore, if we condition on the $\tau$'s of all optimal facilities except $f^*$, the out-edges of the original copies of all facilities in $G_1$ except $f^*$ are determined, where $G_1$ is the 1-forest after degree reduction but before edge deletion. Let $G_1^*$ be $G_1$ with the out-edge of the original copy of $f^*$ removed. If we ignore the identity of the local and optimal surrogates, everything else in $G_1^*$ is determined.
Moreover, the conditioning we did is independent of $\tau(f^*)$, so the conditional probabilities of $\tau(f^*) = \eta_1$ and $\tau(f^*) = \eta_2$ are both still $\nicefrac 12$.

Note that $f^*$ may be a heavy optimal facility, in which case $f^*$ has new copies in $G_1$. We use $f^*$ to refer to only the original copy. $\tau(f^*)$ may also be a heavy local facility when $\tau(f^*) = \eta_2$ (note that we assumed $\eta_1 = f_1$ is not heavy), in which case $f^*$ points to itself in $G_1$. If either $f_1$ or $f_2$ is chosen as a surrogate, then $\cT_{21}\backslash \cD'$ cannot happen because $\cD'$ happens. We thus assume $f_1$ and $f_2$ only appear as their original copies in $G_1$. Since $f^*$ is the only vertex in $G_1^*$ that doesn't have an out-edge, $f^*$ is the root of a tree, and all other connected components of $G_1^*$ are 1-trees. 

We divide our proof into five cases depending on the structure of $G_1^*$:
\begin{OneLiners}
\item[1.] $f^*,f_1,f_2$ are all in the different connected components;
\item[2.] $f^*,f_1$ are in the same tree, different from $f_2$;
\item[3.] $f^*,f_2$ are in the same tree, different from $f_1$;
\item[4.] $f_1,f_2$ are in the same 1-tree (denoted by $T$), different from $f^*$;
\item[5.] all three are in the same tree (denoted by $T^*$).
\end{OneLiners}

Let $\cE_1$ denote the event that $f_1$ and $f_2$ are in the same connected component in $G_2$, where $G_2$ is the graph after the edge deletion procedure. Since $\cD'$ includes the case where the edge from $f^*$ to $\eta_1 = f_1$ is deleted in the edge deletion step, we have $\cE_1\cap\cT_1\subseteq \cT_{11}\cup\cD'$. Let $\cE_0$ denote the event that $f_1$ and $f_2$ are in the same connected component in $G_2$ but different from $f^*$. Since subtracting $\cD'$ rules out the possibility of $f_1$ and $f_2$ being combined in the balancing step, we have $\cT_{21}\backslash\cD'\subseteq \cE_0\subseteq \cE_1$.

In case 1, $\cT_{21}\backslash\cD'$ never happens because $\cE_1$ never happens. Indeed, $f_1,f_2$ must be in different connected components in $G_1$ and thus must be in different connected components in $G_2$.

In cases 2\&3, $\cT_{21}\backslash\cD'$ never happens either because $\cE_0$ never happens. Indeed, the only way $f_1$ can connect to $f_2$ (by an undirected path in $G_1$) is through $f^*$, and in the edge deletion procedure, there is no way to put $f_1$, $f_2$ in the same connected component of $G_2$ without also putting $f^*$ in it.

In case 4, $f^*$ is not on the cycle part of $T$, so the height threshold $\thh$ and the choice of $r\in T$ in the edge deletion step are both independent of $\tau(f^*)$. Once conditioned on $\thh,r$, whether or not $f_1$ and $f_2$ are in the same connected component in $G_2$ is determined. We assume that $f_1$ and $f_2$ are in the same connected component of $G_2$ because otherwise $\cT_{21}\backslash\cD'$ never happens. If $\tau(f^*) = \eta_1 (= f_1)$, then we know $\cT_{11}\cup\cD'$ must happen, because $\cE_1\cap\cT_1$ happens. Moreover, $\cT_{21}\backslash \cD'$ happens only when $\tau(f^*) = \eta_2$ simply because $\cT_{21}\subseteq \cT_2$. Therefore, if we let $\cE$ be the event summarizing all the conditioning we did so far, we have
\begin{align*}
\Pr[\cT_{11}\cup\cD'|\cE]= & \Pr[\tau(f^*) = \eta_1|\cE] = \nicefrac 12,\\
\Pr[\cT_{21}\backslash\cD'|\cE]\leq & \Pr[\tau(f^*) = \eta_2|\cE] = \nicefrac 12,
\end{align*}
and thus (\ref{eq:conditional}) holds for $\cE_i = \cE$.

Case 5 is a little tricky since the cycle structure of $T$, the 1-tree in $G_1$ containing all of $f^*,f_1,f_2$, may depend on where $f^*$ points to. Condition on the height threshold $\thh$ being fixed, and let $\cE$ be the event summarizing all the conditioning we did so far. Let $f_a$ be the least common ancestor of $f_1$ and $f_2$ in $T^*$, and let $s$ denote the path length from $f_i$ to $f_a$ maximized over $i = 1,2$. 

Conditioned on $\tau(f^*) = \eta_1$, or equivalently $\cT_1$, the probability of $\cT_{11}\cup\cD'$ is 1 if the cycle length $\ell$ of $T$ is in the range $(\thh, \lceil \nicefrac 1\varepsilon \rceil\cdot \thh)$, and if $\ell$ is not in the range, the conditional probability of $\cT_{11}\cup\cD'$ is at least the conditional probability of $\cE_1$, which is at least $\max\{\frac{\thh - s}{\thh},0\}(1- O(\varepsilon))$ by \Cref{lm:survive-lower} (Observe that $f_1 = \eta_1$ is on the cycle of $T$ because $f^*$ points to it on event $\cT_1$). Therefore,
\begin{align}
\label{eq:averaging1}
 \Pr[\cT_{11}\cup\cD'|\cE]& \geq \Pr[\tau(f^*) = \eta_1|\cE]\cdot \max\Big\{\frac{\thh-s}{\thh}, 0\Big\}(1 - O(\varepsilon))\notag\\
 & = \nicefrac 12\cdot\max\Big\{\frac{\thh-s}{\thh}, 0\Big\}(1 - O(\varepsilon)).
\end{align}
On the other hand, $\cT_{21}\backslash\cD'$ happens only when $\tau(f^*) = \eta_2$. Condition on $\tau(f^*) = \eta_2$. If the cycle length $\ell$ is in the range $(\thh, \lceil \nicefrac 1\varepsilon \rceil \cdot \thh)$, then $\cT_{21}\backslash\cD'$ never happens. If $\ell\leq \thh$, and $f_a$ is on the cycle, then $\cT_{21}\backslash\cD'$ never happens either because $\cE_0$ never happens. Indeed, the only possible undirected path in $T$ connecting $f_1$ with $f_2$ without passing through $f^*$ intersects the cycle, so $f_1,f_2$ have to connect to the cycle after edge deletion to make $\cE_0$ happen, but the cycle contains $f^*$ and remains connected after edge deletion (because $\ell\leq\thh$). Therefore, we assume either $\ell\geq \lceil\nicefrac 1\varepsilon\rceil\cdot\thh$, or $\ell\leq\thh$ and $f_a$ is not on the cycle. In this case, the conditional probability of $\cT_{21}\backslash\cD'$ is at most the conditional probability of $\cE_0$, which is at most $\max\Big\{\frac{\thh-s}{\thh}, 0\Big\}(1 + O(\varepsilon))$ by \Cref{lm:survive-upper}. Therefore,
\begin{align}
\label{eq:averaging2}
\Pr[\cT_{21}\backslash\cD'|\cE] & \leq \Pr[\tau(f^*) = \eta_2|\cE]\cdot \max\Big\{\frac{\thh-s}{\thh}, 0\Big\}(1 + O(\varepsilon))\notag\\
& = \nicefrac 12\cdot \max\Big\{\frac{\thh-s}{\thh}, 0\Big\}(1 + O(\varepsilon)).
\end{align}
Combining (\ref{eq:averaging1}) and (\ref{eq:averaging2}), we know (\ref{eq:conditional}) holds for $\cE_i = \cE$.
\end{proof}

\subsection{Proof of \Cref{claim:type3Swapcases}: Subtypes within Type \xC}
\label{sec:proof-typeC}
\typeThreeSwapcases*
\begin{proof}
Recall that $g^*$ is $\pi(f_1)$ and $\cS'_b$ is the event that $\cP$ is a simple swap and
$g^*$ points to $\eta_b(g^*)$. Similarly $\cT'_b$ is the event that $\cP$ is a tree swap and 
$g^*$ points to $\eta_b(g^*)$.

If either $f_1$ or $f_2$ is \heavy, then condition (a) or (b) holds. We assume neither $f_1$ nor $f_2$ is \heavy henceforth. In other words, the swaps $\move{\neg f_1}$ and $\move{\neg f_2}$ both exist.

Let $g$ be the closest local facility to $g^*$ that is different from $f_1$ and $f_2$. 
Intuitively, we show that either a client is close to $g$ or there is a tree that contains all
$f_1$, $f_2$, and $f^*$.

If $d(g,g^*) \leq d(f_1,g^*)$, then we have $d(c,g) \leq d_1 + d(f_1,g^*) + d(g,g^*) \leq d_1 + 2d(f_1,g^*)\leq 3d_1 + 2d^*$.
Furthermore, 
when we generate tree swaps, $f_1$ points to $g^* = \pi(f_1)$ in the 1-forest $G_1$ after degree reduction.
If $f_1$ points to a new copy of $g^*$, we know that $f_1$ is not among the $\thd$ closest local facilities to $g^*$ in $\pi^{-1}(g^*)$. 
Therefore, we know $d(g,g^*)\leq d(f_1,g^*)$.
Note that $g$ and $f_1$ are not closed in the same simple swap by implication (Siii') of amenability, so condition (c) holds in this case.

We can now assume that $f_1$ points to the original copy of $g^*$ and $d(g,g^*) > d(f_1,g^*)$. If $g^* = f^*$, 
we know condition (e) holds, because both edges $f_1\rightarrow f^*, f^*\rightarrow f_2$ remain after the edge deletion step by amenability. We assume $g^*\neq f^*$ henceforth.

If $\rho(g^*) \leq \nicefrac 23$, we know $\tau(g^*) = \eta_1(g^*)$ deterministically. Moreover, $d(g,g^*) > d(f_1,g^*)$ implies that $\tau(g^*)$ is either $f_1$ or $f_2$. If $\eta_1(g^*) = f_1$, then $\move{\neg f_1}$ must open $g^*$ by implication (ii') of amenability, so condition (c) holds in this case since $d(c,g^*)\leq d_1 + d(f_1,g^*) \leq 2d_1 + d^*$. Otherwise, $\eta_1(g^*) = f_2$, and then condition (e) holds, because the edges $f_1\rightarrow g^*, g^*\rightarrow f_2, f^*\rightarrow f_2$ all survive edge deletion by amenability, so $f_1,f_2,f^*$ must all be in the same swap.

It remains to consider the case where $\rho(g^*) > \nicefrac 23$. If $f_2 = \eta_b(g^*)\in\{\eta_1(g^*),\eta_2(g^*)\}$, then condition (f) holds because the edges $f_1\rightarrow g^*,g^*\rightarrow f_2,f^*\rightarrow\tau(f^*)\in\{f_1,f_2\}$ all survive edge deletion on $\cT_b'\cap\cA$ (see the left graph in \Cref{fig:typeCsubcase}). Otherwise, $f_2\notin\{\eta_1(g^*),\eta_2(g^*)\}$, and in this case we know $\eta_1(g^*) = f_1$ and $\eta_2(g^*) = g$ because $d(g,g^*) > d(f_1,g^*)$. We show that condition (c) or (d) holds, depending on whether $\rho(g^*) \leq \nicefrac 34$. Indeed, on $\cS_1'\cap\cA$, we know $\move{\neg f_1}$ opens $g^*$ at distance $\leq 2d_1 + d^*$ by implication (ii') of amenability, and on $\cS_2'\cap\cA$, we know either $g^*$ or $g$ is open after swap $\move{\neg f_1}$, again by implication (ii') of amenability, and $d(c,g)\leq 2d_1 + d^* + \nicefrac 1 {\rho(g^*)} \cdot (d^* + d_1)$ (see the right graph in \Cref{fig:typeCsubcase}).
\end{proof}

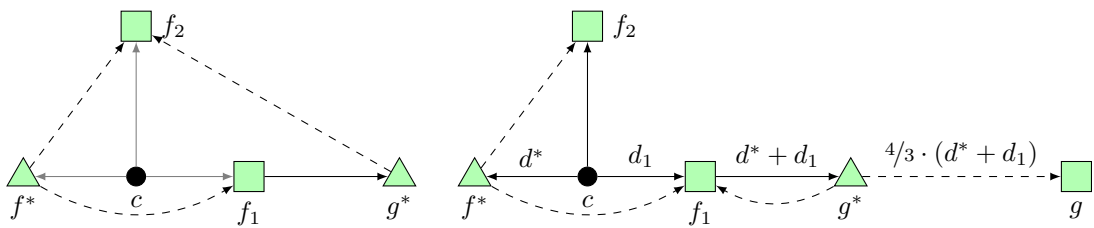
\begin{figure}[H]
\center
\begin{tikzpicture}[scale=1,
		opt/.style={shape=regular polygon,regular polygon sides = 3,draw=black,minimum size=0.5cm,inner sep = 0pt},
		local/.style={shape=rectangle, draw=black,minimum size= 0.4cm},
		client/.style={shape=circle, draw=black,minimum size=0.25cm,inner sep = 0pt,fill=black}]
		
	\pgfmathtruncatemacro{\ZERO}{6}
	\def \GREEN {green!30}
	
    \node[opt,label={[]south:{$f^*$}},fill = \GREEN]     (f*\ZERO) at (\ZERO-3,0) {};
    \node[client,label=south:$c$]   (c\ZERO)  at (\ZERO-1.5,0) {};
    \node[local,label=south:{$f_1$},fill = \GREEN] (f1\ZERO) at (\ZERO+0,0) {};
    \node[opt,label=south:$g^*$,fill = \GREEN]     (g*\ZERO) at (\ZERO+2,0) {};
    \node[local,label=south:{$g$},fill = \GREEN]   (g\ZERO)  at (\ZERO+5,0) {};
    	
    \node[local,label={[]east:{$f_2$}},fill = \GREEN]       (f2\ZERO) at (\ZERO-1.5,2) {} ;

	\path [->,>=latex] (c\ZERO) edge[above] node {$d^*$}(f*\ZERO);
	\path [->,>=latex] (c\ZERO) edge[above] node {$d_1$}(f1\ZERO);
	\path [->,>=latex] (c\ZERO) edge[right] node {}(f2\ZERO);
	\path [->,>=latex] (f1\ZERO) edge[above] node  {$ d^*+d_1$} (g*\ZERO);
	\path [->,>=latex] (g*\ZERO) edge[above, in = -30, out = -150, dashed] node  {} (f1\ZERO);
	\path [->,>=latex] (g*\ZERO) edge[above, dashed] node  {$ \nicefrac 43\cdot (d^*+d_1)$} (g\ZERO);
	\path [->,>=latex] (f*\ZERO) edge[dashed] (f2\ZERO);
	\path [->,>=latex,bend right] (f*\ZERO) edge[dashed] (f1\ZERO);


	\pgfmathtruncatemacro{\ZERO}{0}

    \node[opt,label=south:$f^*$,fill=\GREEN]     (f*\ZERO) at (\ZERO-3,0) {};
    \node[client,label=south:$c$]   (c\ZERO)  at (\ZERO-1.5,0) {};
    \node[local,label=south:{$f_1$},fill=\GREEN] (f1\ZERO) at (\ZERO+0,0) {};
    \node[opt,label=south:$g^*$,fill=\GREEN]     (g*\ZERO) at (\ZERO+2,0) {};
    \node[local,label=east:{$f_2$},fill=\GREEN]       (f2\ZERO) at (\ZERO-1.5,2) {} ;

	\path [->,>=latex,gray] (c\ZERO) edge[above] node {}(f*\ZERO);
	\path [->,>=latex,gray] (c\ZERO) edge[above] node {}(f1\ZERO);
	\path [->,>=latex,gray] (c\ZERO) edge[right] node {}(f2\ZERO);
	\path [->,>=latex] (f1\ZERO) edge[above] node  {} (g*\ZERO);
	\path [->,>=latex] (g*\ZERO) edge[above,dashed] node  {} (f2\ZERO);
	\path [->,>=latex,dashed] (f*\ZERO) edge (f2\ZERO);
	\path [->,>=latex,bend right,dashed] (f*\ZERO) edge (f1\ZERO);	
	
\end{tikzpicture}
\caption{In the figure, dashed edges represent the random function $\tau$. In the left graph, whenever $g^*$ points to $f_2$, $f_1,f_2,f^*$ are all in the same swap, so condition (f) holds. In the right graph, condition (d) holds.}
\label{fig:typeCsubcase}
\end{figure}


\subsection*{Acknowledgments}

We thank Amit Kumar, Ola Svensson, and Justin Ward for fruitful
discussions. Special thanks to Guru Guruganesh, with whom we obtained
some early results on this problem. LH is supported by NSF Award IIS-1908774 and a
VMware fellowship. This work was conducted in part while LH was an
undergraduate at Tsinghua University visiting CMU and TTI-Chicago. 

Supported in part by NSF awards CCF-1907820, CCF1955785, and CCF-2006953.

This work was [partially] funded by the grant ANR-19-CE48-0016 from the French National Research Agency (ANR).
{\footnotesize
\bibliographystyle{alpha}
\bibliography{hoon,references}
}

\end{document}